\documentclass[11pt]{article}

\usepackage{amsmath,amsfonts,amsthm}
\usepackage{xcolor}
\usepackage{amssymb}
\usepackage[numbers,sort]{natbib}
\usepackage[backref,pageanchor=true,plainpages=false, pdfpagelabels, bookmarks,bookmarksnumbered,
]{hyperref}
\usepackage{mathtools}
\usepackage{wrapfig}
\usepackage{cleveref}
\usepackage{nicefrac}
\usepackage{thmtools,thm-restate} 

\usepackage{caption}
\usepackage{graphicx,subcaption}
\usepackage[ruled,vlined]{algorithm2e}

\usepackage[margin=1in]{geometry}
\usepackage{fullpage}

\newtheorem{theorem}{Theorem}
\newtheorem{lemma}{Lemma}
\newtheorem{claim}{Claim}
\newtheorem{corollary}{Corollary}
\newtheorem{definition}{Definition}
\newtheorem{observation}{Observation}
\newtheorem{fact}{Fact}

\def\E{\mathbb{E}}
\def\Re{\mathbb{R}}
\def\N{\mathbb{N}}
\def\RePart{\mathrm{Re}}
\def\MCRATIO{$0.861$\xspace}
\def \MtwoRatio{$0.875$}

\def\MtwoRatioSD{$0.921$\xspace}
\def\MDCRatio{$0.79$\xspace}
\def\MDCRatioSD{$0.81$\xspace}
\def\slowdown{$1.61$\xspace}

\DeclarePairedDelimiter{\abs}{\lvert}{\rvert}
\DeclarePairedDelimiter{\norm}{\|}{\|}
\DeclarePairedDelimiter{\braket}{\langle}{\rangle}
\DeclarePairedDelimiter{\set}{ \{ }{ \} }
\newcommand{\bkt}[2]{\braket{#1, #2}}

\renewcommand{\L}{\ensuremath{\mathcal{L}}}

\newcommand{\eps}{\varepsilon}

\newcommand{\MC}{\textsc{Max-Cut}\xspace}
\newcommand{\MdC}{\textsc{Max-DiCut}\xspace}
\newcommand{\MtwoSAT}{\textsc{Max-2SAT}\xspace}
\newcommand{\MS}{\textsc{Max-SAT}\xspace}
\newcommand{\MB}{\textsc{Max-Bisection}\xspace}
\newcommand{\cardMC}{\textsc{Max-Cut-SC}\xspace}

\newcommand{\Ito}{It\^{o}\xspace}
\newcommand{\Int}{\textcsc{Int}}
\newcommand{\bd}{\partial}

\newcommand{\OPT}{\textsc{OPT}}
\newcommand{\SDP}{\textsc{SDP}}
\def\bw {{\bf w}}
\def\bX {{\bf X}}
\def\bY {{\bf Y}}
\def\bW {{\bf W}}
\def\bB {{\bf B}}

\def\bb {{\bf b}}

\def\bv {{\bf v}}
\def\by {{\bf y}}

\def\bO {{\bf O}}
\def\bx {{\bf x}}

\def\FF {{\cal F}}

\def \bd {{\partial}}
\def \Int {{\operatorname{Int} }}

\newcommand{\rhombus}{\mathbb{S}}

\def \MtwoRatio{$0.8749$\xspace}
\def \MtwoRatioText{0.8749\xspace}

\newcommand{\mcbasicratio}{\MCRATIO}
\newcommand{\mcslowratio}{0.874\xspace}
\newcommand{\mcexpratio}{0.878\xspace}
\newcommand{\GWratio}{0.878\xspace}
\newcommand{\mcallratio}{0.843\xspace}
\newcommand{\mcslowallratio}{0.858\xspace}

\newcommand{\eqdef}{:=}

\newcommand{\NOCOMMENT}{}

\ifdefined\NOCOMMENT
    \newcommand{\ggnote}[1]{}
    \newcommand{\msnote}[1] {}
    \newcommand{\nbnote}[1] {}
    \newcommand{\snote}[1] {}
    \newcommand{\rsnote}[1] {}
    \newcommand{\saznote}[1] {}

\else
\newcounter{note}[section]
\renewcommand{\thenote}{\thesection.\arabic{note}}
\newcommand{\ggnote}[1]{\refstepcounter{note}\textcolor{blue}{$\ll${\bf Guru~\thenote:} {\sf #1}$\gg$\marginpar{\tiny\bf GG~\thenote}}}
\newcommand{\nbnote}[1]{\refstepcounter{note}\textcolor{red}{$\ll${\bf Nikhil~\thenote:} {\sf #1}$\gg$\marginpar{\tiny\bf NB~\thenote}}}
\newcommand{\snote}[1]{\refstepcounter{note}\textcolor{brown}{$\ll${\bf Sasho~\thenote:} {\sf #1}$\gg$\marginpar{\tiny\bf SN~\thenote}}}
\newcommand{\rsnote}[1]{\refstepcounter{note}\textcolor{magenta}{$\ll${\bf Roy~\thenote:} {\sf #1}$\gg$\marginpar{\tiny\bf RS~\thenote}}}
\newcommand{\msnote}[1]{\refstepcounter{note}\textcolor{cyan}{$\ll${\bf Mohit~\thenote:} {\sf #1}$\gg$\marginpar{\tiny\bf MS~\thenote}}}
\newcommand{\saznote}[1]{\refstepcounter{note}\textcolor{orange}{$\ll${\bf Sepehr~\thenote:} {\sf #1}$\gg$\marginpar{\tiny\bf SA~\thenote}}}
\fi

\newcommand{\cut}[1]{}

\newmuskip\pFqmuskip

\newcommand*\pFq[6][8]{%
  \begingroup 
  \pFqmuskip=#1mu\relax
  \mathchardef\normalcomma=\mathcode`,
  \mathcode`\,=\string"8000
  \begingroup\lccode`\~=`\,
  \lowercase{\endgroup\let~}\pFqcomma
  {}_{#2}F_{#3}{\left[\genfrac..{0pt}{}{#4}{#5};#6\right]}%
  \endgroup
}
\newcommand{\pFqcomma}{{\normalcomma}\mskip\pFqmuskip}

\date{}

\begin{document}
\title{Sticky Brownian Rounding and its Applications to Constraint Satisfaction Problems}

\author{
 Sepehr Abbasi-Zadeh \thanks{University of Toronto. E-mail: \texttt{sabbasizadeh@gmail.com}.}
 \and
 Nikhil Bansal\thanks{TU Eindhoven, and Centrum Wiskunde \& Informatica. E-mail: \texttt{bansal@gmail.com}.}
 \and
 Guru Guruganesh \thanks{Google Research. Email: \texttt{gurug@google.com}.}
 \and
 Aleksandar Nikolov \thanks{University of Toronto. Email:\texttt{anikolov@cs.toronto.edu}.}
 \and
 Roy Schwartz \thanks{Technion. Email:\texttt{schwartz@cs.technion.ac.il}.}
 \and
 Mohit Singh \thanks{Georgia Institute of Technology. Email:\texttt{mohitsinghr@gmail.com}.}
}

\maketitle
\begin{abstract}
    Semidefinite programming is a powerful tool in the design and analysis of approximation algorithms for combinatorial optimization problems.  In particular, the random hyperplane rounding method of \citet{GW95} has been extensively studied for more than two decades, resulting in various extensions to the original technique and beautiful algorithms for a wide range of applications. Despite the fact that this approach yields tight approximation guarantees for some problems, {\em e.g.}, \MC, for many others, {\em e.g.}, \MS and \MdC, the tight approximation ratio is still unknown. One of the main reasons for this is the fact that very few techniques for rounding semi-definite relaxations are known.

In this work, we present a new general and simple method for rounding semi-definite programs, based on Brownian motion.  Our  approach is inspired by recent results in algorithmic discrepancy theory.  We develop and present tools for analyzing our new rounding algorithms, utilizing mathematical machinery from the theory of Brownian motion, complex analysis, and partial differential equations.  Focusing on constraint satisfaction problems, we apply our method to several classical problems, including \MC, \MtwoSAT, and \MdC, and derive new algorithms that are competitive with the best known results.  To illustrate the versatility and general applicability of our approach, we give new approximation algorithms for the \MC problem with side constraints that crucially utilizes measure concentration results for the Sticky Brownian Motion, a feature missing from hyperplane rounding and its generalizations.

\end{abstract}

\thispagestyle{empty}
\newpage
\setcounter{page}{1}

\section{Introduction}

Semi-definite programming (SDP) is one of the most powerful tools in the design of approximation algorithms for combinatorial optimization problems.
Semi-definite programs can be viewed as relaxed quadratic programs whose variables are allowed to be vectors instead of scalars and scalar multiplication is replaced by inner products between the vectors.
The prominent approach when designing SDP based approximation algorithms is {\em rounding}: $(1)$ an SDP relaxation is formulated for the given problem, $(2)$ the SDP relaxation is solved, and lastly $(3)$ the fractional solution for the SDP relaxation is transformed into a feasible integral solution to the original problem, hence the term {\em rounding}.

In their seminal work, Goemans and Williamson \cite{GW95} presented an elegant and remarkably simple rounding method for SDPs: a uniformly random hyperplane (through the origin) is chosen, and then each variable, which is a vector, is assigned to the side of the hyperplane it belongs to.
This (binary) assignment is used to round the vectors and output an integral solution.
For example, when considering \MC, each side of the hyperplane corresponds to a different side of the cut.
Using the random hyperplane rounding, \cite{GW95} gave the first non-trivial approximation guarantees for fundamental problems such as \MC, \MtwoSAT, and \MdC.
Perhaps the most celebrated result of \cite{GW95} is the $\GWratio$ approximation for \MC, which is known to be tight \cite{KKMO,MOO} assuming Khot's Unique Games Conjecture \cite{Khot02-UG}.
Since then, the random hyperplane method has inspired, for more than two decades now, a large body of research, both in approximation algorithms and in hardness of approximation.
In particular, many extensions and generalizations of the random hyperplane rounding method have been proposed and applied to a wide range of applications
, {\em e.g.}, \MdC and \MtwoSAT \cite{feige1995approximating,LLZ}, \MS \cite{asano2002improved,avidor2005improved}, \textsc{Max-Bisection} \cite{RaghavendraT12,AustrinBG16},  \textsc{Max-Agreement} in correlation clustering \cite{CharikarGW05}, the \textsc{Cut-Norm} of a matrix~\cite{AlonNaor}.

Despite this success and the significant work on variants and extensions of the random hyperplane method, the best possible approximation ratios for many fundamental problems still remain elusive.
Several such examples include \MS, \textsc{Max-Bisection}, \textsc{Max-2CSP}, and \MdC.
Perhaps the most crucial reason for the above failure is the fact that besides the random hyperplane method and its variants, very few methods for rounding SDPs are known.

A sequence of papers by \citet{Austrin10,Raghavendra08,RaghavendraS09} has shown that SDP rounding algorithms that are based on the random hyperplane method and its extensions nearly match the Unique Games hardness of any \textsc{Max-CSP}, as well as the integrality gap of a natural family of SDP relaxations. 
However,  the universal rounding proposed by Raghavendra and Steurer is impractical, as it involves a brute-force search on a large constant-sized instance of the problem.
Moreover, their methods only allow computing an $\varepsilon$ additive approximation to the approximation ratio in time double-exponential in $1/\varepsilon$.

\subsection{Our Results and Techniques.}

Our main contributions are (1) to propose a new SDP rounding technique that is 
based on diffusion processes, and, in particular, on Brownian motion;
(2) to develop the needed tools for analyzing our new SDP rounding
technique by deploying a variety of mathematical techniques from
probability theory, complex analysis and partial differential
equations (PDEs); (3) to show that this rounding technique has useful
concentration of measure properties, not present in random hyperplane based techniques, that can be used
to obtain new approximation algorithms for a version of the \MC
problem with multiple global side constraints.

Our method is inspired by the recent success of Brownian motion based
algorithms for constructive discrepancy minimization, where it was
used to give the first constructive proofs of some of the most
powerful results in discrepancy
theory~\cite{bansal2010constructive,lovett2015constructive,BansalDG16,BansalDGL18}. The
basic idea is to use the solution to the semi-definite program to
define the starting point and the covariance matrix of the
diffusion process, and let the process evolve until it reaches an
integral solution. As the process is forced to stay inside the cube
$[-1,1]^n$ (for \MC) or $[0,1]^n$ (for \MtwoSAT and other problems), and to stick to any
face it reaches, we call the most basic version of our algorithm
(without any enhancements) the \emph{Sticky Brownian Motion}
rounding. The algorithm is defined more formally in Section
\ref{s:sbm-alg}.

\paragraph{Sticky Brownian Motion.}
Using the tools we introduce, we show that
this algorithm is already competitive with the state of the art
results for \MC, \MtwoSAT, and \MdC.

\begin{theorem}\label{thm:maxcut}
  The basic Brownian rounding achieves an approximation ration of \MCRATIO
  for the \MC problem.
  Moreover, when the \MC instance has value $1-\eps$, Sticky Brownian
  Motion achieves value $1- \Omega(\sqrt{\eps})$.
\end{theorem}
In particular, using complex analysis and evaluating various elliptic integrals, we show that the separation probability for any two unit vectors $u$ and $v$ separated by an angle $\theta$, is given by a certain hypergeometric function of $\theta$ (see Theorem \ref{thm:cut-prob} for details). This precise characterization of the separation probability also proves that the Sticky Brownian Motion rounding is different from the random hyperplane rounding. The overview of the analysis is in Section \ref{s:mc-overview} and Section \ref{sec:conformal} has the details.

We can also analytically show the following upper bound for \MtwoSAT.
\begin{theorem}\label{thm:max2sat}
  The Sticky Brownian Motion rounding achieves approximation ratio of at least \MtwoRatio for \MtwoSAT.
\end{theorem}
While the complex analysis methods also give exact results for
\MtwoSAT, the explicit expressions are much harder to obtain as one
has to consider all possible starting points for the diffusion
process, while in the \MC case the process always starts at the
origin. Because of this, in order to prove Theorem~\ref{thm:max2sat}
we introduce another method of analysis based on partial differential
equations (PDEs), and the maximum principle, which allows us to prove
analytic bounds on PDE solutions. Moreover, numerically solving the
PDEs suggests the bound 0.921.  The overview and details of the
\MtwoSAT analysis are, respectively, in Sections \ref{s:ov-max2sat} and
\ref{sec:pde}.  Section~\ref{sec:slowdown} has details about numerical
calculations for various problems.

For comparison, the best known approximation
ratio for \MC is the Goemans-Williamson constant $\alpha_{GW} \approx
\GWratio$, and the best known approximation ratio for \MtwoSAT is
$0.94016$~\cite{LLZ}. The result for \MC instances of value $1-\eps$
is optimal up to constants~\cite{KKMO}, assuming the Unique Games Conjecture.

We emphasize that
our results above are achieved with a single algorithm ``out of the box'', without any
additional engineering. While the analysis uses sophisticated
mathematical tools, the algorithm itself is simple, efficient, and
straightforward to implement.

\paragraph{Extensions.} Next, we consider two different 
modifications of Sticky Brownian Motion that allow us to improve the approximation guarantees
above, and show the flexibility of diffusion based rounding algorithms. The first one is to
smoothly slow down
the process depending on how far it is from the boundaries of the cube.
As a
proof of concept, we show, numerically, that a simple modification of
this kind matches the Goemans-Williamson approximation of \MC up to
the first three digits after the decimal point. We also obtain significant improvements for other problems over the vanilla method.

Second, we propose a variant of Sticky Brownian Motion running
in $n+1$ dimensions rather than $n$ dimensions, and we analyze it for
the \MdC problem. The extra dimension is used to determine whether
the nodes labeled $+1$ or those labeled $-1$ are put on the left
side of the cut. We show that this modification achieves an
approximation ratio of \MDCRatio for \MdC. Slowing down the process further improves this approximation to \MDCRatioSD. We give a summary of the obtained results\footnote{Our numerical results are not obtained via simulating the
random algorithm but solving a discretized version of a PDE that analyzes the
performance of the algorithm. Error analysis of such a discretization can allow
us to prove the correctness of these bounds within a reasonable accuracy.} in Table~\ref{tab:overview-results}.
An overview and details of the extensions are given, respectively, in Sections \ref{s:ov-ext} and \ref{sec:slowdown}.

\begin{table}[!t]
    \centering
    \begin{tabular}{|c|c|c|c|c|}
         \hline
         Algorithm  &  \MC & \MtwoSAT & \MdC$^\star$ \\
         \hline
         Brownian Rounding & $\mcbasicratio$& 0.921 & \MDCRatio\\
         \hline
         Brownian with Slowdown & $\mcexpratio^\dagger$  & 0.929 & \MDCRatioSD \\
         \hline

    \end{tabular}
    \caption{Approximation ratios for  Sticky Brownian Motion rounding and Sticky Brownian Motion with Slowdown.
		$\dagger$ indicates that for \MC, the approximation for the  slowed down walk differs from the GW bound only in the fourth decimal. For \MdC, the ${}^\star$ indicates that we only consider the $n+1$-dimensional walk. }
    \label{tab:overview-results}
\end{table}

\paragraph{Recent Progress.} Very recently, in a beautiful result,
Eldan and Naor \cite{EN16} describe a slowdown process that
exactly achieves the Goemans-Williamson (GW) bound of 0.878 for
Max-Cut, answering an open question posed in an earlier version of
this paper.  This shows that our rounding techniques are at least as powerful as the classical random hyperplane rounding, and are potentially more general and flexible.

In general, given the dearth of techniques for rounding semidefinite
programs, we expect that rounding methods based on diffusion
processes, together with the analysis techniques introduced in this
paper, will find broader use,
and, perhaps lead to improved results for Max-CSP
problems.

\paragraph{Applications.}
To further illustrate the versatility and general applicability of our
approach, we consider the \MC with Side Constraints problem,
abbreviated \cardMC, a generalization of the \MB problem which allows
for multiple global constraints.  In an instance of the \cardMC
problem, we are given an $n$-vertex graph $G=(V,E)$, a collection $\FF
= \{F_1, \ldots, F_k\}$ of subsets of $V$, and cardinality bounds
$b_1, \ldots, b_k \in \N$.  The goal is to find a subset $S \subset V$
that maximizes the weight $|\delta(S)|$ of edges crossing the cut $(S,
V\setminus S)$, subject to having $|S \cap F_i| = b_i$ for all $i \in
[k]$.

Since even checking whether there is a feasible solution is NP-hard~\cite{dischard}, we aim for bi-criteria approximation algorithms.\footnote{We say that a set $S \subset V$ is an $(\alpha,
\eps)$-approximation if $\bigl||S \cap F_i| - b_i \bigr| \le \eps n$
for all $i \in [k]$, and $|\delta(S)| \ge \alpha \cdot |\delta(T)|$ for
all $T\subset V$ such that $|T \cap F_i| = b_i$ for all $i \in [k]$.}
We give the following result for the problem, using the Sticky Brownian Motion as a building tool.
\begin{restatable}{theorem}{mcconstraints}
  \label{thm:mc-constraints}
  There exists a $O(n^{\mathrm{poly}(\log(k)/\eps)})$-time algorithm
  that on input a satisfiable instance $G = (V, E)$, $\FF$, and $b_1,
  \ldots, b_k$, as defined above, outputs a $(\mcallratio-\eps,
  \eps)$-approximation with high probability.
\end{restatable}

In the presence of a single side constraint, the problem is closely
related to the \MB problem~\cite{RaghavendraT12,AustrinBG16}, and,
more generally to \MC with a cardinality constraint. While our methods
use the stronger semi-definite programs considered in
\cite{RaghavendraT12}~and~\cite{AustrinBG16}, the main new technical
ingredient is showing that the Sticky Brownian Motion possesses
concentration of measure properties that allow us to approximately
satisfy multiple constraints. By contrast, the hyperplane rounding and
its generalizations that have been applied previously to the \MC and
\MB problems do not seem to allow for such strong concentration
bounds. For this reason, the rounding and analysis used
in~\cite{RaghavendraT12} only give an
$O(n^{\mathrm{poly}(k/\eps)})$ time algorithm for the \cardMC problem,
which is trivial for $k = \Omega(n)$, whereas our algorithm has
non-trivial quasi-polynomial running time even in this regime.  We
expect that this concentration of measure property will find further
applications, in particular to constraint satisfaction problems with
global constraints.

\paragraph{Remark.} 
We can achieve better results using Sticky Brownian Motion with
slowdown. In particular, in time $O(n^{\mathrm{poly}(\log(k)/\eps)})$
we can get a $(\mcslowallratio-\eps, \eps)$-approximation with high
probability for any satisfiable instance.  However, we focus on the
basic Sticky Brownian Motion algorithm to simplify exposition.  Note
that due to the recent work by Austrin and
Stankovi\'c~\cite{austrin2019global}, we know that adding even a
single global cardinality constraint to the \MC problem makes it
harder to approximate. In particular, they show that subject to a
single side constraint, \MC is Unique Games-hard to approximate within
a factor of approximately $\mcslowallratio$.  Thus, assuming the
Unique Games conjecture, our approximation factor for the \cardMC
problem is optimal up to small numerical errors. (We emphasize the
possibility of numerical errors as both our result, and the hardness
result in~\cite{austrin2019global} are based on numerical
calculations.)




%


\subsection{Overview}
\subsubsection{The Sticky Brownian Motion Algorithm.}
\label{s:sbm-alg}
Let us describe our basic algorithm in some detail. Recall that the
Goemans-Williamson SDP for \MC is equivalent to the
following vector program: given a graph $G = (V, E)$, we write
\begin{align*}
\max ~~~ & \sum _{(i,j)\in E} \frac{1-\bw _i \cdot \bw _j}{2} & \\
s.t. ~~~ & \bw _i \cdot \bw _i = 1 & \forall i\in V
\end{align*}
where the variables $\bw_i$ range over $n$ dimensional real vectors
($n = |V|$). The Sticky Brownian Motion rounding algorithm we propose
maintains a sequence of random fractional solutions $\bX _0, \ldots,
\bX _T$ such that $\bX _0 = \mathbf{0}$ and $\bX _T \in \{-1, +1\}^n$ is
integral. Here, a vertex of the hypercube $\{-1,+1\}^n$ is naturally
identified with a cut, with vertices assigned $+1$ forming one side of the
cut, and the ones assigned $-1$ forming the other side.


Let $A_t$ be the random set of coordinates of $\bX _{t-1}$ which are
\emph{not} equal to $-1$ or $+1$; we call these coordinates active.
At each time step $t = 1, \ldots, T$, the algorithm picks
$\mathbf{\Delta X}_t$ sampled from the Gaussian distribution with mean
$\mathbf{0}$ and covariance matrix $\bW_t$, where $(\bW_t)_{ij} = \bw_i\cdot
\bw_j$ if $i,j \in A_{t}$, and $(\bW_{t})_{ij} = 0$ otherwise. The
algorithm then takes a small step in the direction of $\mathbf{\Delta
  X}_t$, i.e.~sets $\bX _t = \bX _{t-1} + \gamma \mathbf{\Delta
  X}_{t}$ for some small real number $\gamma$. If the $i$-th
coordinate of $\bX _t$ is very close to $-1$ or $+1$ for some $i$,
then it is rounded to either $-1$ or $+1$, whichever is closer. The
parameters $\gamma$ and $T$ are chosen so that the fractional
solutions $\bX _t$ never leave the cube $[-1, 1]^n$, and so that the
final solution $\bX _T$ is integral with high probability. As $\gamma$
goes to $0$, the trajectory of the $i$-th coordinate of $\bX _t$
closely approximates a Brownian motion started at $0$, and stopped when
it hits one of the boundary values $\{-1, +1\}$. Importantly, the
trajectories of different coordinates are correlated according to the
SDP solution.
A precise definition of the algorithm is given in Section~\ref{sec:mc-process-defn}.

The algorithm for \MtwoSAT (and \MdC) is
essentially the same, modulo using the covariance matrix from the
appropriate standard SDP relaxation, and starting the process at
the marginals for the corresponding variables. We explain this in greater
detail below.  

\subsubsection{Overview of the Analysis for \MC}
\label{s:mc-overview}
In order to analyze this algorithm, it is sufficient to understand the
probability that an edge $(i,j)$ is cut as a function of the angle
$\theta$ between the vectors $\bw _i$ and $\bw _j$. Thus, we can focus
on the projection $((\bX_t)_i,(\bX_t)_j)$ of $\bX_t$. We observe that
$((\bX_t)_i,(\bX_t)_j)$ behaves like a discretization of correlated
2-dimensional Brownian motion started at $(0,0)$, until the first time
$\tau$ when it hits the boundary of the square $[-1,1]^2$. After
$\tau$, $((\bX_t)_i,(\bX_t)_j)$ behaves like a discretization of a
1-dimensional Brownian motion restricted to one of the sides of the
square. From now on we will treat the process as being continuous, and
ignore the discretization, which only adds an arbitrarily small error
term in our analysis. It is convenient to apply a linear
transformation to the \emph{correlated} Brownian motion $((\bX_t)_i,(\bX_t)_j)$ so that it behaves
 like a \emph{standard} 2-dimensional Brownian motion $\bB_t$
 started at $(0,0)$. We show that this linear
transformation maps the square $[-1,1]^2$ to a rhombus $\mathbb{S}$
centered at $\mathbf{0}$ with internal angle $\theta$; we can then think of $\tau$ as
the first time $\bB _t$ hits the boundary of $\mathbb{S}$. After time
$\tau$, the transformed process is distributed like a 1-dimensional
Brownian motion on the side of the rhombus that was first hit. To
analyze this process, we need to understand the probability
distribution of $\bB _\tau$. The probability measure associated with
this distribution is known as the \emph{harmonic measure} on the
boundary $\bd\rhombus$ of $\rhombus$, with respect to the starting point $\mathbf{0}$.
These transformations and connections are explained in detail in Section~\ref{sec:mc-analysis}.

The harmonic measure has been extensively studied in
probability theory and analysis. The simplest special case is the
harmonic measure on the boundary of a disc centered at $\mathbf{0}$ with
respect to the starting point $\mathbf{0}$. Indeed, the central symmetry of the
disc and the Brownian motion implies that it is just the
uniform measure. A central fact we use is that harmonic measure in 2
dimensions is preserved under conformal (i.e.~angle-preserving)
maps. Moreover, such maps between polygons and the unit disc have been
constructed explicitly using complex analysis, and, in particular, are
given by the Schwarz-Christoffel formula~\cite{A96}.  Thus, the
Schwarz-Christoffel formula gives us an explicit formulation of
sampling from the harmonic measure on the boundary $\bd\rhombus$ of
the rhombus: it is equivalent to sampling a uniformly random point on
the boundary of the unit disc $\mathbb{D}$ centered at the origin, and
mapping this point via a conformal map $F$ that sends $\mathbb{D}$ to
$\mathbb{S}$. Using this formulation, in
Section~\ref{sec:schwarz-christoffel} we show how to write the
probability of cutting the edge $(i,j)$ as an elliptic integral.

Calculating the exact value of elliptic integrals is a challenging
problem.  Nevertheless, by exploiting the symmetry in the \MC
objective, we relate our particular elliptic integral to integrals of
the incomplete beta and hypergeometric functions.  We further simplify
these integrals and bring them into a tractable form using several key identities from
the theory of special functions.  Putting everything together, we get
a precise closed form expression for the probability that the Sticky
Brownian Motion algorithm cuts a given edge in
Theorem~\ref{thm:cut-prob}, and, as a consequence, we obtain the
claimed guarantees for \MC in
Theorems~\ref{thm:maxcut}~and~\ref{th:pi}.

\subsubsection{Overview of the Analysis for \MtwoSAT}
\label{s:ov-max2sat}
The algorithm for \MtwoSAT is almost identical to the \MC algorithm, except
that the SDP solution is asymmetric, in the following sense. We can think of
the SDP as describing  the mean and covariance of a ``pseudo-distribution''
over the assignments to the variables. In the case of \MC, we could assume
that, without loss of generality, the mean of each variable
(i.e.~one-dimensional marginal) is $0$ {{since $S$ and $\overline{S}$ are equivalent solutions. However, this is not the case for \MtwoSAT.}} 
We use this information, and instead of starting the diffusion process at the
center of the cube, we start it at the point given by the marginals. For
convenience, and also respecting standard convention, we work in the cube
$[0,1]^n$ rather than  $[-1,1]^n$. Here, in the final solution
$X_T$, if $(X_T)_i=0$ we set the $i$-th variable to \emph{true}
and if $(X_T)_i=1$, we set it to \emph{false}. We again analyze each clause $C$
separately, which allows us to focus on the diffusion process projected to the
coordinates $((X_t)_i, (X_t)_j)$, where $i$ and $j$ are the variables appearing
in $C$. 
However, the previous approach of using the Schwarz-Christoffel
formula to obtain precise bounds on the probability does not easily go
through, since it relies heavily on the symmetry of the starting point
of the Brownian motion. It is not clear how to extend the analysis
when we change the starting point to a point other than the center, as
the corresponding elliptic integrals appear to be intractable.

Instead, we appeal to a classical connection between diffusion
processes and partial differential equations~\cite[Chapter
9]{Oksendal-SDE}. Recall that we are focusing on a single clause $C$
with variables $i$ and $j$, and the corresponding diffusion process
$((X_t)_i, (X_t)_j)$ in the unit square $[0,1]^2$ starting at a point
given by the marginals and stopped at the first time $\tau$ when it
hits the boundary of the square; after that time the process continues
as a one-dimensional Brownian motion on the side of the square it
first hit. For simplicity let us assume that both variables appear
un-negated in $C$. The probability that $C$ is satisfied then equals
the probability that the process ends at one of the points $(0,1)$,
$(1,0)$ or $(0,0)$.
Let $u:[0,1]^2 \rightarrow [0,1]$ be the function which assigns to
$(x,y)$ the probability that this happens when the process is started
at $(x,y)$. Since on the boundary $\bd [0,1]^2$ of the square our process is a
one-dimensional martingale, the value of $u(x,y)$ is easy to
compute on $\bd [0,1]^2$, and in fact equals $1-xy$.  Then, in
the interior of the square, we have $u(x,y) = \E[u((X_\tau)_i,
(X_\tau)_j)]$. It turns out that this identifies $u$ as the unique
solution to an elliptic partial differential equation (PDE) $\L u = 0$
with the Dirichlet boundary condition $u(x,y) = 1-xy \ \ \forall (x,y)
\in \partial [0,1]^2$.
In our case, the operator $\L$ just corresponds to Laplace's operator $\L[u] = \frac{\partial^2u}{\partial x^2} + \frac{\partial^2u}{\partial y^2} $
after applying a linear transformation to the variables and the domain. This connection between our rounding algorithm and PDEs is explained in Section~\ref{sec:sde2pde}.

Unfortunately, it is still not straightforward to solve the obtained PDE
analytically. We deal with this difficulty using two natural approaches.
First, we use the \emph{maximum principle} of elliptic PDE's \cite{GT15}, which
allows us to bound the function $u$ from below. In particular, if we can
find a function $g$ such that $g(x,y) \le u(x,y) = 1 - xy$ on the boundary of
the square, and $\L g \ge 0$ in the interior, then the maximum principle tells
us that $g(x,y) \le u(x,y)$ for all $x,y$ in the square.  We exhibit simple low-degree polynomials which satisfy the boundary conditions by
design, and use the sum of squares proof system to certify non-negativity under the operator $\L$.
In Section~\ref{sec:max-principle}, we use  this method to show that Sticky Brownian Motion rounding achieves approximation ratio at least  \MtwoRatio.

Our second approach is to solve the PDE  numerically to a high degree of accuracy using
finite element methods. We use this approach in Section~\ref{sec:slowdown} to numerically obtain results showing
a $\MtwoRatioSD$ approximation ratio for \MtwoSAT.

\subsubsection{Extensions of Sticky Brownian Motion.}
\label{s:ov-ext}
\paragraph{Using different slowdown functions.}
Recall that in the Sticky Brownian Motion rounding each increment is proportional to $\mathbf{\Delta X}_t$ sampled from a Gaussian distribution with mean $\mathbf{0}$ and covariance matrix $\bW _t$. The covariance is derived from the SDP: for example, in the case of \MC, it is initially set to be the Gram matrix of the vectors produced by the SDP solution. Then, whenever a coordinate $(\bX _t)_i$ reaches $\{-1, +1\}$, we simply zero-out the corresponding row and column of $\bW _t$. This process can be easily modified by varying how the covariance matrix $\bW_t$ evolves with time. Instead of zeroing out rows and columns of $\bW_t$, we can smoothly scale them based on how far $(\bX_{t-1})_i$ is from the boundary values $\{-1, 1\}$. A simple way to do this, in the case of the \MC problem, is to set
\[
(\bW_t)_{ij} = (1- (\bX_{t-1})_i^2)^{\alpha/2}(1-\bX_{t-1})_j^2)^{\alpha/2} \bw_i\cdot \bw_j
\]
for a constant $0 \le \alpha < 2$. Effectively, this means that the
process is slowed down smoothly as it approaches the boundary of the cube
$[-1,+1]^n$.  This modified diffusion process, which we call Sticky
Brownian Motion with Slowdown, still converges to $\{-1, +1\}^n$ in
finite time. Once again, the probability of cutting an edge $(i,j)$ of
our input graph can be analyzed by focusing on the two-dimensional
projection $((\bX_t)_i, (\bX_t)_j)$ of $\bX_t$. Moreover, we can still
use the general connection between diffusion processes and PDE's
mentioned above. That is, if we write $u(x,y):[-1,1]^2 \to [0,1]$ for the
probability that edge $(i,j)$ is cut if the process is started at
$(x,y)$, then $u$ can be characterized as the
solution of an elliptic PDE with boundary conditions $u(x,y) = \frac{1
  - xy}{2} \ \ \ \forall (x,y) \in \partial [-1,1]^2$. We solve this
PDE numerically using the finite element method to estimate the
approximation ratio for a fixed value of the parameter $\alpha$, and
then we optimize over $\alpha$. At the value $\alpha = \slowdown$ our
numerical solution shows an approximation ratio that matches the
Goemans-Williamson approximation of \MC up to the first three digits
after the decimal point. We also analyze an analogous algorithm for
\MtwoSAT and show that for $\alpha = \slowdown$ it achieves an approximation
ratio of 0.929. The detailed analysis of the slowed down Sticky Brownian Motion rounding is given in Section~\ref{sec:slowdown}. 

\paragraph{A higher-dimensional version.}
We also consider a higher-dimensional version of the Sticky Brownian Motion
rounding, in which the Brownian motion evolves in $n+1$ dimensions rather than
$n$. This rounding is useful for asymmetric problems like \MdC\footnote{The input for \MdC is a directed graph $G=(V,E)$ and the goal is to find a cut $S\subseteq V$ that maximizes the number of edges going from $S$ to $\overline{S}$.} in which the SDP
produces non-uniform marginals, as we discussed above in the context of
\MtwoSAT. Such an SDP has a vector $\bw_0$ in addition to $\bw_1, \ldots,
\bw_n$, and the marginals are given by $\bw_0\cdot \bw_i$.  Now, rather than
using the marginals to obtain a different starting point, we consider the
symmetric Sticky Brownian Motion process starting from the center but using all
the $n+1$ vectors $\bw_0, \ldots, \bw_n$. At the final step $T$ of the process,
in the case of \MdC, the variables whose value is equal to $(\bX_T)_0$ are
assigned to the left side of the cut, and the variables with the opposite value
are assigned to the right side of the cut. Thus, for an edge $i \to j$ to be
cut, it must be the case that  $(\bX_T)_i = (\bX_T)_0$ and $(\bX_T)_j = 1 -
(\bX_T)_0$. While analyzing the probability that this happens is a question
about Brownian motion in three rather than two dimensions, we reduce it to a
two-dimensional question via the inclusion-exclusion principle. After this
reduction, we can calculate the probability that an edge is cut by using
the exact formula proved earlier for the \MC problem.  Our analysis, which is given in Section~\ref{sec:higher-dim}, shows that this
$(n+1)$-dimensional Sticky Brownian Motion achieves an approximation of
\MDCRatio for \MdC. Moreover, combining the two ideas, of changing the
covariance matrix at each step, as well as performing the $n+1$-dimensional Sticky
Brownian Motion, achieves a ratio of \MDCRatioSD.

\subsubsection{Overview of the Analysis for \cardMC.}

The starting point for our algorithm for the \cardMC problem is a
stronger SDP relaxation derived using the Sum of Squares (SoS)
hierarchy. Similar relaxations were previously considered
in~\cite{RaghavendraT12,AustrinBG16} for the \MB problem. In addition
to giving marginal values and a covariance matrix for a
``pseudo-distribution'' over feasible solutions, the SoS SDP makes it
possible to condition on small sets of variables. The global
correlation rounding method~\cite{BarakRS11,GuruswamiS11} allows us to
choose variables to condition on so that, after the conditioning, the
covariance matrix has small entries on average. Differing from
previous works~\cite{RaghavendraT12,AustrinBG16}, we then run the
Sticky Brownian Motion rounding defined by the resulting marginals and
covariance matrix. We can analyze the weight of cut edges using the
PDE approach outlined above. The main new challenge is to bound the
amount by which the side constraints are violated. To do so, we show
that Sticky Brownian Motion concentrates tightly around its mean, and,
in particular, it satisfies sub-Gaussian concentration in directions
corresponding to sets of vertices. Since the mean of the Sticky
Brownian Motion is given by the marginals, which satisfy all side
constraints, we can bound how much constraints are violated via the
concentration and a union bound. To show this key concentration
property, we use the fact that the covariance that defines the
diffusion has small entries, and that Brownian Motion is a
martingale. Then the concentration inequality follows, roughly, from a
continuous analogue of Azuma's inequality. The detailed analysis is
given in Section~\ref{sect:mc-constraints}. We again remark that such
sub-Gaussian concentration bounds are not known to hold for the random
hyperplane rounding method or its generalizations as considered in
\cite{RaghavendraT12,AustrinBG16}. 

\subsection{Related Work}

In their seminal work, \citet{GW95} presented the random hyperplane rounding method which yielded an approximation of \GWratio for \MC.
For the closely related \MdC problem they presented an approximation of $0.796$.
This was subsequently improved in a sequence of papers: Feige and Goemans \cite{feige1995approximating} presented an approximation of $0.859$; Matuura and Matsui improved the factor to $0.863$; and culminating in the work of Lewin {\em et. al.} \cite{LLZ} who present the current best known approximation of $0.874$, getting close to the \GWratio~ approximation of \cite{GW95} for \MC.
Another fundamental and closely related problem is \textsc{Max-Bisection}.
In their classic work \cite{Frieze1997}, Frieze and Jerrum present an approximation of $0.651$ for this problem.
Their result was later improved to $0.699$ by \citet{Ye2001}, to $0.701$ by \citet{HalperinZ02}, and to $0.702$ by \citet{FeigeL01}. Using the sum of squares hierarchy, \citet{RaghavendraT12} gave a further improvement to $0.85$, and finally, Austrin {\em et. al.} \cite{AustrinBG16} presented an almost tight approximation of $0.8776$.
With respect to hardness results, H{\aa}stad \cite{Hastad2001} proved a hardness of $\nicefrac[]{16}{17}$ for \MC (which implies the exact same hardness for \textsc{Max-Bisection}) and a hardness of $\nicefrac[]{11}{12}$ for \MdC (both of these hardness results are assuming $\mathsf{P}\neq \mathsf{NP}$).
If one assumes the Unique Games Conjecture of Khot \cite{Khot02-UG}, then it is known that the random hyperplane rounding algorithm of \cite{GW95} is tight \cite{KKMO,MOO}.
Thus, it is worth noting that though \MC is settled conditional on the Unique Games conjecture, both \MdC and \textsc{Max-Bisection} still remain unresolved, even conditionally.

Another fundamental class of closely related problems are \MS and its special cases \textsc{Max-$k$-SAT}.
For \MtwoSAT Goemans and Williamson \cite{GW95}, using random hyperplane rounding, presented an approximation of \GWratio.
This was subsequently improved in a sequence of works: Feige and Goemans \cite{feige1995approximating} presented an approximation of $0.931$; Matsui and Matuura \cite{MM01} improved the approximation factor to $0.935$; and finally Lewin {\em et. al.} \cite{LLZ} presented the current best known approximation of $0.94016$.
Regarding hardness results for \MtwoSAT, assuming $\mathsf{P}\neq \mathsf{NP}$, H{\aa}stad \cite{Hastad2001} presented a hardness of $\nicefrac[]{21}{22}$.
Assuming the Unique Games Conjecture Austrin \cite{Austrin07} presented a (virtually) tight hardness of $0.94016$, matching the algorithm of \cite{LLZ}.
For \textsc{Max-3SAT}, Karloff and Zwick \cite{karloff19977} and Zwick \cite{Zwick1998} presented an approximation factor of $\nicefrac[]{7}{8}$ based on the random hyperplane method.
The latter is known to be tight by the celebrated hardness result of H{\aa}stad \cite{Hastad2001}.
For \textsc{Max-4SAT} Halperin and Zwick \cite{HZ99} presented an (almost) tight approximation guarantee of $0.8721$.
When considering \MS in its full generality, a sequence of works \cite{asano2002improved,A06,avidor2005improved} slowly improved the known approximation factor, where the current best one is achieved by Avidor {\em et. al.} \cite{avidor2005improved} and equals $0.797$.\footnote{Avidor {\em et. al.} also present an algorithm with a {\em conjectured} approximation of $0.8434$, refer to \cite{avidor2005improved} for the exact details.}
For the general case of \textsc{Max-CSP} a sequence of works \cite{Austrin10,Raghavendra08} culminated in the work of Raghavendra and Steurer \cite{RaghavendraS09} who presented an algorithm that assuming the Unique Games Conjecture matches the hardness result for any constraint satisfaction problem.
However, as previously mentioned, this universal rounding is impractical as it involves a brute-force solution to a large constant instance of the problem.
Moreover, it only allows computing an $\varepsilon$ additive approximation to the approximation ratio in time double-exponential in $\nicefrac[]{1}{\varepsilon}$.

Many additional applications of random hyperplane rounding and its extensions exist.
Some well known examples include: \textsc{3-Coloring} \cite{ACC06,BK97,KMS98}, \textsc{Max-Agreement} in correlation clustering \cite{CharikarGW05,Swamy04}, the maximization of quadratic forms \cite{Charikar04}, and the computation of the \textsc{Cut-Norm} \cite{AlonNaor}.

Let us now briefly focus on the extensions and generalizations of random hyperplane rounding.
The vast majority of the above mentioned works use different extensions of the basic random hyperplane rounding.
Some notable examples include: rotation of the vectors \cite{AlonNaor,N98,Ye2001,zwick1999outward}, projections \cite{FeigeL01,OW08,Charikar04}, and combining projections with clustering \cite{Austrin10,Raghavendra08,RaghavendraS09}.
It is worth noting that the above extensions and generalizations of the basic random hyperplane method are not the only approaches known for rounding SDPs.
The most notable example of the latter is the seminal work of Arora {\em et. al.} \cite{ARV09} for the \textsc{Sparsest-CUT} problem.
Though the approach of \cite{ARV09} uses random projections, it is  based on different mathematical tools, {\em e.g.}, L\'{e}vy's isoperimetric inequality.
Moreover, the algorithmic machinery that was developed since the work of \cite{ARV09} has found uses for {\em minimization} problems, and in particular it is useful for minimization problems that relate to graph cuts and clustering.

Brownian motion was first used for rounding SDPs in Bansal \cite{bansal2010constructive} in the context of constructive discrepancy minimization. This approach has since proved itself very successful in this area, and has led to new constructive proofs of several major results~\cite{lovett2015constructive,BansalDG16,BansalDGL18}.  However, this line of work has largely focused on improving logarithmic factors, and its methods are not precise enough to analyze constant factor approximation ratios.

\section{Brownian Rounding for \MC via Conformal Mappings}\label{sec:conformal}
In this section, we use \MC as a case study for the method of rounding a
semi-definite relaxation via Sticky Brownian Motion. Recall, in an instance
of the \MC problem we are given a graph $G=(V,E)$ with edge weights
$a:E\rightarrow \Re_+$ and the goal is to find a subset $S\subset V$ that
maximizes the total weight of edges crossing the cut $(S,V\setminus S)$, i.e.,
$a(\delta(S))\eqdef\sum_{\{u,v\}\in E: u\in S, v\notin S} a_{uv}$. We first
introduce the standard semi-definite relaxation for the problem and introduce
the sticky Brownian rounding algorithm. To analyze the algorithm, we
use the invariance of Brownian motion with respect to conformal
maps, along with several identities of special functions.

\subsection{SDP Relaxation and Sticky Brownian Rounding Algorithm}
\label{sec:mc-process-defn}

Before we proceed, we recall again the SDP formulation for the \MC
problem,  famously studied by \citet{GW95}.
\begin{align*}
\max ~~~~~ & \sum_{e=(i,j)\in E} a(e)\frac{\left( 1-\bw _i \cdot \bw _j\right)}{2} & \\
s.t. ~~~~~ & \bw _i \cdot \bw _i = 1 & \forall i=1,...,n
\end{align*}

We now describe the Sticky Brownian Motion rounding algorithm
specialized to the \MC problem.  Let $\bW$ denote the positive
semi-definite correlation matrix defined by the vectors $\bw_1,\ldots,
\bw_n$, {\em i.e.}, for every $1\leq i,j\leq n$ we have that: $\bW
_{i,j} = \bw _i \cdot \bw _j$.  Given a solution $\bW$ to the
semi-definite program, we perform the following rounding process:
start at the origin and perform a Brownian motion inside the
$[-1,1]^n$ hypercube whose correlations are governed by $\bW$.
Additionally, the random walk is {\em sticky}: once a coordinate
reaches either $-1$ or $+1$ it is fixed and does not change anymore.

Formally, we define a random process $\left\{ \bX_t\right\} _{t\geq 0}$ as
follows. We fix $\bX_0 = \mathbf{0}$. Let $\left\{ \bB
_t\right\} _{t\geq 0}$ be standard Brownian motion in $\mathbb{R}^n$
started at the origin,\footnote{We will always assume that a standard Brownian motion
starts at the origin. See Appendix~\ref{app:brownian-defn} for a
precise definition.} and let $\tau_1 = \inf\{t:
\mathbf{x}_0 +\bW^{\nicefrac[]{1}{2}} \bB_t \not \in [-1,1]^n\}$ be the first time
$\mathbf{x}_0 + \bW^{\nicefrac[]{1}{2}} \bB_t$ exits the cube. With
probability $1$, you can assume that $\tau_1$ is also the first time
that the process lies on the boundary of the cube. Here $\bW
^{\nicefrac[]{1}{2}}$ is the principle square root of $\bW$. Then, for
all $0 \le t \le \tau_1$ we define
\[
\bX _t = \mathbf{x}_0 +\bW ^{\nicefrac[]{1}{2}}\bB _t.
\]
This defines the process until the first time it hits a face of the
cube. From this point on, we will force it to stick to this face. Let
$A_t = \{i: (\bX _t)_i \neq \pm 1\}$ be the active coordinates of the
process at time $t$, and let $F_t = \{\mathbf{x} \in [-1,1]^n: x_i =
(\bX_t)_i \forall i \in A_t\}$ be the face of the cube on which
$\bX_t$ lies at time $t$. With probability $1$, $F_{\tau_1}$ has
dimension $n-1$. We define the covariance matrix $(\bW_t)_{ij} = \bW_{ij}$
when $i,j \in A_t$, and $(\bW_t)_{ij} = 0$ otherwise. Then we
take $\tau_2 = \inf\{t\ge \tau_1: \bX_{\tau_1} +
\bW_{\tau_1}^{\nicefrac[]{1}{2}} (\bB_t - \bB_{\tau_1}) \not \in F_{\tau_1}\}$ to
be the first time that Brownian motion started at $\bX_{\tau_1}$ with
covariance given by $\bW_{\tau_1}$ exits
the face $F_{\tau_1}$. Again, with probability $1$, we can assume that
this is also the first time the process lies on the boundary of
$F_{\tau_1}$. For all $\tau_1<  t \le \tau_2$ we define
\[
\bX _t = \bX_{\tau_1} + \bW_{\tau_1} ^{\nicefrac[]{1}{2}} (\bB_t - \bB_{\tau_1}).
\]
Again, with probability $1$, $\mathrm{dim}\ F_{\tau_2} = n-2$.
The process is defined analogously from here on. In general, $\tau_i =
\inf\{t \ge \tau_{i-1}: \bX_{\tau_{i-1}} + \bW_{\tau_{i-1}}^{\nicefrac[]{1}{2}}
(\bB_t - \bB_{\tau_{i-1}})\not \in F_{\tau_{i-1}}\}$ is (with
probability $1$) the first time that the process
hits a face of the cube of dimension $n-i$. Then for $\tau_{i-1} < t
\le \tau_i$ we have $\bX _t = \bX_{\tau_{i-1}} + \bW_{\tau_{i-1}}
^{\nicefrac[]{1}{2}}(\bB_t - \bB_{\tau_{i-1}})$. At time $\tau_n$, $\bX_{\tau_n} \in \{-1,
1\}^n$, so the process remains fixed, i.e.~for any $t \ge \tau_n$,
$\bX_t  = \bX_{\tau_n}$.
The output of the algorithm then corresponds to a cut $S \subseteq V$
defined as follows:
\[ S=\left\{ i\in V:\left( \bX _{\tau_n}\right)_i = 1\right\}.\]
We say that a pair of nodes $\{i,j\}$ is \emph{separated} when $|S
\cap \{i,j\}|= 1$.

\noindent\textbf{Remark:} While we have defined the algorithm as a continuous
diffusion process, driven by Brownian motion, a standard
discretization will yield a polynomial time algorithm that achieves
the same guarantee up to an error that is polynomially small. Such a
discretization was outlined in the Introduction. An analysis of the
error incurred by discretizing a continuous diffusion process in this
way can be found, for example, in \cite{Gobet00} or the book
\cite{Gobet16}. More sophisticated discrete simulations of such
diffusion processes are also available, and can lead to better time
complexity as a function of the error. One example is the Walk on
Spheres algorithm analyzed by~\citet{BinderBraverman}.  This algorithm
allows us to draw a sample $\bX_{\tau}$ from the continuous diffusion
process, stopped at a random time $\tau$, such that $\bX_\tau$ is
within distance $\varepsilon$ from the boundary of the cube $[-1,
1]^n$. The time necessary to sample $\bX_\tau$ is polynomial in $n$
and $\log(1/\varepsilon)$. We can then round $\bX_\tau$ to the
nearest point on the boundary of the cube, and continue the simulation
starting from this rounded point. It is straightforward to show, using
the fact that the probability to cut an edge is continuous in the
starting point of our process, that if we set $\varepsilon =
o(n^{-1})$, then the approximation ratio achieved by this simulation
is within an $o(1)$ factor from the one achieved by the continuous
process.  In the rest of the paper, we focus on the continuous process
since our methods of analysis are naturally amenable to it.

\subsection{Analysis of the Algorithm}\label{sec:mc-analysis}

Our aim is to analyze the expected value of the cut output by the
Sticky Brownian Motion
rounding algorithm. Following Goemans and Williamson~\cite{GW95},  we aim to
bound the probability an edge is cut as compared to its contribution to the SDP
objective.  Theorem~\ref{thm:cut-prob} below gives an \emph{exact}
characterization of the probability of separating a pair of vertices
$\{i,j\}$ in terms of the gamma function and hypergeometric
functions. We refer to Appendix~\ref{sec:primer} for the definitions of these
functions and a detailed exposition of their basic properties.

\cut{We give a more detailed exposition of basic facts about
special functions in the appendix, and for now just state the main
definitions.

\begin{restatable}[Gamma Function]{definition}{gammadef}
    \label{def:gamma}
    The gamma function is defined as
    \[ \Gamma(z) \eqdef \int_0^\infty x^{z-1} e^{-x} dx. \]
    for all complex numbers $z$ with non-negative real part, and
    analytically extended to all $z \neq 0,-1,-2,...$.
\end{restatable}

\begin{restatable}[Hypergeometric Function]{definition}{hypergeomdef}
    \label{def:hypergeom}
       The {\em hypergeometric function} ${}_pF_q(a_1,\ldots,a_p,b_1;\ldots,b_q;z)$ is defined as
\[ {}_pF_q(a_1,\ldots,a_p,b_1,\ldots,b_q;z)  \eqdef \sum_{n=0}^\infty \frac{ (a_1)_n \cdots (a_p)_n}{(b_1)_n \cdot (b_q)_n} \frac{z^n}{n!} \]
where the  {\em Pochhammer symbol} (rising factorial) is defined
inductively as  \[(a)_n :=  a(a+1) \cdots(a+n-1) \text{ and }
(a)_0=1.\]
\end{restatable}

We now give the main theorem of the section.}

\begin{restatable}{theorem}{cutprobthm}
 \label{thm:cut-prob}
    The probability that the Sticky Brownian Motion rounding algorithm
    will separate a pair $\{i,j\}$ of vertices for which
    $\theta=cos^{-1}(\bw_i \cdot \bw_j)$ equals

    \[ 1- \frac{\Gamma( \frac{a+1}{2})}{\Gamma(\frac{1-a}{2})\Gamma(\frac{a}{2}+1)^2} \cdot
    \pFq{3}{2}{\frac{1+a}{2},\frac{1+a}{2},\frac{a}{2}}{\frac{a}{2},\frac{a}{2}+1}{1} \]

    where $a = \theta/\pi$, $\Gamma$ is the gamma function, and
    ${}_3F_2$ is the hypergeometric function.
\end{restatable}

Theorem~\ref{thm:maxcut} will now follow from the following corollary of Theorem~\ref{thm:cut-prob}. The corollary follows from numerical estimates of the gamma and hypergeometric functions.

\begin{restatable}[]{corollary}{paircutcor}
\label{corl:pair-cut}
For any pair $\{i,j\}$, the probability that the pair $\{i,j\}$ is separated is at least $\MCRATIO \cdot \frac{1-\bw_i\cdot \bw_j}{2}$.
\end{restatable}

We now give an outline of the proof of Theorem~\ref{thm:cut-prob}. The
plan is to first show that the desired probability can be obtained by
analyzing the two-dimensional standard Brownian motion starting at the
center of a rhombus. Moreover, the probability of separating $i$ and
$j$ can be computed using the distribution of the first point on the
boundary that is hit by the Brownian motion. Conformal mapping and, in
particular, the Schwarz-Christoffel formula, allows us to obtain a
precise expression for such a distribution and thus for the separation
probability, as claimed in the theorem. We now expand on the above
plan.

First observe that to obtain the probability $i$ and $j$ are separated, it
is enough to consider the 2-dimensional process obtained by projecting to the
$i^{th}$ and $j^{th}$ coordinates of the vector $\bX_t$. Projecting the process
onto these coordinates, we obtain a process $\tilde{\bX}_t\in \mathbb{R}^2$ that
can be equivalently defined as follows. Let
\[ \tilde{\bW}=\begin{pmatrix} 1 & \cos (\theta)
            \\ \cos (\theta) & 1 \end{pmatrix},
\] where $\theta$ is the angle between $\bw _i$ and
$\bw _j$.
Let $\bB _t$ be standard Brownian motion in $\Re^2$ started at $0$, and
let $\tau = \inf\{t: \tilde{\bW}^{\nicefrac[]{1}{2}} \bB_t \not \in
[-1,1]^t\}$ be the first time the process hits the boundary of the square.
Then for all $0 \le t \le \tau$ we define $\tilde{\bX}_t =
\tilde{\bW}^{\nicefrac[]{1}{2}} \bB_t$. Any coordinate $k$ for
which $(\tilde{\bX}_\tau)_k \in \{\pm 1\}$ remains fixed from then on,
i.e.~for all $t > \tau$, $(\tilde{\bX}_t)_k =
(\tilde{\bX}_\tau)_k$. The coordinate $\ell$ that is not fixed at time
$\tau$ (one exists with probability $1$)
continues to perform one-dimensional Brownian motion started from
$(\tilde{\bX}_\tau)_\ell$ until it hits $-1$ or $+1$, at which point it
also becomes fixed. Let $\sigma$  be the time this happens; it is easy to
show that $\sigma < \infty$ with probability $1$, and, moreover, $\E[\sigma] <
\infty$. We say that the process $\set{\tilde{\bX}_t}_{t \ge 0}$ is \emph{absorbed} at
the vertex $\tilde{\bX}_\sigma \in \{-1,1\}^2$.

\begin{observation}
    The probability that the algorithm separates vertices $i$ and $j$ equals
    \[ \Pr \Big[ \set{\tilde{\bX_t}}_t \text{ is absorbed in } \set{(+1,-1),(-1,+1)}\Big]. \]
\end{observation}

With an abuse of notation, we denote $\tilde{\bX_t}$ by $\bX_t$ and
$\tilde{\bW}$ by $\bW$ for the rest of the section which is aimed at
analyzing the above probability.  We also denote by $\rho = \cos
(\theta)$ the correlation between the two coordinates of the random
walk, and call the two-dimensional process just described a
$\rho$-correlated walk.  It is easier to bound the probability that
$i$ and $j$ are separated by transforming the
$\rho$-correlated walk inside $[-1,1]^2$ into a
standard Brownian motion inside an appropriately scaled rhombus. We do
this by transforming $\set{\bX_t}_{t \ge 0}$ linearly into an
auxiliary random process $\set{\bY_t}_{t \ge 0}$ which will be sticky
inside a rhombus (see Figures~\eqref{fig:A}-\eqref{fig:B}).  Formally,
given the random process $\left\{ \bX _t\right\}_{t\geq 0}$, we
consider the process $\bY_t = \bO \cdot \bW^{-\nicefrac[]{1}{2}}\cdot
\bX_t$, where $\bO$ is a rotation matrix to be chosen
shortly. Recalling that for $0 \le t \le \tau$ the process $\left\{
  \bX _t\right\}_{0 \le t\le \tau}$ is distributed as
$\left\{\bW^{\nicefrac[]{1}{2}} \bB _t\right\}_{0 \le t\le \tau}$, we
have that, for all $0 \le t \le \tau$,
\[
\bY_t = \bO\cdot \bB_t \equiv \bB_t.
\]
Above $\equiv$ denotes equality in distribution, and follows from the
invariance of Brownian motion under rotation.
Applying $\bO \bW^{-\nicefrac[]{1}{2}}$ to the points inside $[-1,1]^2$, we
get a rhombus $\rhombus$ with vertices $b_1,\dots,b_4$, which are the images of the
points $(+1,-1),(+1,+1),(-1,+1),(-1,-1)$, respectively. We choose
$\bO$ so that $b_1$ lies on the positive $x$-axis and $b_2$ on the
positive $y$-axis.
Since $\bO \bW^{-\nicefrac[]{1}{2}}$ is a linear
transformation, it maps the interior of $[-1,1]^2$ to the interior of
$\rhombus$ and the sides of $[-1,1]^2$ to the sides of $\rhombus$. We have then that $\tau$ is
the first time $\bY_t$ hits the boundary of $\rhombus$, and that after
this time $\bY_t$ sticks to the side of $\rhombus$ that it first hit
and evolves as (a scaling of) one-dimensional Brownian motion
restricted to this side, and started at $\bY_\tau$. The process then
stops evolving at the
time $\sigma$ when $\bY_\sigma \in \{b_1, \ldots, b_4\}$. We say that
$\set{\bY_t}_{t\ge 0}$ is absorbed at $\bY_\sigma$.

The following lemma, whose proof appears  in the appendix, formalizes the main facts we use about this
transformation.
\begin{restatable}{lemma}{basictransform} \label{lem:basic-transform}
    Applying the transformation $\bO \bW^{-\nicefrac[]{1}{2}}$  to
    $\set{\bX_t}_{t\ge 0}$ we
    get a new random process $\set{\bY_t}_{t\ge 0}$ which has the following properties:
    \begin{enumerate}
        \item If $\bX_t$ is in the interior/boundary/vertex of $[-1,1]^2$ then
            $\bY_t$ is in the interior/boundary/vertex of $\rhombus$,
            respectively.
        \item $\rhombus$ is  a rhombus whose internal angles at $b_1$
          and $b_3$ are $\theta$, and at $b_2$ and $b_4$ are $\pi-
          \theta$. The vertex $b_1$ lies on the positive $x$-axis, and $b_2, b_3,
          b_4$ are arranged counter-clockwise.
        \item The probability that the algorithm will separate the pair $\{i,j\}$ is
    exactly $\Pr[\bY_t \text{ is absorbed in } b_1 \text{ or } b_3]. $
    \end{enumerate}
\end{restatable}

In the following useful lemma we show that, in order to compute the
probability that the process $\set{\bY_t}_{t\ge 0}$ is absorbed in
$b_1$ or $b_3$, it suffices to determine the distribution of the first
point $\bY_\tau$ on the boundary $\bd \rhombus$ that the process
$\set{\bY_t}_{t\ge 0}$ hits. This distribution is a probability
measure on $\bd \rhombus$ known in the literature as \emph{the
  harmonic measure} (with respect to the starting point $0$). We
denote it by $\mu_{\bd \rhombus}$. The statement of the lemma follows.

\begin{lemma} \label{lem:rhombus-symmetry}
  \[ \Pr[\bY_t \text{ is absorbed in } b_1 \text{ or } b_3]=
  4 \cdot \int_{b_1}^{b_2} 1- \frac{ \norm{p-b_1}}{ \norm{b_2- b_1}} d\mu_{\partial \rhombus}(p) .
  \]
\end{lemma}
\begin{proof}
  Since both $\rhombus$ and Brownian motion are symmetric
  with respect to reflection
  around the coordinate axes, we see that
  $\mu_{\bd \rhombus}$ is the same as we go from $b_1$ to $b_2$ or
  $b_4$, and as we go from $b_3$ to $b_2$ or $b_4$. Therefore,
 \[ \Pr[\text{pair $\{i,j\}$ is separated}] = 4\cdot \Pr[\text{pair $\{i,j\}$ is separated} \mid
 \bY_\tau \text{ lies on the segment } [b_1,b_2]]. \]
 The process $\set{\bY_t}_{\tau \le t \le \sigma}$ is a one-dimensional
 martingale, so $\E[\bY_\sigma|\bY_\tau] = \bY_\tau$ by the optional
 stopping theorem~\cite[Proposition 2.4.2]{MP10}. If we also condition on $\bY_\tau \in [b_1,b_2]$, we
 have that $\bY_\sigma \in \{b_1,b_2\}$.  An easy calculation then shows
 that the probability of being absorbed in $b_1$ conditional on
 $\bY_\tau$ and on the event $\bY_\tau \in [b_1, b_2]$ is exactly $
 \frac{ \norm{\bY_\tau-b_2}}{ \norm{b_2- b_1}} = 1 - \frac{ \norm{\bY_\tau-b_1}}{ \norm{b_2- b_1}}$. Then,
 \[
 \Pr[\text{pair $\{i,j\}$ is separated} \mid \bY_\tau \in [b_1,b_2]]
 = \E\left[1 - \frac{ \norm{\bY_\tau-b_1}}{ \norm{b_2- b_1}}\right]
 = \int_{b_1}^{b_2}
   1- \frac{ \norm{p-b_1}}{ \norm{b_2- b_1}} d\mu_{\partial \rhombus}(p) .
 \]
This proves the lemma.
\end{proof}

\begin{figure}
    \centering
    \begin{subfigure}[b]{0.31\textwidth}
        \includegraphics[width=\textwidth]{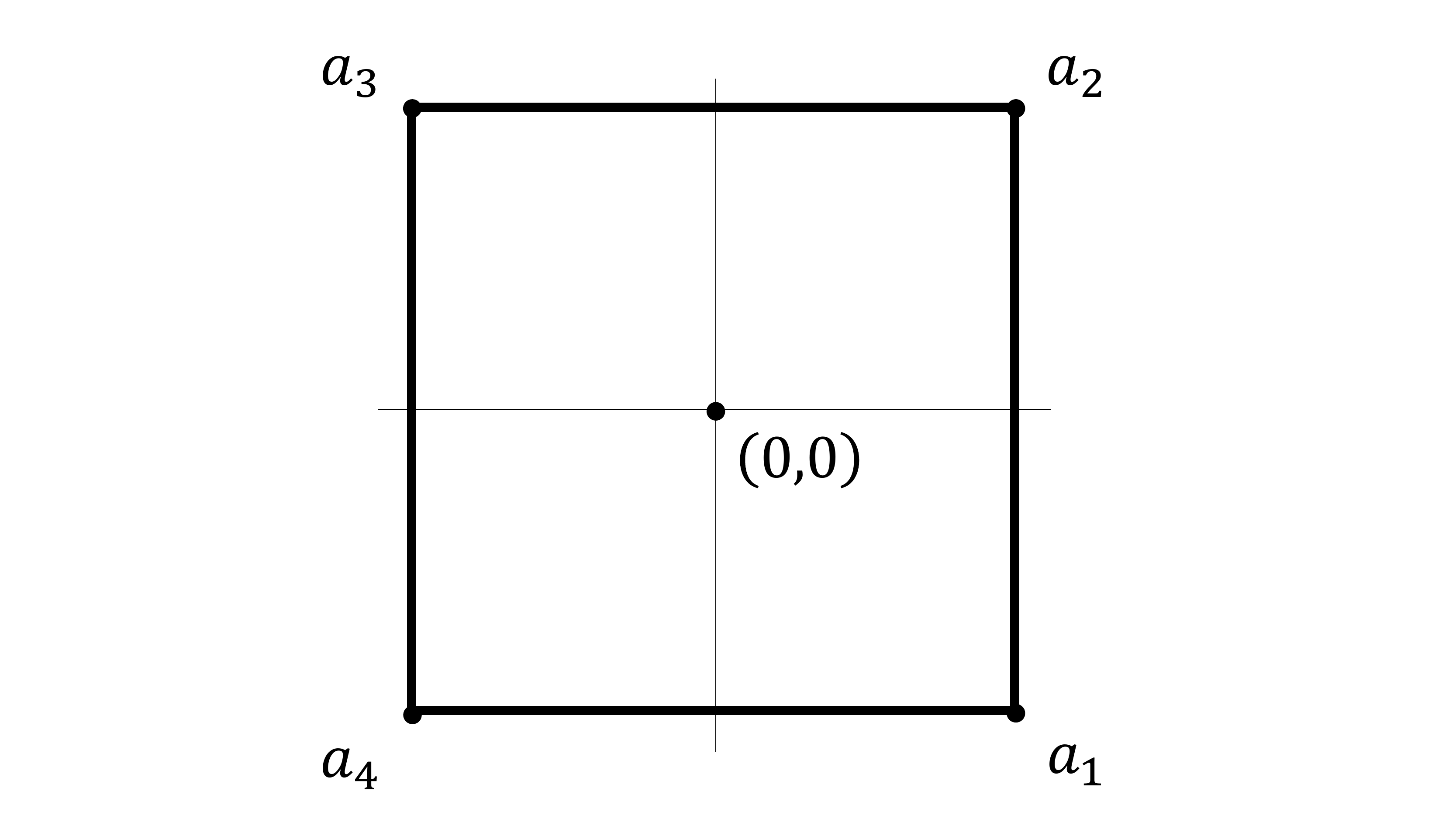}
        \caption{$\left\{\bX_t\right\} _{t\geq 0}$ in $[-1,1]^2$ square}
        \label{fig:A}
    \end{subfigure}
    ~ 
    \begin{subfigure}[b]{0.31\textwidth}
        \hfil
        \includegraphics[width=0.5\textwidth]{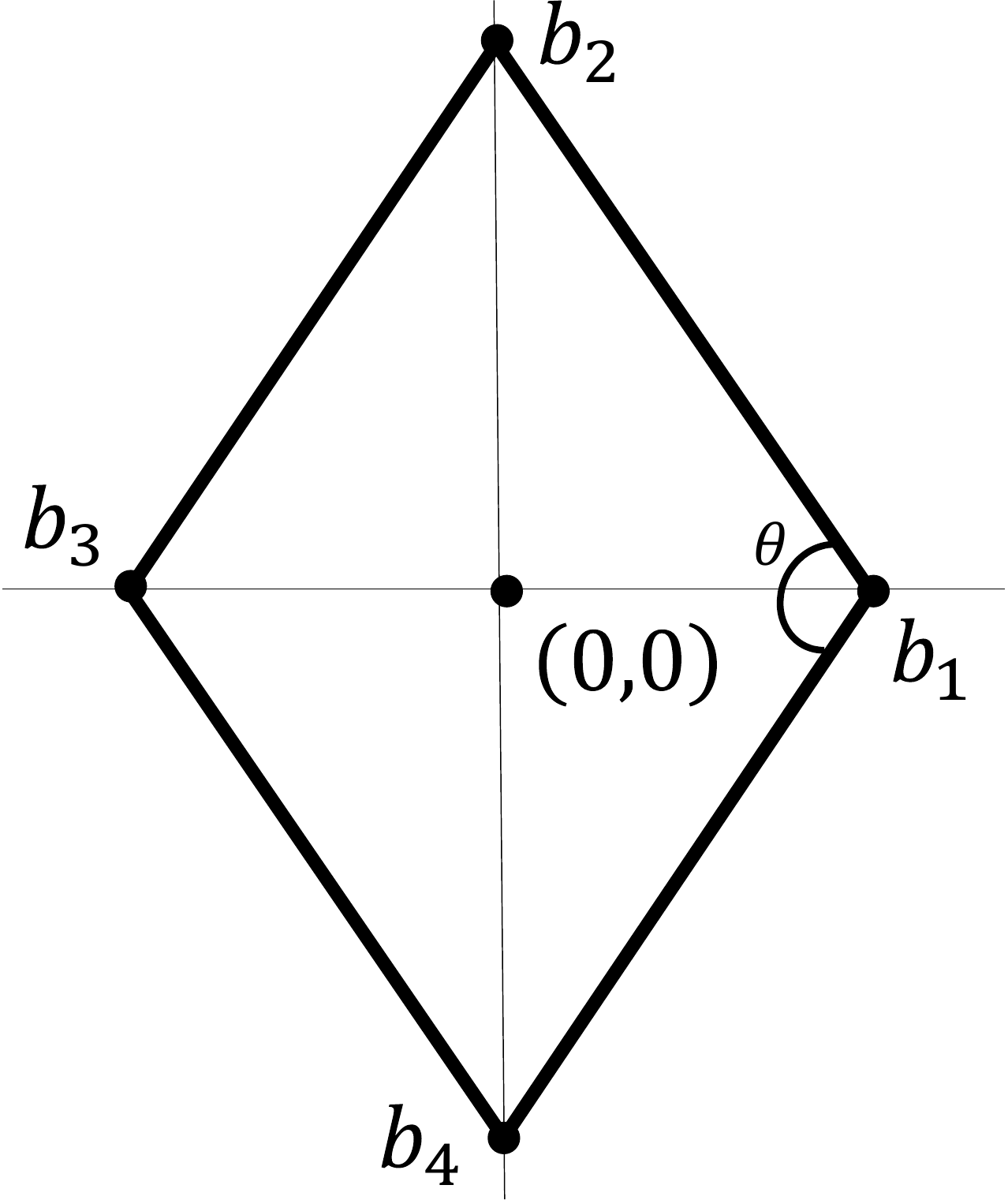}
        \caption{$\left\{\bY _t\right\} _{t\geq 0}$ in $\rhombus$}
        \label{fig:B}
    \end{subfigure}
    ~ 
    \begin{subfigure}[b]{0.31\textwidth}
        \includegraphics[width=\textwidth]{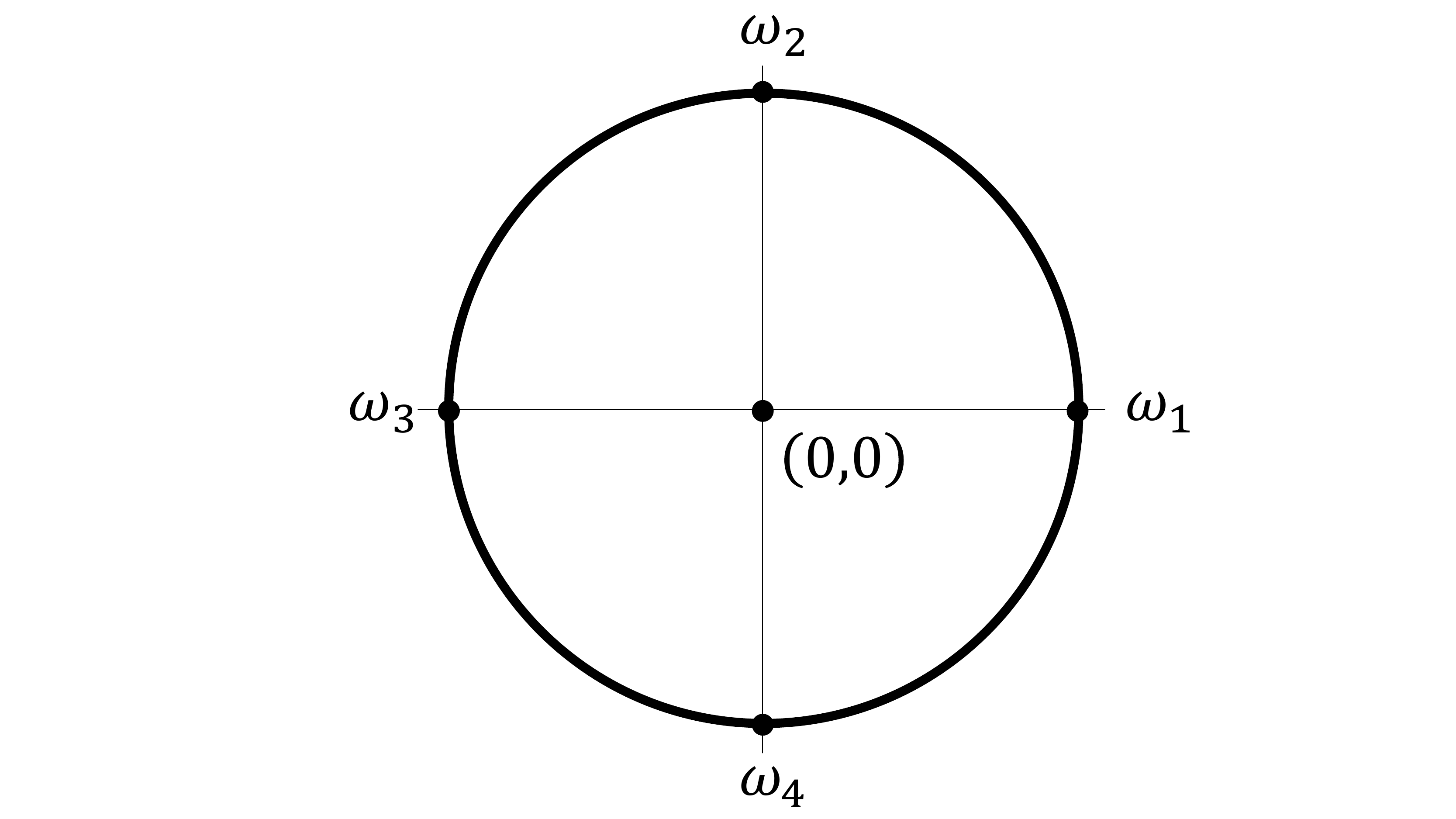}
        \caption{$\left\{ \bB_t\right\} _{t\geq 0}$ in $\mathbb{D}$}
        \label{fig:C}
    \end{subfigure}
    \caption{Figure (a) depicts $\left\{ \bX_t\right\}_{t\geq 0}$ in the $[-1,1]^2$ square, Figure (b) depicts $\left\{ \bY_t\right\}_{t\geq 0}$ in the rhombus $ \rhombus$, and Figure (c) depicts $\left\{ \bB_t\right\}_{t\geq 0}$ in the unit disc $\mathbb{D}$. The linear transformation $\bW^{-\nicefrac[]{1}{2}}$ transforms the $[-1,1]^2$ square to $\rhombus$ (Figure (a) to Figure (b)), whereas the conformal mapping $F_{\theta}$ transforms $\mathbb{D}$ to $\rhombus$ (Figure (c) to Figure (b)).}\label{fig:Proof}
\end{figure}

To obtain the harmonic measure directly for the rhombus $\rhombus$ we
appeal to conformal mappings. We use the fact that the harmonic
measure can be defined for any simply connected region $U$ in the
plane with $0$ in its interior. More precisely, let $\bB_t$ be
standard $2$-dimensional Brownian motion started at $0$, and $\tau(U)
= \inf\{t: \bB_t \not \in U\}$ be the first time it hits the boundary
of $U$. Then $\mu_{\bd U}$ denotes the probability measure induced by
the distribution of $\bB_{\tau(U)}$, and is called the
harmonic measure on $\bd U$ (with respect to $0$). When $U$ is the
unit disc centered at $0$, the harmonic measure is uniform on its
boundary because Brownian motion is invariant under rotation. Then the main idea is to use conformal maps to relate harmonic measures
on the different domains, namely the disc and our rhombus $\rhombus$.

\subsection{Conformal Mapping}\label{sec:schwarz-christoffel}

Before we proceed further, it is best to transition to the language
of complex numbers and identify $\mathbb{R}^2$ with the complex plane
$\mathbb{C}$.  A complex function $F: U \to V$ where $U,V \subseteq
\mathbb{C}$ is conformal if it is holomorphic (i.e.~complex differentiable) and its derivative
$f'(x) \neq 0$ for all $x \in U$.  The key fact we use about conformal maps
is that they preserve harmonic measure.
 Below we
present this theorem from~\citet{MP10} specialized to our setting. In
what follows, $\mathbb{D}$ will be the unit disc in $\mathbb{C}$
centered at $0$.
\begin{theorem}\label{thm:conformal-invariant} \cite[p.~204, Theorem 7.23]{MP10}.
    Suppose $F_{\theta}$ is a conformal map from the unit disk
    $\mathbb{D}$  to  $\rhombus$.
    Let $\mu_{\partial \mathbb{D}}$ and $\mu_{\partial \rhombus}$ be the harmonic measures
    with respect to $0$.
    Then $\mu_{\partial \mathbb{D}} \circ F_{\theta}^{-1} = \mu_{\partial \rhombus}$.
\end{theorem}
Thus the above theorem implies that in our setting,
the probability that a standard Brownian motion will first hit any segment $S$ of the boundary
of $\mathbb{D}$ is the same as the probability of the standard
Brownian motion first hitting
its image under $F_{\theta}$, i.e.~$F_{\theta}(S)$ in $\bd \rhombus$.

To complete the picture, the Schwarz-Christoffel formula gives a conformal
mapping from the unit disc $\mathbb{D}$ to $\rhombus$ that we utilize.

\begin{theorem}\cite[Theorem 5, Section 2.2.2]{A96}
  \label{thm:conformal-a96}
  Define the function $F_{\theta}(\omega)$ by
  \[ F_\theta(\omega) = \int_{s=0}^\omega f_\theta(s) ds
  = \int_{s=0}^{\omega} (1-s)^{-(1-\theta/\pi)}
  (1+s)^{-(1-\theta/\pi)} (s-i)^{-\theta/\pi} (s+i)^{-\theta/\pi}
  ds.\]
  Then, for some real number $c > 0$, $c F_\theta(\omega)$ is a
  conformal map from the unit-disk $\mathbb{D}$ to the rhombus $\rhombus$.
\end{theorem}
The conformal map has some important properties which will aid us in
calculating the probabilities. We collect them in the following lemma,
which follows from standard properties of the Schwarz-Christoffel
integral~\cite{A96}, and is easily verified.
\begin{lemma}
\label{lem:conformal-properties}
The conformal map $c F_{\theta}(\omega)$  has the following properties:
\begin{enumerate}
    \item The four points located at $\set{1,i,-1,-i}$ map to the four
      vertices $\{b_1,\dots,b_4\}$ of the rhombus $\rhombus$, respectively.
    \item The origin maps to the origin.
    \item The boundary of the unit-disk $\mathbb{D}$ maps to the boundary of $\rhombus$. Furthermore,
        the points in the arc from $1$ to $i$ map to the segment $[b_1,b_2]$.
\end{enumerate}
\end{lemma}
Define the function $r:[0,\pi/2]\rightarrow \mathbb{R}$ as
     $r(\phi) \eqdef \abs{F_{\theta}(e^{i \phi}) - F_{\theta}(1)}$.
\begin{lemma}\label{lem:prob-cut}
    The probability that vertices $\{i,j\}$ are separated, given that
    the angle between $\bw_i$ and $\bw_j$ is $\theta$, is
    \[ \frac{2}{\pi} \int_{0}^{\pi/2} 1 - \frac{r(\phi)}{r(\pi/2)} d\phi \]
\end{lemma}
\begin{proof}

  Rewriting the expression in~\Cref{lem:rhombus-symmetry} in complex
  number notation, we have
  \[
  \Pr[\{i,j\} \text{ separated}] =
  4 \cdot \int_{b_1}^{b_2}1- \frac{ \abs{z-b_1}}{ \abs{b_2- b_1}}
  d\mu_{\partial \rhombus}(z)
  =
  4 \cdot \int_{b_1}^{b_2} 1- \frac{ \abs{z-cF_{\theta}(1)}}{
  c\abs{F_\theta(i)- F_{\theta}(1)}}  d\mu_{\partial \mathbb{S}}(z).
  \]

  Since the conformal map $F_{\theta}$ preserves the harmonic measure
  between the rhombus $\mathbb{S}$ and the unit-disk $\mathbb{D}$ (see~Theorem~\ref{thm:conformal-invariant})
  and by \Cref{lem:conformal-properties}, the segment from $b_1$ to $b_2$ is
  the image of the arc from $1$ to $i$ under $cF_\theta$, we can rewrite
  the above as 
  \[ = 4 \cdot \int_{0}^{\pi/2} 1- \frac{ \abs{cF_{\theta}(e^{i\phi})-cF_{\theta}(1)}}{
  c\abs{F_\theta(i)- F_{\theta}(1)}}  d\mu_{\partial \mathbb{D}}(e^{i\phi}). \] 
  
  The harmonic measure $\mu_{\bd \mathbb{D}}$ on the unit-disk is
  uniform due to the rotational symmetry of Brownian
  motion. 
   \[ = 4 \cdot \int_{0}^{\pi/2} 1- \frac{ \abs{cF_{\theta}(e^{i\phi})-cF_{\theta}(1)}}{
   c\abs{F_\theta(i)- F_{\theta}(1)}} \frac{ d \phi}{2\pi} . \]

  Simplifying the above, we see that the right hand side above equals
  \[
  \frac{2}{\pi} \cdot \int_{0}^{\pi/2}
  1-\frac{\abs{F_\theta(e^{i\phi})-F(1)}}{\abs{F_\theta(i)- F(1)}} d\phi
  = \frac{2}{\pi} \cdot
  \int_{0}^{\pi/2} 1- \frac{r(\phi)}{r(\pi/2)}  d\phi.
  \]
  This completes the proof.
\end{proof}

To calculate the approximation ratio exactly, we will make use of the theory
of special functions. While these calculations are technical, they are not trivial. To
aid the reader, we give a brief primer in Appendix~\ref{sec:primer} and refer them to
the work of \citet{BW10,andrews1999special} for a more thorough introduction.

The proof of~\Cref{thm:cut-prob}, will follow from the following key claims
whose proofs appear in the appendix. Letting $a = \theta/\pi$ and $b=1-a$, we have

\begin{restatable}{claim}{rphiclm}
\label{l:rphi}
\[ r(\phi) = \frac{1}{4} \beta_{\sin^2 \phi} (a/2,b/2)\] when $\phi \in [0,\pi/2]$.
\end{restatable}

\begin{restatable}{claim}{rphiintclm}
    \label{lem:rphiint}
    \[ 4 \cdot \int_0^{\pi/2} r(\phi) d\phi
            = \frac{\beta(a/2+1/2,1/2)}{2a}  \cdot
    \pFq{3}{2}{\frac{1+a}{2},\frac{1+a}{2},\frac{a}{2}}{\frac{a}{2},\frac{a}{2}+1}{1} \]
\end{restatable}

\subsection{Asymptotic Calculation for \texorpdfstring{$\theta$}{theta} close to \texorpdfstring{$\pi$}{pi}.}

We consider the case when the angle $\theta=(1-\epsilon)\cdot \pi$ as $\epsilon \to 0$.
The hyperplane-rounding algorithm separates such an edge by $\theta/\pi$, and
hence has a separation probability of $1-\epsilon$.  We show a similar
asymptotic behaviour for the Brownian rounding algorithm, albeit with slightly worse constants. We defer the proof to the appendix.

\begin{restatable}{theorem}{piepsthm}
\label{th:pi}
    Given an edge $\{i,j\}$ with
    $cos^{-1}(\bw_i^T \bw_j)= \theta = (1-\epsilon) \pi$,
    the Sticky Brownian Motion rounding will cut the edge
    with probability at least
    $1-\left(\frac{4}{\pi} \epsilon +O(\epsilon^2)\right)$.
\end{restatable}

\section{Brownian Rounding for \MtwoSAT via Partial Differential Equations}\label{sec:pde}

In this section we use \MtwoSAT as a case study for extending the
Sticky Brownian Motion rounding method to other constraint satisfaction problems besides \MC.
In the \MtwoSAT problem we are given $n$ variables $z_1,\ldots,z_n$ and $m$ clauses $C_1,\ldots,C_m$, where the $j$\textsuperscript{th} clause is of the form $y_{j_1}\vee y_{j_2}$ ($y_j$ is a literal of $z_j$, {\em i.e.}, $z_j$ or $\overline{z}_j$).
The goal is to assign to each variable $z_i$ a value of true or false
so as to maximize the number of satisfied clauses.

\subsection{Semi-definite Relaxation and Brownian Rounding Algorithm}
The standard SDP relaxation used for \MtwoSAT is the following:
\begin{align}
\max ~~~& \sum _{j=1}^m \left( 1-\bv _{j_1}\cdot \bv _{j_2}\right)& \nonumber \\
s.t. ~~~& \bv _0 \cdot \bv _0 = 1 & \label{2SAT:Const1}\\
& \bv _0 \cdot \bv _i = \bv _i \cdot  \bv _i & \forall i=-n,\ldots,n \label{2SAT:Const2}\\
& \bv _i \cdot  \bv _{-i} = 0 & \forall i=1,\ldots,n \label{2SAT:Const3}\\
& \bv _0 \cdot (\bv _i +\bv _{-i}) = 1 & \forall i=1,\ldots,n \label{2SAT:Const4}\\
& 1 \geq \bv _0 \cdot \bv _i + \bv _j \cdot \bv _0 - \bv _i \cdot \bv _j & \forall i,j=-n,\ldots,n \label{2SAT:Const5}\\
& \bv _i \cdot \bv _0 \geq \bv _i \cdot \bv _j & \forall i,j=-n,\ldots,n \label{2SAT:Const6}\\
& \bv _i \cdot \bv_j \geq 0 & \forall i,j=-n.\ldots,n \label{2SAT:Const7}
\end{align}
In the above $\bv _0$ is a unit vector that denotes the {\em false} assignment (constraint \ref{2SAT:Const1}), whereas a zero vector denotes the {\em true} assignment.
We use the standard notation that $\bv_i$ denotes the literal $z_i$ and $\bv _{-i}$ denotes the literal $\overline{z}_i$.
Therefore, $\bv _i\cdot \bv _{-i} =0$ for every $i=1,\ldots.n$ (constraints \ref{2SAT:Const3} and \ref{2SAT:Const4}) since $z_i$ needs to be either true or false.
The remainder of the constraints (constraints \ref{2SAT:Const5},
\ref{2SAT:Const6} and \ref{2SAT:Const7}) are equivalent to the $\ell
_2^2$ triangle inequalities over all triples of vectors that include $\bv _0$.

When trying to generalize the Brownian rounding algorithm for \MC
presented in Section \ref{sec:conformal} to \MtwoSAT, there is a
problem: unlike \MC the \MtwoSAT problem is {\em not} symmetric.
Specifically, for \MC both $S$ and $\overline{S}$ are equivalent solutions having the same objective value.
However, for \MtwoSAT an assignment to the variables $z_1=\alpha_1,\ldots,z_n=\alpha_n$ is {\em not} equivalent to the assignment $z_1=\overline{\alpha}_1,\ldots,z_n=\overline{\alpha}_n$ (here $\alpha_i\in \{ 0,1\}$ and $\overline{\alpha}_i=1\oplus \alpha _i$).
For example, if $\bv _i\cdot \bv _0 = 1$ then we would like the
Brownian rounding algorithm to always assign $z_i$ to false.
The Brownian rounding for \MC cannot handle such a requirement.
In order to tackle the above problem we incorporate $\bv_0$ into both the starting point of the Brownian motion and the covariance matrix.

Let us now formally define the Brownian rounding algorithm for \MtwoSAT.
For simplicity of presentation denote for every $i=1,\ldots,n$ by
$x_i$ the {\em marginal value} of $z_i$, formally: $ x_i \eqdef \bv _i \cdot \bv _0$.
Additionally, let $\bw _i$ be the (unique) unit vector in the direction of
the projection of $\bv_i$ to the subspace orthogonal  to $\bv _0$, {\em i.e.}, $\bw _i$ satisfies $\bv _i = x_i \bv _0 + \sqrt{x_i-x_i^2}\bw _i$.\footnote{It is easy to see that $x_{-i}=1-x_i$ and $\bw _{-i}=-\bw _i$ for every $i=1,\ldots,n$.}
Similarly to \MC, our Sticky Brownian Motion rounding algorithm performs a random walk in $\mathbb{R}^n$, where the $i$\textsuperscript{th} coordinate corresponds to the variable $z_i$.
For simplicity of presentation, the random walk is defined in $[0,1]^n$ as opposed to $[\pm 1]^n$, where $1$ denotes false and $0$ denotes true.\footnote{We note that the Brownian rounding algorithm for \MtwoSAT can be equivalently defined in $[-1,1]^n$, however, this will incur some overhead in the notations which we would like to avoid.}
Unlike \MC, the starting point $\bX_0$ is not the center of the
cube. Instead, we use the marginals, and set $(\bX_0)_i \eqdef x_i$.
The covariance matrix $\bW$ is defined by $\bW _{i,j}\eqdef \bw _i
\cdot \bw _j$ for every $i,j=1,\ldots,n$, and similarly to \MC, let
$\bW ^{\nicefrac[]{1}{2}}$ be the principle square root of $\bW$.
Letting $ \{ \bB _t\} _{t\geq 0}$ denote standard Brownian motion in
$\mathbb{R}^n$, we define $\tau_1 = \inf\{t: \bW ^{\nicefrac[]{1}{2}}\bB
_t + \bX_0 \not \in [0,1]^n\}$ to be the first time the process hits
the boundary of $[0,1]^n$. Then, for all times $0 \le t \le \tau_1$,
the process $\bX_t$ is defined as
\[
\bX _t = \bW ^{\nicefrac[]{1}{2}}\bB _t + \bX_0.
\]
After time $\tau_1$, we force $\bX_t$ to stick to the face $F_1$ hit
at time $\tau_1$: i.e.~if $(\bX_{\tau_1})_i \in \{0,1\}$, then we fix
it forever, by zeroing out the $i$-th row and column of the
covariance matrix of $\bW$ for all future time steps.
The rest of the process is defined analogously to
the one for \MC: whenever $\bX_t$ hits a lower dimensional face of
$[0,1]^n$, it is forced to stick to it until finally a vertex is
reached, at which point $\bX_t$ stops changing. We use $\tau_i$ for
the first time that $\bX_t$ hits a face of dimension $n-i$; then,
$\bX_{\tau_n} \in \{0,1\}^n$.

The output of the algorithm corresponds to the collection of the
variables assigned a value of true $T\subseteq \{ 1,\ldots,n\}$:
\[
T=\{ i:  \left(\bX _{\tau_n}\right) _i = 0\},
\]
whereas implicitly the collection of variables assigned a value of false are $\{ i:  \left(\bX _{\tau_n}\right) _i = 1\}$.

\subsection{Analysis of the Algorithm}
\label{sec:sde2pde}

Our goal is to analyze the expected value of the assignment produced
by the Sticky Brownian Motion rounding algorithm.
Similarly to previous work, we aim to give a lower bound on the
probability that a fixed clause $C$ is satisfied.
Unfortunately, the conformal mapping approach described in Section
\ref{sec:conformal} does not seem to be easily applicable to the
extended Sticky Brownian Motion rounding described above for \MtwoSAT,
because our calculations for \MC relied heavily on the symmetry of the
starting point of the random walk.
We propose a different method for analyzing the Brownian rounding
algorithm that is based on partial differential equations and the
maximum principle. We prove analytically the following theorem which
gives a guarantee on the performance of the algorithm.
We also note that numerical calculations show that the algorithm in
fact achieves the better approximation ratio of $0.921$ (see
Section~\ref{sec:slowdown} for details).

\begin{theorem}
The Sticky Brownian Motion rounding algorithm for \MtwoSAT achieves an approximation of at least $0.8749$.
\end{theorem}

\subsubsection{Analysis via Partial Differential Equations and Maximum
  Principle}

As mentioned above, our analysis focuses on the probability that a
single clause $C$ with variables $\{z_i,z_j\}$ is satisfied.  We
assume the variables are not negated.  This is without loss of
generality as the algorithm and analysis are invariant to the sign of
the variable in the clause.

For simplicity of notation we denote by $x$ the marginal value of $z_i$ and by $y$ the marginal value of $z_j$.
Thus, $\bv_i=x \bv_0 + \sqrt{x-x^2} \bw_i$ and $\bv_j= y\bv_0 + \sqrt{y-y^2} \bw_j$.
Projecting the random process $\set{\bX}_{t\ge 0}$ on the $i$ and $j$ coordinates of the random process, we obtain a new process $\set{\tilde{\bX}_t}_{t\ge 0}$ where $\tilde{\bX}_0=(x,y)$. Let
\[ \tilde{\bW}=\begin{pmatrix} 1 & \cos (\theta)
            \\ \cos (\theta) & 1 \end{pmatrix},
        \] where $\theta$ is the angle between $\bw _i$ and
$\bw _j$.
Then $\tilde{\bX}_t = \tilde{\bX}_0 + \tilde{\bW}^{\nicefrac[]{1}{2}} \bB _t$ for all
$0 \le t \le \tau$, where $\tau = \inf\{t: \tilde{\bX}_0 +
\tilde{\bW}^{\nicefrac[]{1}{2}} \bB _t \not \in [0,1]^2\}$ is the
first time the process hits the boundary of the square. After time
$\tau$, the process $\tilde{\bX}_t$ performs a one-dimensional standard
Brownian motion on the first side of the square it has hit, until it
hits a vertex at some time $\sigma$. After time $\sigma$ the process
stays fixed. Almost surely $\sigma < \infty$, and, moreover, it is
easy to show that $\E \sigma < \infty$. We say that
$\set{\tilde{\bX}_t}_{t \ge 0}$
is absorbed at $\tilde{\bX}_\sigma \in \{0,1\}^2$.

\begin{observation}\label{Obs:absorption}
    The probability that the algorithm satisfies the clause $\{z_i,z_j\}$ equals \[\Pr \left[  \tilde{\bX}_\sigma \text{ is absorbed in } \set{(0,0),(0,1),(1,0)} \right].\]
\end{observation}
We abuse notation slightly and denote $\tilde{\bX_t}$ by $\bX_t$ and
$\tilde{\bW}$ by $\bW$ for the rest of the section which is aimed at analyzing
the above probability.  We also denote $\rho = \cos (\theta)$.  

Our next step is fixing $\theta$ and analyzing the probability of
satisfying the clause for all possible values of marginals $x$ and
$y$. Indeed, for different $x$ and $y$ but the same $\theta$, the
analysis only needs to consider the same random process with a
different starting point. Observe that not all such $x,y$ are
necessarily feasible for the SDP: we characterize which ones are feasible for a
given $\theta$ in\ Lemma~\ref{lem:feasibility}. But considering all $x,y$ allows us to
handle the probability in Observation~\ref{Obs:absorption} analytically.

For any $0\leq x\leq 1, 0\leq y\leq 1$, let $u(x,y)$ denote the probability of starting the random walk at the point $(x,y)$ and ending at one of the corners $(0,0)$, $(0,1)$ or $(1,0)$.  This
captures the probability of a clause being satisfied  when the walk begins with
marginals $(x,y)$ (and angle $\theta$). We can easily calculate this
probability exactly when either $x$ or $y$ are in the set $\set{0,1}$.
We obtain the following easy lemma whose proof appears in the appendix.

\begin{restatable}{lemma}{boundary}\label{lem:boundary}
For $\phi(x,y) = 1-xy$, we have
\vspace{-0.7em}
\begin{align}
    u(\bx) &= \phi(\bx) \text{ ~~~for all }\bx \in \bd[0,1]^2 \label{eqn:boundary-max-cut}
\end{align}
Moreover, for all $\bx$ in the interior of the square $[0,1]^2$,
$u(\bx) = \E^\bx[\phi(\bX_\tau)]$, where $\E^\bx$ denotes expectation
with respect to starting the process at $\bX_0 = \bx$.
\end{restatable}
Next we use the fact that Brownian motion gives a solution to the
Dirichlet boundary problem. While Brownian motion gives a solution to
Laplace's equation (\cite{MP10} chapter 3), since our random process
is a diffusion process, we need a slightly more general
result\footnote{This result can also be derived from Theorem\ 3.12 in
  \cite{MP10} after applying a linear transformation to the
  variables.}. We state the following result from \cite{Oksendal-SDE},
specialized to our setting, that basically states that given a
diffusion process in $[0,1]^2$ and a function $\phi$ on the boundary,
the extension of the function defined on the interior by the expected value of the function at the first hitting point on the boundary is characterized by an elliptic partial differential equation. 

\begin{theorem}[\cite{Oksendal-SDE} Theorem 9.2.14]
    \label{thm:oksendal-max-principle}
Let $D=(0,1)^2\subseteq \mathbb{R}^2$,  $\mathbf{\Sigma}\in  {{\Re^{2\times 2}}}$ and let $a_{11},a_{12},a_{21},a_{22}$ be defined as follows
$$ \begin{pmatrix} a_{11} & a_{12} \\
a_{21}  & a_{22} \end{pmatrix} = \frac12 \mathbf{\Sigma}\mathbf{\Sigma^\top}. $$
For any $\bx\in D$, consider the process $\bX _t = \bX_0 +
\mathbf{\Sigma}\bB _t$ where $\bB_t$ is standard Brownian motion in $\mathbf{R}^2$.
Let $\tau = \inf\{t: \bX_t \not \in D\}$.
Given a bounded continuous function $\phi:\partial D\rightarrow \Re$, define the function $u:D\rightarrow \Re$ such that
  $$ u(\bx)=\mathbb{E}^{\bx} \left[ \phi (\bX _{\tau })\right],$$
where $\E^{\bx}$ denotes the expected value when $\bX_0 = \bx \in \Re^2$.
I.e.,  $u(\bx)$ is the expected value of $\phi$ when first hitting $\partial D$ conditioned on starting at point $\bx$.
Consider the uniformly elliptic partial differential operator $\L$ in $D$ defined by: $$ \L= \sum _{i,j=1}^2 a_{ij}\frac{\partial ^2}{\partial x_i \partial x_j}.$$
Then $u\in C^2(D)$ is the unique solution to the partial differential equation\footnote{$u\in C^k(D) $ means that $u$ has a continuous $k$\textsuperscript{th} derivative over $D$, and $u\in C^0(D)$ means that $u$ is continuous.}:
\begin{align*}
    \L u&=0 && \text{ in } D&  \\
    \lim_{ \substack{\bx\rightarrow \by \\  \bx\in D} } u(\bx)&=\phi(\by) && \text{ for all } \by\in \partial D&
\end{align*}
\end{theorem}

We instantiate our differential equation by choosing $\mathbf{\Sigma}=\bW^{\nicefrac[]{1}{2}}$  and thus $a_{ij}$ are the entries of $\bW$.
It is important to note that all $a_{ij}$s are independent of the starting point $\bx\in [0,1]^2$. 
Thus, we obtain that $u$ is the unique function satisfying the following partial differential equation:
\begin{align*}
\frac{\partial ^2 u}{\partial x^2} + \frac{\partial ^2 u}{\partial y^2} + 2\rho \frac{\partial ^2 u}{\partial x \partial y} & = 0 & \forall (x,y)\in \Int[0,1]^2\\
u(x,y) & = (1-xy) & \forall (x,y)\in \bd[0,1]^2
\end{align*}
Above, and in the rest of the paper, we use $\Int\,D$ to denote the
interior of a set $D$, and $\bd D$ to denote its boundary. 

It remains to solve the above partial differential equation (PDE) that will allow us to calculate $u(x,y)$ and give the probability of satisfying the clause.
\subsection{Maximum Principle}
\label{sec:max-principle}
Finding closed form solutions general PDE's is challenging and, there
is no guarantee any solution would be expressible in terms of simple
functions. However, to find a good approximation ratio, it suffices
for us to
find good lower-bounds on the probability of satisfying the
clause. I.e.~we need to give a lower bound on the function $u(x,y)$
from the previous section over those $(x,y)$ that are feasible.
Since the PDE's generated by our algorithm are elliptic (a
particular kind of PDE), we will use a property of elliptic PDE's
which will allow us to produce good lower-bounds on the solution at
any given point. More precisely, we use the following theorem from
\citet{GT15}.

Let $\L$ denote the operator
    \[ \L := \sum_{ij} a_{ij} \frac{\partial^2}{\partial_i\partial_j}\]
and we say that $\L$ is an elliptic operator if the coefficient matrix $A =
[a_{ij}]_{i,j}$ is positive semi-definite.

 We restate
a version of Theorem 3.1 in~\citet{GT15} that shows how the maximum principle can be used to obtain lower bounds on $u(x,y)$.
Here $\bar{D}$ denotes the closure of $D$.
\begin{theorem}[Maximum Principle]\label{thm:maximum}
    Let $\L$ be elliptic on a bounded domain $D$ and suppose
    $\L[g](x) \geq 0 \;\; \forall x \in D$ for some $g \in C^2(D) \cap C^0(\bar{D})$.
    Then the maximum of $g$ on $D$ is achieved on $\partial D$, that is,
    \[ \sup_{x \in D} g(x) = \sup_{x \in \partial D} g(x) \]
\end{theorem}

Theorem~\ref{thm:maximum} has the following corollary that allows us to obtain lower bounds on $u(x,y)$.
\begin{corollary}
    \label{cor:maximum}
    Let $\L$ be elliptic on a bounded domain $D$ and
    for some $u,g \in C^2(D) \cap C^0(\bar{D})$.
    \begin{enumerate}
        \item $\L[g](x) \geq \L[u](x) \quad \forall x \in D$
        \item $g(x) \leq u(x) \quad \forall x \in \partial D$
    \end{enumerate}
then $g(x) \leq u(x) \forall x \in D$.
\end{corollary}

We refer the reader to \cite{GT15} for a formal proof.
Thus, it is enough to construct candidate functions $g:[0,1]^2\rightarrow \Re$ such that

\begin{align}
 \frac{\partial ^2 g}{\partial x^2} + \frac{\partial ^2 g}{\partial y^2} + 2\rho \frac{\partial ^2 g}{\partial x \partial y} &\geq  0  &\forall (x,y)\in \Int[0,1]^2\\
g(x,y) &\leq (1-xy) &\forall (x,y)\in \partial[0,1]^2
\end{align}
Then we obtain that $g(x,y)\leq u(x,y)$ for all $(x,y)\in [0,1]^2$. In
what follows we construct many different such function each of which
works for a different range of the parameter $\theta$ (equivalently, $\rho$).

\subsection{Candidate Functions for Maximum Principle}

We now construct feasible candidates to the maximum principle as described in Corollary~\ref{cor:maximum}.
We define the following functions:

\begin{enumerate}
\item $g_1(x,y)=1-xy-\cos(\theta)\sqrt{x-x^2}\sqrt{y-y^2}$.
\item $g_2(x,y)=1-xy-2\cos(\theta)(x-x^2)(y-y^2)$.
\item $g_3(x,y)=1-xy-\frac12(1+5 \cos(\theta))(x-x^2)(y-y^2)(x+y)(2-x-y)$.
\end{enumerate}

The following lemma shows that the above functions satisfy the conditions required for the application of the maximum principle (its proof appears in the appendix).
\begin{restatable}{lemma}{feasiblemax}\label{lem:feasible_maximum}
Each of $g_1, g_2,g_3$ satisfies the boundary conditions, i.e. $g_i(x,y)=u(x,y)$ for all $x,y\in \partial[0,1]^2$ and for all values $\theta$. Moreover, we have the following for each $(x,y) \in [0,1]^2$:
\begin{enumerate}
\item If $1\geq \cos(\theta) \geq 0$, then $\L g_1\geq 0$.
\item If $0\geq \cos(\theta) \geq -\frac12$, then $\L g_2\geq 0$.
\item If $-\frac12\geq \cos(\theta) \geq -1$, then $\L g_3\geq 0$.
\end{enumerate}
\end{restatable}
While some of these proofs are based on simple inequalities, proving
others requires us to use sum of squares expressions. For example, to
show $\L g_3\geq 0$, we consider $\L g_3=p(x,y,\cos(\theta))$ as a
polynomial in $x,y$ and $\cos(\theta)$. Replacing $z=\cos (\theta)$,
our aim is to show $p(x,y,z)\geq 0$ if $0\leq x,y\leq 1$ and
$-1\leq z\leq -\frac12$. Equivalently, we need to show $p(x,y,z)\geq
0$ whenever $r_1(x,y,z)\eqdef x-x^2\geq 0$, $r_2(x,y,z)\eqdef y-y^2\geq 0$ and
$r_3(x,y,z)\eqdef -(z+\frac12) \geq 0 $ and $r_4(x,y,z)\eqdef (z+1)\geq 0$. We
show this by obtaining polynomials $q_i(x,y,z)$ for $i=0,1,2,3,4$ such
that each $q_i$ is a sum of squares polynomial of fixed degree and we have
$$p(x,y,z)=q_0(x,y,z)+\sum_{i=1}^4 q_i(x,y,z)r_i(x,y,z).$$
Observe that the above polynomial equality proves the desired result by evaluating the RHS for every $0\leq x,y\leq 1$ and $ -\nicefrac[]{1}{2}\geq z \geq -1$.
Clearly, the RHS is non-negative: each $q_i$ is non-negative since it
is a sum of squares and each $r_i$ is non-negative in the region we
care about, by construction.
We mention that we obtain these proofs via solving a semi-definite
program of fixed degree (at most 6) for each of the $q_i$ polynomials (missing details appear in the appendix).

Let us now focus on the approximation guarantee that can be proved using the above functions $g_1$, $g_2$, and $g_3$.
The following lemma compares the lower bounds on the probability of satisfying a clause, as given by $g_1$, $g_2$, and $g_3$, to the SDP objective.
Recall that the contribution of any clause with marginals $x$ and $y$ and angle $\theta$ to the SDP's objective is given by: $1-xy-\cos(\theta)\sqrt{x-x^2}\sqrt{y-y^2}$.
We denote this contribution by $\SDP(x,y,\theta)$.
It is important to note that not all triples $(x,y,\theta)$ are
feasible (recall that $\theta$ is the angle between $\bw _i$ and $\bw
_j$), due to the triangle inequalities in the SDP.
This is summarized in the following lemma.

\begin{restatable}{lemma}{feasibility}\label{lem:feasibility}
Let $x,y,\theta$ be as defined by a feasible pair of vectors $v_i$ and $v_j$. Then they must satisfy the following constraints:

\begin{enumerate}
\item $0\leq x\leq 1, 0\leq y\leq 1, 0\leq \theta \leq \pi$.
\item $\cos(\theta)\geq -\sqrt{\frac{xy}{(1-x)(1-y)}}$.
\item $\cos(\theta)\geq -\sqrt{\frac{(1-x)(1-y)}{xy}}$.
\end{enumerate}
\end{restatable}

Finally, we prove the following lemma which proves an approximation guarantee of $\MtwoRatio$ for \MtwoSAT via the PDE and the maximum principle approach.
As before, these proofs rely on explicitly obtaining sum of squares proofs as discussed above. We remark that these proofs essentially aim to obtain $\frac78=0.875$-approximation but errors of the order $10^{-5}$ allow us to obtain a slightly worse bound using this methods. The details appear in the appendix.

\begin{restatable}{lemma}{ratiomax}\label{lem:ratio_maximum}
Consider any feasible triple $(x,y,\theta)$ satisfying the condition in Lemma~\ref{lem:feasibility}. We have the following.
\begin{enumerate}
\item If $1\geq \cos(\theta) \geq 0$, then $g_1(x,y)\geq 1\cdot \SDP(x,y,\theta)$.
\item If $0\geq \cos(\theta) \geq -\frac12$, then $g_2(x,y)\geq \MtwoRatio\cdot \SDP(x,y,\theta)$.
\item If $-\frac12\geq \cos(\theta) \geq -1$, then $g_3(x,y)\geq \MtwoRatio\cdot \SDP(x,y,\theta)$.
\end{enumerate}
\end{restatable}

\section{\MC with Side Constraints (\cardMC)}
\label{sect:mc-constraints}
\newcommand{\SoS}{\mathrm{SoS}}
\newcommand{\SSsoln}{\mathcal{V}}

In this section we describe how to apply the Sticky Brownian Motion
rounding and the framework of Raghavendra and
Tan~\cite{RaghavendraT12} to the \cardMC
problem in order to give a bi-criteria approximation algorithm whose
running time is non-trivial even when the the number of constraints is
large.

\subsection{Problem Definition and Basics}

Let us recall the relevant notation and definitions.  An instance of
the \cardMC problem is given by
an $n$-vertex graph $G = (V, E)$ with edge weights $a: E \to \Re_+$,
as well as a collection $\FF = \{F_1, \ldots, F_k\}$ of subsets of
$V$, and cardinality bounds $b_1, \ldots, b_k \in \N$. For ease of
notation, we will assume that $V = \{1, \ldots, n\}$. Moreover, we
denote the total edge weight by $a(E) = \sum_{e \in E}{a(e)}$. The
goal in the \cardMC problem is to
find a subset $S \subset V$ that maximizes the weight $a(\delta(S))$
of edges crossing the cut $(S, V\setminus S)$, subject to having $|S
\cap F_i| = b_i$ for all $i \in [k]$.  These cardinality constraints
may not be simultaneously satisfiable, and moreover, when $k$ grows
with $n$, checking satisfiability is
$\mathsf{NP}$-hard~\cite{dischard}. For these reasons, we allow for
approximately feasible solutions. We will say that a set of vertices
$S \subseteq V$ is an $(\alpha, \eps)$-approximation to the \cardMC
problem if $\bigl||S \cap F_i| - b_i \bigr| \le \eps n$ for all $i \in
[k]$, and $a(\delta(S)) \ge \alpha \cdot a(\delta(T))$ for all $T\subset V$
such that $|T \cap F_i| = b_i\ $ for all $i \in [k]$.  In the
remainder of this section we assume that the instance given by $G$,
$\FF$, and $b$ is satisfiable, i.e.~that there exists a set of
vertices $T$ such that $|T \cap F_i| = b_i\ $ for all $i \in [k]$. Our
algorithm may fail if this assumption is not satisfied. If this
happens, then the algorithm will certify that the instance was not
satisfiable.

We start with a simple baseline approximation algorithm, based on
independent rounding. The algorithm outputs an approximately feasible
solution which cuts a constant fraction of the total edge weight. For
this reason, it achieves a good bi-criteria approximation when
the value of the optimal solution $\OPT$ is much smaller than $\eps
a(E)$. This allows us to focus on the case in which $\OPT$ is bigger
than $\eps a(E)$ for our main rounding algorithm. The proof of the
lemma, which follows from standard arguments, appears in the
appendix. 

\begin{restatable}{lemma}{simpleapprox}
  \label{lm:simple-approx}
  Suppose that $n \ge \frac{2\ln(8k/\eps)}{\eps^2}$ and $\eps \le \frac12$. 
  There exists a polynomial time algorithm that on input a satisfiable
  instance $G = (V, E)$, $\FF$, and $b_1, \ldots, b_k$, as defined
  above, outputs a set $S \subseteq V$ such that, with high
  probability, $a(\delta(S)) \ge \frac{\eps}{2} a(E)$, and $\bigl||S \cap F_i|
  - b_i \bigr| \le \eps n$ for all $i \in [k]$. 
\end{restatable}

\subsection{Sum of Squares Relaxation}

Our main approximation algorithm is based on a semidefinite
relaxation, and the sticky Brownian motion. Let us suppose that we are
given the optimal objective value $\OPT$ of a feasible solution:
this assumption can be removed by doing binary search for $\OPT$. We
can then model the problem of finding an optimal feasible solution by
the quadratic program
\begin{align*}
  &\sum_{e=(i,j) \in E}{a(e)(x_i - x_j)^2} \ge \OPT\\
  \text{s.t.} 
  &~~~~\sum_{j \in F_i}{x_j} = b_i &\forall i = 1, \ldots, k\\
  &x_j(1-x_j) = 0 &\forall j = 1, \ldots, k
\end{align*}
Let us denote this quadratic feasibility problem by $Q$. The Sum of
Squares (Lasserre) hierarchy gives a semidefinite program that relaxes
$Q$. We denote by $\SoS_\ell(Q)$ the solutions to the level-$\ell$ Sum
of Squares relaxations of $Q$. Any solution in $\SoS_\ell(Q)$ can be
represented as a collection of vectors $\SSsoln = \{\bv_S: S \subseteq
[n], 0 \le |S| \le \ell\}$. To avoid overly cluttered notation, we
write $\bv_i$ for $\bv_{\{i\}}$; we also write $\bv_0$ for
$\bv_\emptyset$. We need the following properties of $\SSsoln$, valid
as long as $\ell \ge 2$. 
\begin{enumerate}

\item $\bv_0 \cdot \bv_0 = 1$.
\item $\bv_S \cdot \bv_T = \bv_{S'} \cdot \bv_{T'}$ for any
  $S,S',T,T'$ such that $S\cup T = S'\cup T'$ and $|S\cup T| \le
  k$. In particular, $\bv_i\cdot \bv_i = \bv_i \cdot \bv_0$ for any
  $i$. 

\item For any $i$ and $j$ the following inequalities hold:
  \begin{align}
    & 1 \geq \bv _0 \cdot \bv _i + \bv _j \cdot \bv _0 - \bv _i \cdot \bv _j \label{mc-mult-tri1}\\
    & \bv _i \cdot \bv _0 \geq \bv _i \cdot \bv _j  \label{mc-mult-tri2}\\
    & \bv _i \cdot \bv_j \geq 0 \label{mc-mult-tri3}
  \end{align}

\item $\sum_{e=(i,j) \in E}{a(e)\|\bv_i - \bv_j\|^2} \ge \OPT$

\item For any $i$, there exist two solutions $\SSsoln^{i \to 0}$ and
  $\SSsoln^{i \to 1}$ in $\SoS_{\ell - 1}(Q)$ such that, if we denote
  the vectors in   $\SSsoln^{i \to 0}$ by $\bv_S^0$, and the vectors
  in   $\SSsoln^{i \to 1}$ by $\bv_S^1$, we have 
  \[
  \bv_{S} \cdot \bv_0 = (1- \bv_i \cdot \bv_0) \bv^0_{S} \cdot \bv^0_0 + 
  (\bv_i \cdot \bv_0) \bv^1_{S} \cdot \bv^1_0. 
  \]
\end{enumerate}
Moreover, a solution $\SSsoln \in \SoS_\ell$ can be computed in time
polynomial in $n^\ell$. 

Intuitively, we think of $\SSsoln$ as describing a pseudo-distribution
over solutions to $Q$, and we interpret $\bv_S \cdot \bv_T$ as the
pseudo-probability that all variables in $S \cup T$ are set to one, or,
equivalently, as the pseudo-expectation of $\prod_{i \in S \cup
  T}x_i$. Usually we cannot expect that there is any true distribution
giving these probabilities. Nevertheless, the pseudo-probabilities and
pseudo-expectations satisfy some of the properties of actual
probabilities. For example, the transformation from $\SSsoln$ to
$\SSsoln^{i \to b}$ corresponds to \emph{conditioning} $x_i$ to $b$.

We will denote by $x_S = \bv_S \cdot \bv_0$ the marginal value of set
$S$. In particular, we will work with the single-variable marginals
$x_i = x_{\{i\}} = \bv_i \cdot \bv_0$, and will denote $\bx = (x_1,
\ldots, x_n)$. As before, it will be convenient to work with the
component of $\bv_i$ which is orthogonal to $\bv_0$. We define
$\widetilde{\bw}_i = \bv_i - x_i \bv_0$, and $\bw_i =
\frac{1}{\|\widetilde{\bw}_i\|} \widetilde{\bw}_i$. Note that, by the
Pythagorean theorem, $\|\widetilde{\bw}_i\|^2 = x_i-x^2_i$, and $\bv_i
= x_i \bv_0 + \sqrt{x_i - x_i^2}\bw_i$. We define the matrices
$\widetilde{\bW}$ and $\bW$ by $\widetilde{\bW} _{i,j}\eqdef
\widetilde{\bw} _i \cdot \widetilde{\bw} _j$ and $\bW _{i,j}\eqdef \bw
_i \cdot \bw _j$. We can think of $\widetilde{\bW}$ as the covariance
matrix of the pseudodistribution corresponding to the SDP
solution. The following lemma, due to Barak, Raghavendra, and
Steurer~\cite{BarakRS11}, and, independently, to Guruswami and
Sinop~\cite{GuruswamiS11}, shows that any pseudodistribution can be
conditioned so that the covariances are small on average.

\begin{lemma}\label{lm:global-corr}
  For any $\eps_0 \le 1$, and any $\SSsoln \in \SoS_{\ell}(Q)$, where
  $\ell \ge \frac{1}{\eps_0^4} + 2$, there exists an efficiently
  computable $\SSsoln' \in \SoS_{\ell - 1/\eps_0^4}(Q)$, such that
  \begin{equation}\label{eq:avg-cov}
  \sum_{i = 1}^n\sum_{j = 1}^n{\widetilde{\bW}_{i,j}^2} \le \eps_0^4 n^2. 
  \end{equation}
  In particular, $\SSsoln'$ can be computed by conditioning $\SSsoln$
  on $\frac{1}{\eps_0^4}$ variables. 
\end{lemma}

\subsection{Rounding Algorithm}

For our algorithm, we first solve a semidefinite program to compute a
solution in $\SoS_\ell(Q)$, to which we apply
Lemma~\ref{lm:global-corr} with parameter $\eps_0$, which we will
choose later. In order to be able to apply the lemma, we choose $\ell
= \left\lceil\frac{1}{\eps_0^4}\right\rceil + 2$.  The rounding
algorithm itself is similar to the one we used for \MtwoSAT. We
perform a Sticky Brownian Motion with initial covariance $\bW$,
starting at the initial point $\bX_0 = \bx$, i.e.~at the marginals
given by the SDP solution. As variables hit $0$ or $1$, we freeze
them, and delete the corresponding row and column of the covariance
matrix. The main difference from the \MtwoSAT rounding is that we stop
the process at time $\tau$, where $\tau$ is another parameter that we will
choose later. Then, independently for each $i = 1, \ldots, n$, we
include vertex $i$ in the final solution $S$ with probability
$(\bX_\tau)_i$, and output $S$. 

The key property of this rounding that allows us to handle a large
number of global constraints is that, for any $F_i \in \FF$, the value
$\sum_{j \in F_i}(\bX_\tau)_j$ that the fractional solution assigns to
the set $F_i$ satisfies a sub-Gaussian concentration bound around
$b_i$. Note that $\sum_{j \in F_i}(\bX_t)_j$ is a martingale with
expectation equal to $b_i$. Moreover, by Lemma~\ref{lm:global-corr},
the entries of the covariance matrix $\widetilde{\bW}$ are small on
average, which allows us to also bound the entries of the covariance
matrix $\bW$, and, as a consequence, bound how fast the variance of
the martingale increases with time. The reason we stop the walk at
time $\tau$ is to make sure the variance doesn't grow too large: this
freedom, allowed by the Sticky Brownian Motion rounding, is important
for our analysis. The variance bound then implies the sub-Gaussian
concentration of $\sum_{j \in F_i}(\bX_\tau)_j$ around its mean $b_i$,
and using this concentration we can show that no constraint is
violated by too much. This argument crucially uses the fact that our
rounding is a random walk with small increments, and we do not expect
similarly strong concentration results for the random hyperplane
rounding or its variants.

The analysis of the objective function, as usual, reduces to analyzing
the probability that we cut an edge. However, because we start the
Sticky Brownian Motion at $\bx$, which may not be equal to $\mathbf{0}$, our
analysis from Section~\ref{sec:conformal} is not sufficient. Instead,
we use the PDE based analysis from Section~\ref{sec:pde}, which easily
extends to the \MC objective. One detail to take care of is that,
because we stop the walk early, edges incident on vertices that have
not reached $0$ or $1$ by time $\tau$ may be cut with much smaller
probability than their contribution to the SDP objective. To deal with
this, we choose the time $\tau$ when we stop the walk large enough, so that any
vertex has probability at least $1-\mathrm{poly}(\eps)$ to have
reached $\{0,1\}$ by time $\tau$. We show that this happens for $\tau
= \Theta(\log(1/\eps))$. This value of $\tau$ is small enough so that
we can usefully bound the variance of $\sum_{j \in  F_i}(\bX_\tau)_i$
and prove the sub-Gaussian concentration we mentioned above.

Let us recall some notation that will be useful in our analysis. We
will use $\tau_i$ for the first time $t$ that $\bX_t$ hits a face of
$[0,1]^n$ of dimension $n-i$; then, $\bX_{\tau_n} \in \{0,1\}^n$. We
also use $\bW_t$ for the covariance used at time step $t$, which is
equal to $\bW$ with rows and columns indexed by $\{i: (\bX_t)_i \in
\{0,1\}\}$ zeroed out.

As discussed, our analysis relies on a martingale concentration
inequality, and the following lemma, which is proved with the methods
we used above for the \MtwoSAT problem. A proof sketch can be found
in the appendix.
\begin{restatable}{lemma}{mcall}\label{lm:mc-all-marginals}
  For the SDP solution $\SSsoln$ and the Sticky Brownian Motion
  $\bX_t$ described above, and for any pair $\{i,j\}$ of vertices
  \[
  \Pr[(\bX_{\tau_n})_i \neq (\bX_{\tau_n})_j]
  \ge 
  \mcallratio \cdot
  \|\bv_i - \bv_j\|^2.
  \]
\end{restatable}

The next lemma shows that the probability that any coordinate is fixed
by time $t$ drops exponentially with $t$. We use this fact to argue
that by time $\tau = \Theta(\log(1/\eps))$ the endpoints of any edge
have probability at least $1-\mathrm{poly}(\eps)$ to be fixed, and,
therefore, edges are cut with approximately the same probability as if
we didn't stop the random walk early, which allows us to use
Lemma~\ref{lm:mc-all-marginals}. The proof of this lemma, which is
likely well-known, appears in the appendix. 
\begin{restatable}{lemma}{hittingtime}\label{lm:hitting-time}
  For any $i$, and any integer $t \ge 0$, 
  $\Pr[\forall s \le t: 0 < (\bX_s)_i < 1] < 4^{-t}$.
\end{restatable}

The following concentration inequality is our other key lemma. The
statement is complicated by the technical issue that the concentration
properties of the random walk depend on the covariance matrix $\bW$,
while Lemma~\ref{lm:global-corr} bounds the entries of
$\widetilde{\bW}$. When $x_i(1-x_i)$ or $x_j(1-x_j)$ is small,
$\widetilde{\bW}_{i,j}$ can be much smaller than $\bW_{i,j}$. Because
of this, we only prove our concentration bound for sets of vertices
$i$ for which $x_i(1-x_i)$ is sufficiently large. For those $i$ for which
$x_i(1-x_i)$ is small, we will instead use 
the fact that such $x_i$ are already nearly integral to prove a
simpler concentration bound.

\begin{lemma}\label{lm:concentration}
  Let $\eps_0, \eps_1 \in [0,1]$, and $n \ge
  \frac{\eps_1}{8\tau\eps_0^2}$. Define $V_{> \eps_1} = \{i: 2x_i(1-x_i)
  > \eps_1\}$. For any set $F \subseteq V_{> \eps_1}$, and any $t \ge
  0$, the random set $S$ output by the rounding algorithm satisfies
  \[
  \Pr\left[\Bigl| |F \cap S| - \sum_{i \in F}{x_i}\Bigr| \ge t\eps_0 n  \right]
  \le 
  4\exp\left(-\frac{\eps_1 t^2}{4\tau}\right).
  \]
\end{lemma}

We give the proof of Lemma~\ref{lm:concentration} after we finish the
proof of Theorem~\ref{thm:mc-constraints}, restated below for convenience.

\mcconstraints*
\begin{proof}
  The algorithm outputs either the set $S$ output by the Sticky
  Brownian Rounding described above, or the one guaranteed by
  Lemma~\ref{lm:simple-approx}, depending on which one achieves a cut
  of larger total weight. If $\OPT \le \frac{\eps}{2} a(E)$, then 
  Lemma~\ref{lm:simple-approx} achieves the approximation we are
  aiming for. Therefore, for the rest of the proof, we may assume that
  $\OPT \ge \frac{\eps}{2} a(E)$, and that the algorithm outputs the set
  $S$ computed by the Sticky Brownian Rounding. Then, it is enough to
  guarantee that, with high probability, 
  \begin{equation}\label{eq:apx-additive}
  a(\delta(S)) \ge \mcallratio \cdot \OPT - \frac{\eps^2}{2}a(E).
  \end{equation}

  Let us set $\eps_1 = \eps^2\eps_0$, and define, as above, $V_{>
    \eps_1} = \{i: 2x_i(1-x_i) > \eps_1\}$ and let $V_{\le \eps_1} =
  \{i: 2x_i(1-x_i) \le \eps_1\}$. Let $\bY$ be the indicator vector of
  the set $S$ output by the algorithm.  Observe that, for each $i$,
  since $\bY_i$ is a Bernoulli random variable with expectation $x_i$,
  we have $\E\left[\sum_{i \in V_{\le\eps_1}}{|\bY_i - x_i|}\right]
  \le \eps_1n$, and, therefore,
  \[
  \Pr\left[\sum_{i \in V_{\le\eps_1}}{|\bY_i - x_i|} \ge
    16\eps_0n\right] 
  \le \frac{\eps^2}{16}.
  \]
  Then, for any $F_i \in \FF$, by Lemma~\ref{lm:concentration} applied
  to $F_i \cap V_{> \eps_1}$, we have
  \[
  \Pr\left[\bigl| |F_i \cap V_{> \eps_1} \cap S| - 
    \sum_{i \in F \cap V_{> \eps_1}}{x_i}\bigr| \ge
    \sqrt{\frac{4\tau}{\eps_1}\ln \frac{32k}{\eps^2}}\eps_0 n  \right]
  \le 
  \frac{\eps^2}{16k}.
  \]
  Therefore, with probability at least $1 - \frac{\eps^2}{8}$, for all $i \in [k]$ we have
  \begin{align*}
  \left||F_i \cap S| - \sum_{i \in F}{x_i}\right|
  &\le 
  \left||F_i \cap V_{\le \eps_0} \cap S| - \sum_{i \in F\cap V_{\le \eps_0}}{x_i}\right|
  +
  \left||F_i \cap V_{>\eps_0} \cap S| - \sum_{i \in F\cap V_{> \eps_0}}{x_i}\right|\\
  &\le 
  \sum_{i \in F\cap V_{\le\eps_0}}{|\bY_i - x_i|} + \left||F_i \cap  V_{>\eps_0} \cap S| - \sum_{i \in F\cap V_{> \eps_0}}{x_i}\right|\\
  &\le
  \left(16 + \sqrt{\frac{4\tau}{\eps^2\eps_0}\ln \frac{32k}{\eps^2}}\right)\eps_0n.
  \end{align*}
  This means that, with probability at least $1 - \frac{\eps^2}{8}$,
  $S$ satisfies all the constraints up to additive error $\eps n$, as
  long as 
  \[
  \eps_0 \le \min\Biggl\{
    \frac{\eps}{32}, 
    \frac{\eps^4}{4 \sqrt{\tau \ln\frac{32k}{\eps^2}}}
  \Biggr\}.
  \]

  It remains to argue about the objective function. For $\tau \ge
  \log_2\frac{2\sqrt{2}}{\eps}$, Lemma~\ref{lm:hitting-time} implies
  that, for any vertex $i$, $\Pr[(\bX_\tau)_i \not \in \{0,1\}] \le
  4^{-\tau} \le \frac{\eps^2}{8}$.
  By
  Lemma~\ref{lm:mc-all-marginals}, any pair of vertices $\{i,j\}$ is
  separated with probability
  \[
  \Pr[(\bX_{\tau_n})_i \neq (\bX_{\tau_n})_j]\ge  \mcallratio \cdot \|\bv_i - \bv_j\|^2,
  \]
  where we recall that, for edge $e = (i,j)$, $a(e)\|\bv_i -
  \bv_j\|^2$ is the contribution of $e$ to the objective value.  Then,
  \begin{align*}
  \Pr[(\bX_{\tau})_i \neq (\bX_{\tau})_j]
  &\ge 
  \Pr[(\bX_{\tau_n})_i \neq (\bX_{\tau_n})_j, 
  (\bX_\tau)_i = (\bX_{\tau_n})_i, (\bX_\tau)_j = (\bX_{\tau_n})_j]\\
  &= 
  \Pr[(\bX_{\tau_n})_i \neq (\bX_{\tau_n})_j, 
  (\bX_\tau)_i \in \{0,1\}, (\bX_\tau)_j \in \{0,1\}]\\
  &\ge 
  \Pr[(\bX_{\tau_n})_i \neq (\bX_{\tau_n})_j] - \frac{\eps^2}{4}\\
  &\ge 
  \mcallratio\cdot \|\bv_i - \bv_j\|_2^2 - \frac{\eps^2}{4}.
  \end{align*}
  Therefore, $\E[a(\delta(S))] \ge \mcallratio \cdot \OPT -
  \frac{\eps^2}{4} a(E)$. By Markov's inequality applied to $a(E) -
  a(\delta(S))$, 
  \[
  \Pr\Bigl[a(\delta(S)] < \mcallratio \cdot \OPT -  \frac{\eps^2}{2} a(E)\Bigr] 
  < 
  \frac{1 + \frac{\eps^2}{4}}{1 + \frac{\eps^2}{2}}
  < 1 - \frac{\eps^2}{5}.
  \]
  In conclusion, we have that with probability at least
  $\frac{3}{40}\eps^2$, \eqref{eq:apx-additive} is satisfied, and all
  global constraints are satisfied up to an additive error of $\eps
  n$. The probability can be made arbitrarily close to $1$ by
  repeating the entire algorithm $O(\eps^{-2})$ times. To complete the
  proof of the theorem, we can verify that the running time is
  dominated by the time required to find a solution in $\SoS_\ell(Q)$,
  which is polynomial in $n^\ell$, 
  where $\ell = O(\eps_0^{-4}) = \mathrm{poly}(\log(k)/\eps)$. 
\end{proof}

We finish this section with the proof of Lemma~\ref{lm:concentration}

\begin{proof}[Proof of Lemma~\ref{lm:concentration}]
  Since each $i$ is included in $S$ independently with probability
  $(\bX_\tau)_i$, by Hoeffding's inequality we have
  \[
  \Pr\left[\Bigl| |F \cap S| - \sum_{i \in F}{(\bX_\tau)_i}\Bigr| \ge t\eps_1 n  \right]
  \le 
  2e^{-2\eps_1^2 t^2 n} \le    2\exp\left(-\frac{\eps_1 t^2}{4\tau}\right),
  \]
  where the final inequality follows by our assumption on
  $n$. Therefore, it is enough to establish
  \begin{equation}
    \label{eq:martingle-conc}
    \Pr\left[\Bigl|\sum_{i \in F}{(\bX_\tau)_i - x_i}\Bigr| \ge t\eps_1 n  \right]
    \le   2\exp\left(-\frac{\eps_1 t^2}{4\tau}\right).
  \end{equation}
  Suppose $\by \in \{0,1\}^n$ is the indicator vector of $F$ so that
  \[\by^\top (\bX_\tau - \bX_0) = \by^\top (\bX_\tau - \bx)= \sum_{i \in F}{((\bX_\tau)_i - x_i)}.\] A
  standard calculation using It\^{o}'s lemma (see Exercise
  4.4.~in~\cite{Oksendal-SDE}) shows that, for any $\lambda \ge 0$,
  the random process
  \[
  Y_t = \exp\left(\lambda \by^\top (\bX_t - \bx) 
    -\frac{\lambda^2}{2} \int_0^t (\by^\top  \bW_s \by) ds \right)
  \]
  is a martingale with starting state $Y_0 = 1$. Since, for any $s$, $\bW_s$ equals $\bW$ with some
  rows and columns zeroed out, we have that $\bW - \bW_s$ is positive
  semidefinite, and $\by^\top  \bW_s \by \le \by^\top
  \bW\by$. Therefore, 
  \[
  \E\left[\exp\Bigl(\lambda \by^\top (\bX_\tau- \bx) 
    - t \frac{\lambda^2}{2} \by^\top \bW \by \Bigr)\right] 
  \le \E[Y_\tau] = 1. 
  \]
  Rearranging, this gives us that, for all $\lambda \ge 0$,
  \begin{equation}\label{eq:mgf}
  \E[e^{\lambda \by^\top (\bX_\tau - \bx)}] \le 
  \E[e^{\tau \lambda^2 \by^\top \bW \by / 2}]
  \le e^{\tau\lambda^2 \by^\top \bW \by / 2}.
  \end{equation}
  We can bound $\by^\top \bW \by$ using Cauchy-Schwarz, the assumption
  that $2x_i (1-x_i) > \eps_1$ for all $i \in F$, and
  \eqref{eq:avg-cov}:
  \begin{align*}
    \by^\top \bW \by &= \sum_{i \in F}\sum_{j \in F} {\bW_{i,j}}\\
    &\le |F|\left(\sum_{i \in F}\sum_{j \in F}{\bW_{i,j}^2}\right)^{1/2}\\
    &= |F|\left(\sum_{i \in F}\sum_{j \in F}{\frac{\widetilde{\bW}_{i,j}^2}{x_i x_j (1-x_i) (1-x_j)}}\right)^{1/2}\\
    &< \frac{2n}{\eps_1} \left(\sum_{i \in F}\sum_{j \in F} {\widetilde{\bW}_{i,j}^2}\right)^{1/2}\\
    &\le \frac{2\eps_0^2 n^2}{\eps_1}.
  \end{align*}
  Plugging back into \eqref{eq:mgf}, we get 
  $\E[e^{\lambda \by^\top  (\bX_\tau - \bx)}] \le e^{\tau\lambda^2 \eps_0^2  n^2/\eps_1}$. 
  The  standard exponential moment argument then implies \eqref{eq:martingle-conc}.
\end{proof}

\section{Extensions of the Brownian Rounding}\label{sec:general}
In this section, we consider two extensions of the Brownian rounding algorithm.
We also present numerical results for these variants showing improved
performance over the sticky Brownian Rounding analyzed in previous
sections.

\subsection{Brownian Rounding with Slowdown}\label{sec:slowdown}

As noted in~\cref{sec:conformal}, the Sticky Brownian rounding algorithm does
not achieve the optimal value for the \MC problem. A natural question is to ask
if we can modify the algorithm to achieve the optimal constant. In this
section, we will show that a simple modification achieves this ratio up to at
least three decimals. Our results are computer-assisted as we solve partial
differential equations using finite element methods. These improvement indicate
that variants of the Brownian Rounding approach offer a direction to obtain
optimal SDP rounding algorithms for \MC problem as well as other CSP problems.

In the sticky Brownian motion, the covariance matrix $W_t$ is a constant, until
some vertex's marginals $(\bX_t)_i$ becomes $\pm1$. At that point, we abruptly
zero the $i^{th}$ row and column. In this section, we analyze the algorithm
where we gradually dampen the step size of the Brownian motion as it approaches the
boundary of the hypercube, until it becomes $0$ at the boundary.
We call this process a ``Sticky Brownian Motion with Slowdown."

Let $(\bX_t)_i$ denote the marginal value of vertex $i$ at time $t$. Initially
$(\bX_0)_i=0$.  First, we describe the discrete algorithm  which will provide
intuition but will also be useful to those uncomfortable with Brownian motion and
diffusion processes.  At each time step, we will take a step whose length is
scaled by a factor of $(1-(\bX_t)_i^2)^{\alpha}$ for some constant $\alpha$. 
In particular, the marginals will evolve according to the equation:
\begin{align}
    (\bX_{t+dt})_i &= (\bX_t)_i + (1-(\bX_t)_i)^2)^{\alpha/2} \cdot \big( \bw_i \cdot \mathbf{G}_t \big) \cdot \sqrt{dt}.  \label{def:slowdown}
\end{align}
where $\mathbf{G}_t$ is distributed according to an $n$-dimensional Gaussian and $dt$ is
a small discrete step by which we advance the time variable.
When $\bX_t$ is sufficiently close to $-1$ or $+1$, we round it to the
nearest one of the two: from then on it will stay fixed because of the
definition of the process, i.e.~we will have $(\bX_{s})_i = (\bX_t)_i$ for
all $s > t$.

More formally, $\bX_t$ is defined as an \Ito diffusion process which satisfies the
stochastic differential equation
\begin{align}
    d\bX_t = \mathbf{A}(\bX_t)\cdot \bW^{\nicefrac[]{1}{2}} \cdot d\bB_t  \label{eqn:sde}
\end{align} where $\bB_t$ is
the standard Brownian motion in $\mathbb{R}^n$ and $\mathbf{A}(\bX_t)$ is the diagonal
matrix with entries $[\mathbf{A}(\bX_t)]_{ii} = (1-(\bX_t)_i^2)^{\alpha/2}$. Since
this process is continuous, it becomes naturally sticky when some
coordinate $(\bX_t)_i$ reaches $\set{-1,1}$. 

Once again, it suffices to restrict our attention to the two dimensional case where we
analyze the probability of cutting an edge $(i,j)$ and we will assume that
\[ \tilde{\bW}=\begin{pmatrix} 1 & \cos (\theta)
            \\ \cos (\theta) & 1 \end{pmatrix},
        \] where $\theta$ is the angle between $\bw _i$ and $\bw _j$.

Let $\tau$ be the first time when $\bX_t$ hits the boundary $\bd[-1,1]^2$.
Since the walk slows down as it approaches the boundary, it is worth asking if
$\E[\tau]$ is finite. In~\Cref{lem:slowdown-stop}, we show that $\E[\tau]$ is
finite for constant $\alpha$.

Let $u(x,y)$ denote the probability of the Sticky Brownian Walk algorithm starting
at $(x,y)$ cutting an edge, i.e.~the walk is absorbed in either
$(+1,-1)$ or $(-1,+1)$. It is easy to give a precise formula for $u$
at the boundary as the algorithm simplifies
to a one-dimensional walk. Thus, $u(x,y)$ satisfies the boundary condition $\phi(x,y) = (1-xy)/2$
for all points $(x,y) \in bd[-1,1]^2$.
For a given  $(x,y) \in \Int [-1,1]^2 $, we can say
\[ u(x,y) = \E^{(x,y)}[ \phi(\tilde{\bX}_{\tau}(i), \tilde{\bX}_{\tau}(j)) ], \]
where $\E^{(x,y)}$ denotes the expectation of diffusion process that begins at $(x,y)$. Informally,
$u(x,y)$ is the expected value of $\phi$ when first hitting $\partial [-1,1]^2$
conditioned on starting at point $(x,y)$. Observe that the probability that the algorithm
will cut an edge is given by $u(0,0)$.

The key fact about $u(x,y)$ that we use is that it is the unique solution
to a Dirichlet Problem, formalized in~\Cref{lem:dirichlet} below.

\begin{lemma} \label{lem:dirichlet}
    Let $\L^{\alpha}$ denote the operator
\[ \L^{\alpha} =  (1-x^2)^{\alpha} \frac{\partial ^2}{\partial x^2}
                 + 2 \cos(\theta) (1-x^2)^{\alpha/2} (1-y^2)^{\alpha/2} \frac{\partial^2}{\partial x \partial y}
             + (1-y^2)^{\alpha} \frac{\partial ^2}{\partial y^2}
                ,\]
                then the function $u(x,y)$ is the unique solution to the Dirichlet Problem:
                \begin{align*}
                    \L^{\alpha}[u] (x,y) &=0  && \forall(x,y) \in \Int([-1,1]^2) \\
                    \lim_{\substack{(x,y)\rightarrow (\tilde{x},\tilde{y}),\\ (x,y) \in \Int([-1,1]^2)}} u(x,y) &= \phi(\tilde{x},\tilde{y})
                                                                                && \forall (\tilde{x},\tilde{y}) \in \bd[-1,1]^2.
                \end{align*}

\end{lemma}
The proof again uses~\cite[Theorem 9.2.14]{Oksendal-SDE},
however, the exact application is a little subtle and we defer the details
to~\Cref{rem:dirichlet}. 

\paragraph*{Numerical Results}
The Dirichlet problem is parameterized by two variables: the
slowdown parameter $\alpha$ and the angle between the vectors $\theta$.  We can
numerically solve the above equation using existing solvers for any given fixed $\alpha$ and
angle $\theta \in [0,\pi]$. We solve these problems for a variety
of $\alpha$ between $0$ and $2$ and all values of $\theta$ in $[0,\pi]$ discretized to a granularity
of $0.02$.\footnote{Our code, containing the details of the implementation, is available at \cite{code-repo}.}

We observe that as we increase $\alpha$ from $0$ to $2$, the approximation
ratio peaks around $\alpha \approx 1.61$ for all values of $\theta$.
In particular, when $\alpha = 1.61$, the approximation ratio is $\GWratio$ which matches
the integrality gap for this relaxation up to three decimal points.

The Brownian rounding with slowdown is a well-defined algorithm for
any $2$-CSP.
We investigate 3 different values of slowdown parameter, i.e., $\alpha$, and show their relative
approximation ratios.
We show that with a slowdown of $1.61$ we achieve an approximation ratio of
$0.929$  for
\MtwoSAT. We list these values below in~\Cref{tab:max-cut-ratio-results}.

For the \MC problem, since we start the walk at the point $(0,0)$, we only need
to investigate the performance of the rounding for all possible angles
between two unit vectors which range in
$[0, \theta]$~(\Cref{fig:maxcut-worst-case-results}).  In particular, we are able to
achieve values that are comparable to the Goemans-Williamson bound.

\begin{table}[ht]
    \centering
    \begin{tabular}{|c|c|c|c|}
         \hline
         $\alpha$  &  \MC \cut{& Balanced-\MtwoSAT }& \MtwoSAT \\
         \hline
         $0$ & $\mcbasicratio$\cut{& 0.936 }& 0.921    \\
         $1$ & $\mcslowratio$ \cut{& 0.9423} & 0.927 \\
         $1.61$ & $\mcexpratio$ \cut{& 0.943} & 0.929 \\
         \hline
    \end{tabular}
    \caption{Approximation ratio of  Sticky Brownian Motion rounding
      with Slowdown for \MC and \MtwoSAT.}
    \label{tab:max-cut-ratio-results}
\end{table}

\begin{figure}[ht]
\begin{subfigure}{.32\textwidth}
  \centering
  \includegraphics[width=1\linewidth]{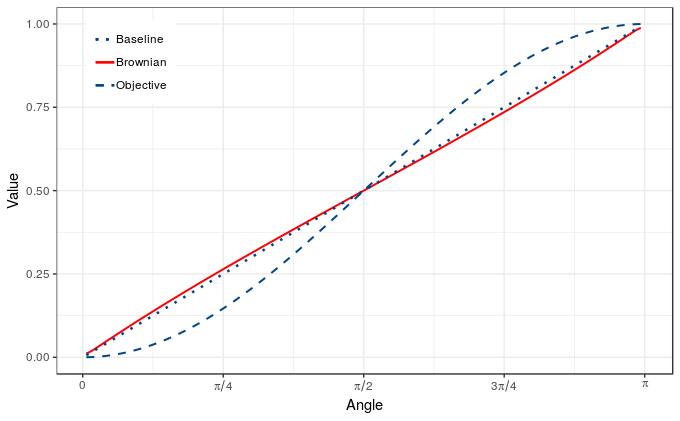}
  \caption{$\alpha=0$.}
\end{subfigure}%
\hspace{0.01\textwidth}
\begin{subfigure}{.32\textwidth}
  \centering
  \includegraphics[width=1\linewidth]{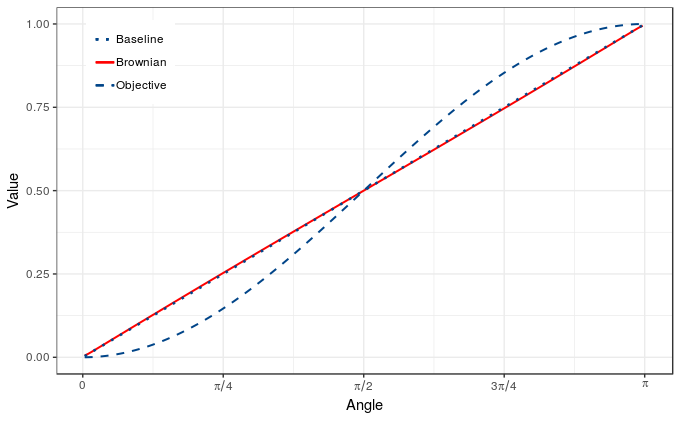}
  \caption{$\alpha=1$.}
\end{subfigure}%
\hspace{0.01\textwidth}
\begin{subfigure}{.32\textwidth}
  \centering
  \includegraphics[width=1\linewidth]{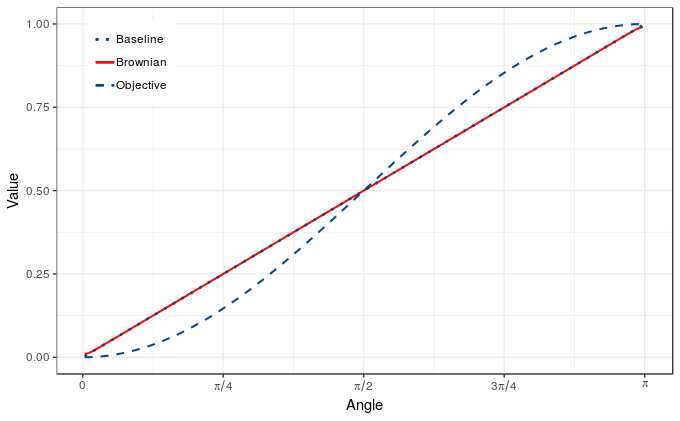}
  \caption{$\alpha=1.61$.}
\end{subfigure}
\caption{Comparing the performance of three values of the slowdown
  parameter for the \MC problem.}
\label{fig:maxcut-worst-case-results}
\end{figure}

\subsection{Higher-Dimension Brownian Rounding}\label{sec:higher-dim}
Our motivating example for considering the higher-dimension Brownian rounding
is the \MdC problem: given a directed graph $G=(V,E)$ equipped with
non-negative edge weights $a:E\rightarrow \mathbb{R}_+$ we wish to find a cut
$S\subseteq V$ that maximizes the total weight of edges going out of $S$.  The
standard semi-definite relaxation for \MdC is the following:
\begin{align*}
\max ~~~ & \sum _{e=(i\rightarrow j)\in E} a_e\cdot \frac{\left( \bw_0 + \bw_i\right) \cdot \left( \bw_0 - \bw_j\right)}{4} & \\
s.t. ~~~ & \bw _i \cdot \bw _i = 1 & \forall i=0,1,...,n \\
 ~~~ & \norm{\bw_i-\bw_j}^2 + \norm{\bw_j - \bw_k}^2 \geq \norm{\bw_i - \bw_k}^2 & \forall i,j,k=0,1,...,n
\end{align*}
In the above, the unit vector $\bw _0$ denotes the cut $S$, whereas $-\bw _0$ denotes $\overline{S}$.
We also include the triangle inequalities which are valid for any valid relaxation.

The sticky Brownian rounding algorithm for \MdC fails to give a good performance guarantee. Thus we design a high-dimensional variant of the algorithm that incorporates
the inherent asymmetry of the problem.
Let us now describe the high-dimensional Brownian rounding algorithm.
It is similar to the original Brownian rounding algorithm given for \MC,
except that it evolves in $\mathbb{R}^{n+1}$ with one additional dimension for
$\bw _0$.  Let $\bW\in \mathbb{R}^{(n+1)\times (n+1)}$ denote the positive semi-definite correlation matrix
defined by the
vectors $\bw_0,\bw_1,\ldots, \bw_n$, {\em i.e.}, for every $0\leq i,j\leq n$ we have
that: $\bW _{i,j} = \bw _i \cdot \bw _j$.
The algorithm starts at the origin and perform a sticky Brownian motion
inside the $[\pm1]^{n+1}$ hypercube whose correlations are governed by $\bW$.

As before, we achieve this by defining a random process $\left\{ \bX_t\right\} _{t\geq 0}$ as follows: $$ \bX _t = \bW ^{\nicefrac[]{1}{2}}\bB _t,$$
where $\left\{ \bB _t\right\} _{t\geq 0}$ is the standard Brownian motion in $\mathbb{R}^{n+1}$ starting at the origin and
$\bW ^{\nicefrac[]{1}{2}}$ is the square root matrix of $\bW$.  Additionally, we
force $\left\{ \bX _t\right\} _{t\geq 0}$ to stick to the boundary of the
$[\pm 1]^{n+1}$ hypercube, {\em i.e.}, once a coordinate of $\bX _t$ equals either $1$ or
$-1$ it is fixed and remains unchanged indefinitely. This description
can be formalized the same way we did for the \MC problem. Below we
use $\sigma$ for the time at which $\bX_\sigma \in \{-1, 1\}^{n+1}$,
which has finite expectation. 

Unlike the Brownian rounding algorithm for \MC, we need to take into
consideration the value $\bw_0$ was rounded to, {\em i.e.}, $\left(\bX _\sigma\right)_0
$, since the zero coordinate indicates $S$.
Formally, the output $S \subseteq V$ is defined as
follows:
\[ S=\left\{ i\in V:\left( \bX _\sigma\right)_i = \left( \bX _\sigma\right)_0\right\}.\]
To simplify the rest of the presentation, let us denote
$Z_i\eqdef \left( \bX _\sigma\right)_i$ for every $i=0,1,\ldots,n$.

The event that an edge $(i\rightarrow j)\in E$ is an outgoing edge from $S$,
{\em i.e.}, $i\in S$ and $j\in \overline{S}$, involves three random variables:
$Z_i$, $Z_j$, and $Z_0$.  Formally, the above event happens if and only if
$Z_i=Z_0$ and $Z_j\neq Z_0$.
We now show how events on any triplet of the random variables $Z_0,Z_1,\ldots,Z_n$ can be precisely calculated.
To simplify the presentation, denote the following for every $i,j,k=0,1,2,\ldots,n$ and $\alpha,\beta,\gamma \in \left\{ \pm 1\right\}$:
\begin{align*}
p_i(\alpha) & \triangleq  \Pr \left[ Z_i=\alpha\right] \\
p_{ij}(\alpha,\beta) & \triangleq \Pr \left[ Z_i=\alpha,Z_j=\beta\right] \\
p_{ijk}(\alpha,\beta,\gamma) & \triangleq \Pr \left[ Z_i=\alpha, Z_j=\beta, Z_k=\gamma\right]
\end{align*}
\begin{observation}\label{obs:probabilities}
The following two hold:
\begin{enumerate}
\item $p_i(\alpha)=p_i(-\alpha)$, $p_{ij}(\alpha,\beta) = p_{ij}(-\alpha,-\beta)$, and $p_{ijk}(\alpha,\beta,\gamma)=p_{ijk}(-\alpha,-\beta,-\gamma)$ for every $i,j,k=0,1,2,\ldots,n$ and $\alpha,\beta,\gamma \in \left\{ \pm 1\right\}$.  \label{obs1}
\item $p_i(\alpha)=\nicefrac[]{1}{2}$ for every $i=0,1,2,\ldots,n$ and $\alpha \in \left\{ \pm 1\right\}$. \label{obs2}
\end{enumerate}
\end{observation}
The proof of Observation \ref{obs:probabilities} follows immediately from symmetry.

The following lemmas proves that every conjunction event that depends on three variables from $Z_0,Z_1,Z_2,\ldots,Z_n$ can be precisely calculated.
\begin{lemma}\label{lem:ProbEvent}
For every $i,j,k=0,1,2,\ldots,n$ and $\alpha,\beta,\gamma \in \left\{ \pm 1\right\}$:
$$ p_{ijk(\alpha,\beta,\gamma)} = \frac{1}{2}\left[ p_{ij}(\alpha,\beta)+p_{ik}(\alpha,\gamma)+p_{jk}(\beta,\gamma)-\frac{1}{2}\right].$$
\end{lemma}
\begin{proof}
\begin{align*}
1-p_{ijk}(\alpha,\beta,\gamma) & = 1-p_{ijk}(-\alpha,-\beta,-\gamma) \\
& = \Pr \left[ Z_i=\alpha \vee Z_j=\beta \vee Z_k=\gamma\right] \\
& = p_i(\alpha)+p_j(\beta)+p_k(\gamma) - p_{ij}(\alpha,\beta) - p_{ik}(\alpha,\gamma)-p_{jk}(\beta,\gamma)+p_{ijk}(\alpha,\beta,\gamma).
\end{align*}
The first equality follows from property (\ref{obs1}) of Observation \ref{obs:probabilities}.
The second equality follows from De-Morgan's law.
The third equality follows from the inclusion and exclusion principle.
Isolating $p_{ijk}(\alpha,\beta,\gamma)$ above and using property (\ref{obs2})
of Observation \ref{obs:probabilities} concludes the proof.
\end{proof}

Let us now consider the case study problem \MdC.
One can verify that an edge $(i\rightarrow j)\in E$ is a forward
edge crossing the cut $S$ if and only if the following event happens:
$\left\{ Z_i=Z_0\neq Z_j\right\}$ (recall that $Z_0$ indicates $S$).
Thus, the value of the Brownian rounding algorithm, when considering only the
edge $(i\rightarrow j)$, equals $p_{0ij}(1,1,-1)+p_{0ij}(-1,-1,1)$.  Lemma
\ref{lem:ProbEvent} above shows that if one knows the values of
$p_{ij}(\alpha,\beta)$ for every $i,j=0,1\ldots,n$ and $\alpha,\beta\in \{ \pm
1\}$, then $p_{0ij}(1,1,-1) $ and $p_{0ij}(-1,-1,1) $ can be calculated (thus
deriving the exact probability that $(i\rightarrow j)$ is a forward edge
crossing $S$).

How can we calculate $p_{ij}(\alpha,\beta)$ for every $i,j=0,1\ldots,n$ and $\alpha,\beta\in \{ \pm 1\}$?
Fix some $i$, $j$, $\alpha$ and $\beta$.
We note that Theorem \ref{thm:cut-prob} can be used to calculate $p_{ij}(\alpha,\beta)$.
The reason is that: $(1)$ Theorem \ref{thm:cut-prob} provides the value of $ p_{ij}(-1,1)+ p _{ij}(1,-1)$; $(2)$ $p_{ij}(-1,-1)+p_{ij}(-1,1)+p_{ij}(1,-1)+p_{ij}(1,1)=1$; and $(3)$ $p_{ij}(-1,-1)=p_{ij}(1,1) $ and $p_{ij}(-1,1)=p_{ij}(1,-1) $ from symmetry.
We conclude that using Theorem \ref{thm:cut-prob} we can exactly calculate the probability that $(i\rightarrow j)$ is a forward edge crossing $S$, and obtain that this probability equals:
$$ \frac{1}{2}(p_{0j}+p_{ij}-p_{0i}),$$
where $p_{ij}$ is the probability that $i$ and $j$ are separated as given by Theorem \ref{thm:cut-prob}.

Similarly to \MtwoSAT, not all triplets of angles
$\{ \theta _{0i}, \theta _{0j}, \theta _{ij}\}$ are possible due to the triangle inequality constraints
(here $ \theta _{ij}$ indicates the angle between $\bw _i$ and $ \bw _j$).  Let
us denote by $\mathcal{F}$ the collection of all possible triplet of angles for
the \MdC problem.
Then, we can lower bound the approximation guarantee of the Brownian rounding algorithm as follows:
$$ \min _{(\theta _{0i},\theta _{0j}, \theta _{ij})\in \mathcal{F}} \left\{ \frac{\frac{1}{2}(p_{0j}+p_{ij}-p_{0i})}{\frac{1}{4}(1-\cos (\theta _{0j})+\cos (\theta _{0i})- \cos(\theta _{ij}))}\right\}.$$
This results in the following theorem.

\begin{theorem}
    The high dimensional Brownian rounding algorithm achieves an approximation ratio of $\MDCRatio$ for the \MdC problem.
\end{theorem}

We also remark that we can introduce slowdown (as discussed in Section~\ref{sec:slowdown} to the high dimensional Brownian rounding algorithm. Numerically, we show that this improves the performance to $0.79$-approximation.

\section{Acknowledgements}
The authors are grateful to Assaf Naor and Ronen Eldan for sharing their
manuscript with us.
SA would like to thank Gantumur Tsogtgerel for bringing the Maximum Principle
to our attention, and Christina C. Christara for helping us compare various numerical PDE solvers.
SA and AN thank Allan Borodin for useful discussions during the initial stages of this research.
GG would like to thank Anupam Gupta and Ian Tice for
useful discussion, and for also pointing out the Maximum Principle.

SA and AN were supported by an NSERC Discovery Grant (RGPIN-2016-06333).
MS was supported by NSF grant CCF-BSF:AF1717947. 
Some of this work was carried out while AN and NB were visitors at the
Simons Institute program on Bridging Discrete and Continuous
Optimization, partially supported by NSF grant \#CCF-1740425. 

\bibliographystyle{plainnat}
\bibliography{main} 

\appendix
\section{Definition of Brownian Motion}
\label{app:brownian-defn}

For completeness, we recall the definition of standard Brownian
motion. 

\begin{definition}
  A stochastic process $\set{\bB_t}_{t\ge 0}$ taking values in $\Re^n$
  is called an $n$-dimensional Brownian motion started at $\bx \in
  \Re^n$ if
  \begin{itemize}
  \item $\bB_0 = \bx$,
  \item the process has independent increments, i.e.~for all $N$ and
    all times $0 \le t_1 \le t_2 \le \ldots \le t_N$, the increments
    $\bB_{t_N} - \bB_{t_{N-1}}$,     $\bB_{t_N-1} - \bB_{t_{N-2}}$,
    $\ldots$,     $\bB_{t_2} - \bB_{t_{1}}$ are independent random
    variables,
  \item for all $t \ge 0$ and all $h > 0$, the increment $\bB_{t+h} -
    \bB_t$ is distributed as a Gaussian random variable with mean $\mathbf{0}$
    and covariance matrix equal to the identity $\mathbf{I}_n$,
  \item the function $f(t) = \bB_t$ is almost surely continuous.
  \end{itemize}
  The process $\set{\bB_t}_{t\ge 0}$ is called standard Brownian
  motion if $\bx = 0$. 
\end{definition}

The fact that this definition is not empty, i.e.~that such a
stochastic process exists, is non-trivial. The first rigorous proof of
this fact was given by Wiener~\cite{Wiener23}. We refer the reader to
the book \cite{MP10} for a thorough introduction to Brownian motion
and its properties. 
\section{Omitted Proofs from Section~\ref{sec:conformal}}

We start with a brief primer about special functions with an emphasis on the
lemmas and identities that will be useful for our analysis. We recommend the
excellent introductions in \citet{andrews1999special,BW10} for a thorough introduction.

\subsection{Special Functions: A Primer}\label{sec:primer}

While there is no common definition of special functions, three basic
functions, $\Gamma$, $\beta$ and the hypergeometric functions ${}_pF_q$  show
up in nearly all treatments of the subject. We will define them and some useful relationships between them.

\begin{definition}[Gamma Function]
    \label{def:gamma}
    The gamma function is defined as
    \[ \Gamma(z) \eqdef \int_0^\infty x^{z-1} e^{-x} dx. \]
    for all complex numbers $z$ with non-negative real part, and
    analytically extended to all $z \neq 0,-1,-2,...$.
\end{definition}

\begin{fact} \label{fact:gamma}
Recall that the gamma function satisfies the recurrence  $\Gamma(z+1)
= z \Gamma(z)$ and it follows easily from the definition that
$\Gamma(1) = 1$.
In particular, when $n$ is a positive integer, $\Gamma(n+1) = n!$
\end{fact}

\begin{definition}[Beta Function]
    The  beta function $\beta(a,b)$ is defined for complex numbers $a$
    and $b$ with $\RePart(a)>0,\RePart(b)>0$ by
    \[ \beta(a,b) = \int_0^1 s^{a-1} (1-s)^{b-1} ds. \]
\end{definition}
Clearly, $\beta(a,b) = \beta(b,a)$.
Setting $s = u/(u+1)$  gives the following alternate form.
\[ \beta(a,b) =  \int_0^\infty u^{a-1} \left(\frac{1}{1+u}\right)^{a+b} du \]
\begin{lemma} \label{lem:betagamma} (Theorem 1.1.4 in~\cite{andrews1999special})
    The beta function can be expressed in terms of the gamma function using the following
    identity:
    \[\beta(a,b)  = \frac{\Gamma(a)\Gamma(b)}{\Gamma(a+b) }\]
\end{lemma}

We will use the following very useful fact.
\begin{lemma} \label{l:trig_beta}
( Exercise 2.2 in~\cite{BW10})
    \[ \int_{0}^{\pi/2} \sin^{a-1}\theta \cos^{b-1} \theta d\theta = \frac12 \beta(a/2,b/2)  \]
\end{lemma}

The next family of functions we utilize are the hypergeometric
functions.

\begin{definition}[Hypergeometric Function]
    \label{def:hypergeom}
    The {\em hypergeometric function}  $\pFq{p}{q}{a_1,\dots,a_p}{b_1,\dots,b_q}{z}$
       is defined as
\[ 
\pFq{p}{q}{a_1,\dots,a_p}{b_1,\dots,b_q}{z}
\eqdef \sum_{n=0}^\infty \frac{ (a_1)_n \cdots (a_p)_n}{(b_1)_n \cdot (b_q)_n} \frac{z^n}{n!} \]
where the  {\em Pochhammer symbol} (rising factorial) is defined
inductively as  \[(a)_n :=  a(a+1) \cdots(a+n-1) \text{ and }
(a)_0=1.\]
\end{definition}

A very simple but useful way to write the binomial theorem using the Pochhammer symbol is
\[ (1-x)^{-a} = \sum_{n=0}^\infty  \frac{(a)_n}{n!}x^n.\]

The Pochhammer symbol also satisfies the formula $(a)_n = \Gamma(a+n)/\Gamma(a)$.

A useful connection between the hypergoemetric function ${}_2F_1$ and the gamma function is given in the following lemma.
\begin{lemma}[Euler's Integral Representation][Theorem~2.2.1 in~\cite{andrews1999special}]
    \label{lem:euler-rep}
    If $\RePart(c) > \RePart(b) > 0$ then
    \[ \pFq{2}{1}{a,b}{c}{x} = \frac{\Gamma(c)}{\Gamma(b) \Gamma(c-b)} \int_{0}^1 t^{b-1} (1-s)^{c-b-1} (1-xs)^{-a} ds \]
    where we assume that $(1-xs)^{-a}$ takes its principal value.
\end{lemma}

\begin{definition}
The {\em incomplete beta integral} is defined as
\[ \beta_x(a,b)= \int_0^x t^{a-1} (1-t)^{b-1} dt, \]
and is well-defined for $\RePart(a)>0$ and $x \notin [1,\infty)$.
\end{definition}

\Cref{l:trig_beta} easily extends to the incomplete
beta integral too, as captured in the following lemma.

\begin{lemma}
    \label{l:incomp_trig}
    \[ \int_{0}^{\phi} \sin^{a-1}\theta \cos^{b-1} \theta d\theta = \frac12 \beta_{\sin^2 \phi}(a/2,b/2)  \]
\end{lemma}
\begin{proof}
Let $\sin \theta = t$, then  $\cos \theta = \sqrt{1-t^2}$ and $(\cos
\theta) d\theta = dt$, and we get
\[\int_{0}^{\phi} \sin^{a-1}\theta \cos^{b-1} \theta d\theta = \int_0^{\sin(\phi)} t^{a-1} (1-t^2)^{(b-2)/2} dt. \]
Setting $s=t^2$   gives 
\[
\int_0^{\sin^2(\phi)} (1/2) s^{(a-2)/2} (1-s)^{(b-2)/2} ds  = (1/2) \beta_{\sin^2 \phi}(a/2,b/2).
\]
This completes the proof.
\end{proof}

The following identity relates the incomplete beta integral to hypergeometric functions.
\begin{lemma}[Gauss's Identity][Exercise 8.7 in~\cite{BW10}]  \label{l:inc_beta_hyper}
\[ \beta_x(a,b) = \int_0^x t^{a-1} (1-t)^{b-1} dt = \frac{ x^a}{a}\cdot \pFq{2}{1}{a,1-b}{a+1}{x}\] 
\end{lemma}
\begin{proof}
    It is natural to substitute $t = sx$, as we can now integrate from $s=0$ to $1$. This gives
    \[ \int_0^1  x^{a-1} s^{a-1} (1-sx)^{b-1} x ds  = x^a \int_0^1 s^{a-1} (1-sx)^{b-1} ds \]
    Using the integral form given in~\Cref{lem:euler-rep} with $1-b$
    in the place of $a$, $a$ in the place of $b$, and $a+1$ in the
    place of $c$, we get that the integral equals
    \[
        x^a \frac{\Gamma(a)\Gamma(1)}{\Gamma(a+1)} \cdot \pFq{2}{1}{1-b,a}{a+1}{x}   
    =
    \frac{x^a}{a} \cdot \pFq{2}{1}{1-b,a}{a+1}{x} 
    \]
    By the symmetry in the definition of ${}_2F_1$ with respect to  the  first two arguments, the result follows.
\end{proof}

\subsection{Proof of\texorpdfstring{~\Cref{thm:cut-prob}}{}}
First, we will prove the claim, which expresses the function $r(\phi)$ in terms of the incomplete
beta function.
\rphiclm*

\begin{proof}
    Recall $r(\phi)= \abs{F_{\theta}(e^{i \phi}) - F_{\theta}(1)}$. Furthermore, from~\Cref{lem:conformal-properties},
    we know that the conformal map maps the arc from $1$ to $i$ to an edge on the rhombus $\mathbb{S}$.
    Hence we can write $r(\phi)$ as an integral
    \begin{align*}
        r(\phi) &= \int_{0}^{\phi} |F_{\theta}'(e^{i \phi})| d\phi  = \int_{0}^{\phi} |f_{\theta}(e^{i \phi})| d\phi.  \\
                \intertext{Expanding $f_{\theta}$, and substituting
                  $a=\frac{\theta}{\pi}$ and $b=1-a$, we have}
                &= \int_{0}^{\phi} |(1-e^{2i\phi} )^{a-1}\cdot (1+e^{2i\phi})^{b-1}| d\phi.  \\
                \intertext{Expanding this in terms of  trigonometric
                  functions, and simplifying using double angle
                  formulas, we get}
                &= \int_{0}^{\phi} |(2 \cdot e^{i\phi}\cdot \sin\phi )^{a-1}\cdot (2\cdot e^{-i\phi}\cdot \cos \phi)^{b-1}| d\phi  \\
                &= \int_{0}^{\phi} |2^{a+b-2}|\cdot  |e^{i\phi(a-1)}| \cdot |\sin\phi^{a-1}|
                            \cdot |e^{-i\phi(b-1)}| \cdot| \cos \phi^{b-1}| d\phi.  \\
                \intertext{Since $|e^{i\phi(a-1)}| = |e^{i\phi(b-1)}| = 1$  and the remaining terms are positive, we drop the norms. }
                &= \int_0^\phi \frac{1}{2} (\sin \phi)^{a-1} (\cos  \phi)^{b-1} d\phi \\
                &= \frac{1}{4} \beta_{\sin^2 \phi} (a/2,b/2) && \text{by~\Cref{l:incomp_trig}.}
    \end{align*}
\end{proof}

By substituting $\phi=\pi/2$ we immediately get the following corollary:
\begin{corollary}
    \label{l:side-length} The length of the side of  rhombus is given by $r = r(\pi/2) = \nicefrac[]{1}{4}\cdot \beta(a/2,b/2)$.
\end{corollary}

The claim below will characterize the integral of  the incomplete beta function which will be important for us later.

\rphiintclm*

\begin{proof}
  By~\Cref{l:incomp_trig}, the left hand side equals
\begin{align*}
    & \int_0^{\pi/2}  \beta_{\sin^2 \phi}\left(\frac{a}{2},\frac{b}{2}\right) d\phi  \\
            &= \int_0^{\pi/2}\frac{ 2\left(\sin^2 \phi\right)^{a/2}}{a}
            \pFq{2}{1}{\frac{a}{2},1-\frac{b}{2}}{\frac{a}{2}+1}{\sin^2\phi } d\phi
            &&\text{By~\Cref{l:inc_beta_hyper}} \\
            &= \int_0^{\pi/2}\frac{ 2 (\sin \phi)^{a}}{a}
           \pFq{2}{1}{\frac{a}{2},\frac{a+1}{2}}{\frac{a}{2}+1}{\sin^2\phi } d\phi
            && \text{Substituting $b=1-a$} \\
     &= \frac{2}{a} \int_0^{\pi/2} \Big(\sum_{n=0}^\infty \frac {(a/2)_n (a/2+1/2)_n}{(a/2+1)_n}
     \frac{(\sin\phi)^{2n+a}}{n!} \Big) d\phi
             && \text{ Expand using~\Cref{def:hypergeom}} \\
             &= \frac{2}{a} \sum_{n=0}^\infty   \Big( \int_0^{\pi/2} (\sin\phi)^{2n+a} d\phi \Big) \frac{(a/2)_n (a/2+1/2)_n}{(a/2+1)_n \cdot n!}  \tag{*} \label{eq:ga}
\end{align*}

We take a brief interlude to analyze the integral in the parenthesis
above:
\begin{align*}
    \int_0^{\pi/2}\left( \sin \phi \right)^{2n+a} d \phi &= \frac{1}{2} \beta(n+a/2+1/2,1/2) && \text{By~\Cref{l:trig_beta}} \\
                                      &= \frac{\Gamma(1/2)}{2}  \frac{\Gamma(n+a/2+1/2)}{\Gamma(n+a/2+1)}  &&\text{By~\Cref{lem:betagamma}} \\
                                      &= \frac{\Gamma(1/2)}{2} \frac{ (a+1/2)_n  \Gamma(a/2+1/2)}{(a/2+1)_n \Gamma(a+1)}\\
                                      &= \frac{\beta(a/2+1/2,1/2)}{2} \frac{ (a/2+1/2)_n }{(a/2+1)_n}.
\end{align*}
Going back and substituting the above result into the~\Cref{eq:ga}, we get
\begin{align*}
  &= \frac{\beta(a/2+1/2,1/2)}{a}  \left(\sum_{n=0}^\infty \frac {(a/2)_n (a/2+1/2)_n (a/2+1/2)_n}{n!(a/2+1)_n (a+1)_n} \right) \\
  &= \frac{\beta(a/2+1/2,1/2)}{a}  \cdot  \pFq{3}{2}{\frac{1+a}{2},\frac{1+a}{2},\frac{a}{2}}{\frac{a}{2}+1,\frac{a}{2}+1}{1}\\
\end{align*}
\end{proof}

Armed with Claim~\ref{lem:rphiint} and~\Cref{l:side-length}, we can prove~\Cref{thm:cut-prob}.
\cutprobthm*
\begin{proof}
    Substituting $r = r(\pi/2)$ below, by~\Cref{lem:prob-cut} we have
    that the probability of separating the vertices is
\[ \frac{2}{\pi} \int_{\phi=0}^{\pi/2} 1-\frac{r(\phi)}{r} d\phi, \]
or equivalently, the probability of \emph{not} separating them is
 \[\frac{2}{\pi} \int_{\phi=0}^{\pi/2} \frac{r(\phi)}{r} d\phi  \]

 \noindent Expanding the above using Claim~\ref{l:rphi}, we get that
 this equals
 \begin{align*}
     & \frac{2}{\pi} \frac{4}{\beta(a/2,b/2)}  \int_{\phi=0}^{\pi/2} \frac14 \beta_{\sin^2 \psi}(a/2,(1-a)/2) d\psi.  \\
     \intertext{Expanding the right hand side integral using Claim~\ref{lem:rphiint}, we get}
     &= \frac{2}{\pi} \frac{1}{\beta(a/2,b/2)} \cdot \frac{ \beta(a/2+1/2,1/2)}{a} \cdot {}_3F_2(1/2+a/2,1/2+a/2,a/2;a/2+1,a+1;1).
     \intertext{Using Lemma~\ref{lem:betagamma} and the fact that
       $\Gamma(1/2)^2 = \pi$ we can simplify this to}
     &= \frac{\Gamma(a+1/2)}{\Gamma(1/2-a/2)\Gamma(a/2+1)^2} \cdot {}_3F_2(1/2+a,1/2+a/2,a/2;a/2+1,a+1;1).
 \end{align*}
 \end{proof}

 \subsection{Proof of\texorpdfstring{~\Cref{th:pi}}{Theorem 5}}
First, we rewrite $r$ in a form that will be useful later. 
\begin{claim}
\label{l:ga-alt-form}
\[
2 \int_0^{\pi/2}  r(\phi) d\phi = \int_0^{\pi/2} \phi (\sin \phi)^{b-1}(\cos \phi)^{a-1} d\phi \]
\end{claim}
\begin{proof}
    The left hand side equation  can be written as
\[ 2 \int_0^{\pi/2}  r(\phi) d\phi = \int_0^{\pi/2} \frac{1}{2} \beta_{\sin^2{\phi}}(a/2,b/2) d\phi
    = \int_0^{\pi/2} \left( \int_0^{\phi}  (\sin \psi)^{a-1}(\cos\psi)^{b-1} d\psi \right)  d\phi \]
Applying integration by parts: $\int p dq = [pq] - \int q dp $ with $q=\pi/2-\phi$ and $p = \int_0^{\phi} (\sin \psi^{a-1})(\cos\psi)^{b-1} d\psi$ gives

 \[ \Big[(\pi/2 - \phi) \int_0^{\phi} (\sin \psi^{a-1})(\cos\psi)^{b-1} d\psi \Big]_0^{\pi/2} + \int_0^{\pi/2} (\pi/2-\phi) \frac{d}{d\phi} \int_0^{\phi} (\sin \psi^{a-1})(\cos\psi)^{b-1} d\psi \]

The first term is $0$, and using Fundamental Theorem of Calculus, the second term is
\[ \int_0^{\pi/2} (\pi/2-\phi) (\sin \phi)^{a-1}(\cos \phi)^{b-1} d\phi \]
Substituting $\phi$ for $\pi/2-\phi$ gives
    \[\int_0^{\pi/2} \phi (\sin \phi)^{b-1}(\cos \phi)^{a-1} d\phi. \]
\end{proof}

Next, we claim that
\begin{claim}
    \label{lem:epspinumeric}
    When $\theta = (1-\epsilon)\pi$, we can say
    \[
    2 \int_{0}^{\pi/2} r(\phi) d\phi \leq   2 \cdot \Big(1+O(\epsilon \log (1/\epsilon)) \Big)
    .\]
\end{claim}

\begin{proof}

    Using Claim~\ref{l:ga-alt-form}, we can write
    \begin{align*}
    2 \int_{0}^{\pi/2}  r(\phi) d\phi  &= \int_0^{\pi/2} \phi (\sin \phi)^{b-1}(\cos \phi)^{a-1} d\phi \\
    &= \int_0^{\pi/2} \frac{ \phi }{\sin  \phi}  (\tan \phi)^{\epsilon} d\phi
    \end{align*}

    Since $\frac{x}{\sin(x)} \leq 2$ for $0 \le x \le
    \nicefrac[]{\pi}{2}$, to prove the claim it suffices to show that
\[ \int_0^{\pi/2} \left( (\tan \phi)^\epsilon -1 \right) d\phi = O(\epsilon \log (1/\epsilon)). \]
Let $\phi_0 = \arctan(1/\epsilon)$. We will break the above integral into two parts and deal with each separately:

\[ \int_0^{\pi/2} \left( (\tan \phi)^\epsilon -1 \right) d\phi =  \int_0^{\phi_0} \left( (\tan \phi)^\epsilon -1 \right)  +
\int_{\phi_0}^{\pi/2} \left( (\tan \phi)^\epsilon -1 \right).  \]

{\bf Case 1} for $\phi \leq \phi_0$,
\[ (\tan \phi)^\epsilon \leq  \left( \frac{1}{\epsilon} \right)^{\epsilon} = \exp(\epsilon \ln(1/\epsilon)) =   1 + O(\epsilon \log (1/\epsilon) ),\]
so  \[ \int_0^{\phi_0} \left((\tan \phi)^\epsilon -1  \right)d\phi = O(\epsilon \log (1/\epsilon) ). \]

{\bf Case 2}
For $\phi > \phi_0$,

\begin{align*}
    \int_{\phi_0}^{\pi/2} \left((\tan \phi)^\epsilon -1\right) d\phi
    &\leq \int_{\phi_0}^{\pi/2}
        1/(\cos \phi)^\epsilon  d\phi \\
    &= \int_0^{\pi/2 - \phi_0} (1/\sin \phi)^\epsilon d \phi && \text{ Since} \sin(x) = \cos(\pi/2-x) \\
    &\leq  \int_0^{\pi/2 -\phi_0} (2/ \phi)^\epsilon d \phi && \text{Since }1\leq \nicefrac[]{x}{\sin(x)} \leq 2 \\
    &\leq  2^{\epsilon}  \frac{(\pi/2-\phi_0)^{1-\epsilon}}{1-\epsilon}  \\
    &\leq (\pi/2-\phi_0) (1+O(\epsilon)).
\end{align*}
Finally, we note that $\pi/2-\phi_0 \leq \tan(\pi/2-\phi_0) = 1/\tan(\phi_0) = \epsilon$.
\end{proof}

\piepsthm*
\begin{proof}
Let $a = 1-\epsilon$ and $b=\epsilon$. 

As discussed in~\Cref{lem:prob-cut}, the non-separation probability is
\[ \frac{2}{\pi r} \int_{0}^{\pi/2} r(\phi) d\phi  \] where $r = r(\pi/2)$.
So we will compute the asymptotics of $r:=r(\pi/2)$ and $ \int_{0}^{\pi/2} 2 \cdot r(\phi) d\phi $ as $\eps \to 0$.

First we compute the asymptotics of $r$ as $\eps \to 0$.  Recall that
\begin{align*}
    r &= (1/4) \beta (a/2,b/2)   && \text{By~\Cref{l:side-length}} \\
      &= \frac{ \Gamma((1-\epsilon)/2)\Gamma(\epsilon/2)}{4\Gamma(1/2)} &&\text{By~\Cref{lem:betagamma}} \\
      &= \frac{ \Gamma((1-\epsilon)/2)\Gamma(1+\epsilon/2)}{2 \epsilon \Gamma(1/2)} &&
            \text{ Using }\Gamma(\frac{\epsilon}{2}) = \frac{2}{\epsilon} \Gamma(1+ \frac{\epsilon}{2}) \\
    \intertext{Using the  standard fact that $\Gamma(z+\epsilon) = \Gamma(z)(1+O(\epsilon))$ for fixed $z >  0 $}
      &= \frac{1}{2 \epsilon \Gamma(1/2)}  (\Gamma(1/2) + O(\epsilon))(\Gamma(1) + O(\epsilon))  \\
      &= \frac{1}{(2\epsilon)} + O(1)
\end{align*}
 which implies that
 \[ 1/r = 2\epsilon + O(\epsilon)^2 \]

Using Claim~\ref{lem:epspinumeric}, we know that $ \int_{0}^{\pi/2} 2 \cdot r(\phi) d\phi $
is at most  $2\cdot \Big(1+O(\epsilon \log (1/\epsilon) \Big) $.
Combining the two, we get the probability of non-separation is
$  \epsilon \frac{4}{\pi}  + O(\epsilon^2) \approx 1.27 \epsilon +O(\epsilon^2)$.
\end{proof}

\subsection{Other Missing Proofs}

\basictransform*
\begin{proof}
Part 1 is immediate from the continuity and linearity of the map $\mathbf{O}\cdot \bW^{-\nicefrac[]{1}{2}}$.

To prove part 2, observe that the $W^{\nicefrac[]{1}{2}}$ is given explicitly
by the matrix
\[
    \bW^{\frac{1}{2}} = \frac{1}{\sqrt{2}} \cdot
    \begin{bmatrix}
            \cos(\frac{\theta}{2}) + \sin(\frac{\theta}{2}) & \cos(\frac{\theta}{2}) - \sin(\frac{\theta}{2}) \\
            \cos(\frac{\theta}{2}) - \sin(\frac{\theta}{2}) & \cos(\frac{\theta}{2}) + \sin(\frac{\theta}{2})
    \end{bmatrix}.
\]
Taking, its inverse, we get the matrix
\[
    \bW^{-\frac{1}{2}} = \frac{1}{\sqrt{8}} \cdot
    \begin{bmatrix}
            \sec(\frac{\theta}{2}) + \csc(\frac{\theta}{2}) & \sec(\frac{\theta}{2}) - \csc(\frac{\theta}{2}) \\
            \sec(\frac{\theta}{2}) - \csc(\frac{\theta}{2}) & \sec(\frac{\theta}{2}) + \csc(\frac{\theta}{2})
    \end{bmatrix}.
\]
Since $\bW^{-\frac12}[-1,1]^2$ is the image of a parallelogram, it
must also be a parallelogram. Moreover, one can directly check that
the diagonals are orthogonal to each other, so it must be a rhombus.
It is easy to calculate the angle between the sides and see that it is
exactly $\theta$ at the image of $(1, -1)$ and $\pi-\theta$ at the image
of $(1,1)$.

Then part 3 follows from the previous parts: if $\bX_t$ is one a
side or a vertex of $[-1,1]$, then $\bY_t$ is on the corresponding
side or vertex of $\rhombus$.
\end{proof}

\section{Omitted Proofs from Section~\ref{sec:pde}}

Recall that, for $0\leq x\leq 1, 0\leq y\leq 1$, the function $u(x,y)$
denotes the probability the probability of a clause being satisfied
when the random walk walk begins with marginals $(x,y)$ and angle
$\theta$. Equivalently, $u(x,y)$ is the probability that the walk,
started at $(x,y)$, ends at one of the corners $(0,0)$, $(0,1)$ or
$(1,0)$. 

\boundary*
\begin{proof}

Recall that $\tau$ is the first time when $\bX_t$ hits the boundary of
$[0,1]^2$, and $\sigma$ is the first time when $\bX_t$ hits a vertex of
$[0,1]^2$. The function $\phi$ evaluates to $1$ at the vertices
$(0,0)$, $(0,1)$, and $(1,0)$, and to $0$ at $(1,1)$. Therefore, by
definition, $u(\bx) = \E^\bx[\phi(\bX_\sigma)]$. 

Let us first consider the case when $\bx$ is on the boundary of
$[0,1]^2$. Then one of the coordinates of $\bX_t$ remains fixed for
the entire process. Since $\phi$ is affine in each of its arguments,
and $\bX_t$ is a martingale, 
\[
\forall \bx \in \bd [0,1]^2: \ u(\bx) = \E^\bx[\phi(\bX_\sigma)]
= \phi(\E^\bx[\bX_\sigma])
= \phi(\bx). 
\]

When $\bx$ is in the interior of $[0,1]^2$, we have, by the law of
total expectation, 
\[
\forall \bx \in \Int [0,1]^2:\ 
u(\bx) = \E^\bx[\phi(\bX_\sigma)]
= \E^\bx[\E^{\bX_\tau}[\phi(\bX_\sigma)] ] 
= \E^\bx[\phi(\bX_\tau)].
\]
The final equality follows by the special case when the starting point
of the random walk is on the boundary of $[0,1]^2$. This proves the
lemma. 
\end{proof}

Before we give an outline of the proofs of  Lemma~\ref{lem:feasible_maximum} and Lemma~\ref{lem:ratio_maximum}, we prove Lemma~\ref{lem:feasibility}. This will allow us to give a completely analytic proof of $\frac34$.

\feasibility*
\begin{proof}
Clearly, the first set of the constraints are obvious. We focus on the second and the third constraint.
Recall that $\bv_i=x \bv_0 + \sqrt{x-x^2} \bw_i$ and $\bv_j= y\bv_0 + \sqrt{y-y^2} \bw_j$ where $\bw_i$ and $\bw_j$ are unit vectors orthogonal to $\bv_0$ with $\cos(\theta)=\bw_j\cdot \bw_j$. Thus we have
\begin{align*}
\bv_i\cdot \bv_j=xy+\cos(\theta)\sqrt{x-x^2}\sqrt{y-y^2}\\
\end{align*}

But then we have the following valid constraint from the SDP:
\begin{align*}
\bv_i\cdot \bv_j \geq 0
\end{align*}
which implies that
\begin{align*}
\cos(\theta)\geq -\sqrt{\frac{xy}{(1-x)(1-y)}}
\end{align*}
proving the second inequality.

For the other inequality, observe that we have $\bv_{-i}=(1-x)\bv_0- \sqrt{x-x^2} \bw_i$ and $\bv_{-j}= (1-y)\bv_0 - \sqrt{y-y^2} \bw_j$. Then we have
\begin{align*}
\bv_{-i}\cdot \bv_{-j}=(1-x)(1-y)+\cos(\theta)\sqrt{x-x^2}\sqrt{y-y^2}\\
\end{align*}

But then we have the following valid constraint from the SDP:
\begin{align*}
\bv_{-i}\cdot \bv_{-j} \geq 0
\end{align*}
which implies that
\begin{align*}
\cos(\theta)\geq -\sqrt{\frac{(1-x)(1-y)}{xy}}
\end{align*}
proving the third inequality.
\end{proof}

To ease the remainder of the presentation we first prove that the Brownian rounding algorithm achieves an approximation of $\nicefrac[]{3}{4}$ for \MtwoSAT via the maximum principle.
In order to achieve that we use the following two functions for different ranges of $\theta$.
\begin{itemize}
    \item $g_1 (x,y) = 1-xy-\cos(\theta) \sqrt{x-x^2} \sqrt{y-y^2}$.
    \item $f(x,y) = 1-xy$.
\end{itemize}
First consider the case when $0\leq \theta\leq\frac{\pi}{2}$.   In this case, we show $g_1$ satisfies the requirement of the Corollary~\ref{cor:maximum} as well as give an approximation factor of $1$. The last fact is trivially true since $g_1$ is exactly the SDP objective.

For conditions of the Corollary~\ref{cor:maximum}, we need to show that
\begin{align*}
 \frac{\partial ^2 g_1}{\partial x^2} + \frac{\partial ^2 g_1}{\partial y^2} + 2\cos (\theta) \frac{\partial ^2 g_1}{\partial x \partial y} &\geq  0  &\forall (x,y)\in \Int[0,1]^2\\
 g_1(x,y) &\leq (1-xy) &\forall (x,y)\in \partial[0,1]^2\\
 \end{align*}
Since $(x-x^2)(y-y^2)=0$ on $\bd[0,1]^2$, we obtain that $g_1(x,y)=1-xy$ on $\bd[0,1]^2$ as required.
It remains to show that
$$\frac{\partial^2}{\partial x^2} g_1(x,y) + \frac{\partial^2}{\partial y^2} g_1(x,y) +2 \cos{\theta}\frac{\partial^2}{\partial x \partial y} g_1(x,y) \geq 0$$
for all $ (x,y)\in (0,1)^2$.
Consider $$h(x,y):=\frac{\partial^2}{\partial x^2} g_1(x,y) + \frac{\partial^2}{\partial y^2} g_1(x,y) +2 \cos{\theta}\frac{\partial^2}{\partial x \partial y} g_1(x,y).$$
To show $h$ is non-negative, we do the following change of variables in $x=\frac{(1+\sin(a))}{2}$ and $y=\frac{(1+\sin(b))}{2}$ for some $|a|,|b|\leq \frac{\pi}{2}$.
Such $a$ and $b$ exist since $0\leq x,y\leq 1$. Now simplifying, we obtain:
\begin{eqnarray*}
 && h\left(\frac{1+\sin(a)}{2},\frac{1+\sin(b)}{2}\right)\\
   &=&2 \cos(\theta)\sec^3(a)\sec^3(b)\Big[\left(cos^2(a)-\cos^2(b)\right)^2    \\
   &&  + 2\cos^2(a)\cos^2(b)\left(1-\cos(\theta)\sin(a)\sin(b)-\cos(a)\cos(b) \right)\Big]
\end{eqnarray*}
Since $|a|,|b|\leq \frac{\pi}{2}$ and $0\leq \theta\leq \frac{\pi}{2}$, we have that $\sec(a),\sec(b),\cos(\theta)\geq 0$.
Thus, it enough  to show that
$$1-\cos(\theta)\sin(a)\sin(b)-\cos(a)\cos(b) \geq 0.$$
Since the above expression is linear in $\cos(\theta)$, it is enough to check for extreme values of $\cos(\theta)$ which takes value between $0$ and $1$.
It is clearly true when $\cos(\theta)=0$. For $\cos(\theta)=1$, it equals $1-\cos(a-b)$ and is thus non-negative.

Now consider $-1\leq \cos(\theta)\leq 0$. We show that $f(x,y)=1-xy$ satisfies the condition of Corollary~\ref{cor:maximum} and is at least $\frac34$ the value of the SDP objective for all feasible $(x,y,\theta)$. First let us focus on the condition of Corollary~\ref{cor:maximum}. Clearly, the boundary conditions are satisfied by construction. Note that $\frac{\partial^2 f(x,y)}{\partial x^2} =0 $, $\frac{\partial^2 f(x,y)}{\partial y^2}=0$, and that $ \frac{\partial^2 f(x,y)}{\partial x \partial y} =1$.
Thus

$$\L f= -\cos(\theta)\geq 0$$
since $\cos(\theta)\leq 0$ as desired.

It remains to show that $f$ provides an approximation guarantee of $\nicefrac[]{3}{4}$ in case $\cos(\theta)< 0$.
Recall that $SDP(x,y,\theta)=1-xy-\cos(\theta)\sqrt{x-x^2}\sqrt{y-y^2}$ is the contribution of a clause to the SDP's objective whose two variables $z_i$ and $z_j$ have marginal values of $x$ and $y$ respectively and that $\cos(\theta)=\bw _i \cdot \bw_j$.
We prove the following claim which would imply that we obtain a $\frac34$-approximation.

\begin{claim}
For any $x,y,\theta$ that satisfy the feasibility conditions in Lemma~\ref{lem:feasibility} and $\cos(\theta)<0$, we have
$$g(x,y)\geq \frac{3}{4} SDP(x,y,\theta).$$
\end{claim}
\begin{proof}
From Lemma~\ref{lem:feasibility}, we have $$-\cos(\theta)\leq \min\left\{\sqrt{\frac{xy}{(1-x)(1-y)}},\sqrt{\frac{(1-x)(1-y)}{xy}}\right\}.$$

Observe that we have $g(x,y)\geq \frac{3}{4} SDP(x,y,\theta)$, if
$$(1-xy)\geq -3 \cos(\theta) \sqrt{(x-x^2)(y-y^2)}.$$

First, suppose $xy\leq \frac{1}{4}$. Then

\begin{eqnarray*}
  -3 \cos(\theta) \sqrt{(x-x^2)(y-y^2)} &\leq & 3 \sqrt{\frac{xy}{(1-x)(1-y)}}\cdot \sqrt{(x-x^2)(y-y^2)} =3 xy \\
   &\leq & 1-xy
\end{eqnarray*}

Else, if $xy\geq \frac14$, then we have

\begin{eqnarray*}
  1-xy+3 \cos(\theta) \sqrt{(x-x^2)(y-y^2)} &\geq & 1-xy -3 \sqrt{\frac{(1-x)(1-y)}{xy}}\cdot \sqrt{(x-x^2)(y-y^2)}\\
  & =& -2+3x+3y -4xy
\end{eqnarray*}
 Over all $1\geq x\geq 0,1\geq y\geq 0$ with fixed $xy$, the quantity  $2+3x+3y -4xy$ is minimized when $x=y$. Since $xy\geq \frac14$, we must have $x\geq \frac12$. But then it becomes $-2(1-3x +2x^2)=-2(1-2x)(1-x)\geq 0$ since $\frac12\leq x\leq 1$.
 This proves the $\frac34$-approximation.
\end{proof}

We now give a brief outline of the proof of Lemma~\ref{lem:feasible_maximum} and Lemma~\ref{lem:ratio_maximum}. The complete proofs involve long sum of square expression that are available at \cite{code-repo}.
\feasiblemax*
\begin{proof}
$ $\\
\textbf{Feasibility of $g_1(x,y)$.} We already showed in the above proof of $\frac34$-approximation.

\noindent \textbf{Feasibility of $g_2(x,y)$.} Now we consider $g_2(x,y)=1-xy-2\cos(\theta)(x-x^2)(y-y^2)$. Since $(x-x^2)(y-y^2)=0$ on $\bd[0,1]^2$, we obtain that $g_2(x,y)=1-xy$ on $\bd[0,1]^2$ as required. It remains to show that

$$\L g_2=\frac{\partial^2}{\partial x^2} g_2(x,y) + \frac{\partial^2}{\partial y^2} g_2(x,y) +2 \cos{\theta}\frac{\partial^2}{\partial x \partial y} g_2(x,y) \geq 0$$
for all $ (x,y)\in (0,1)^2$ for any $0\geq \cos(\theta) \geq -\frac12$.
A simple calculation allows us to obtain that
$$\L g_2=-2\cos(\theta)\left(1+2x^2+2y^2+2\cos(\theta)-2x-2y-4y\cos(\theta)-4x\cos(\theta)+8xy\cos(\theta)\right).$$

Since $-2\cos(\theta)>0$, it is enough to show that for any $0\leq x\leq 1$ and $0\leq y\leq 1$,
$$h(x,y)=1+2x^2+2y^2+2\cos(\theta)-2x-2y-4y\cos(\theta)-4x\cos(\theta)+8xy\cos(\theta)\geq 0.$$

We now prove the above inequality. Since the above expression is linear in $\cos(\theta)$, for any fixed $x,y$ the minimum appears at either $\cos(\theta)=0$ or $\cos(\theta)=-\frac12$. First consider $\cos(\theta)=0$. In this case, we obtain
$$h(x,y)=1+2x^2+2y^2-2x-2y=\frac12\left(1-2x\right)^2 +\frac12 \left(1-2y\right)^2\geq 0$$
as required.

Now if $\cos(\theta)=-\frac12$, we obtain
$$h(x,y)=2x^2+2y^2-4xy=2 (x-y)^2\geq 0$$
as required. This proves $\L g_2\geq 0$.

\noindent \textbf{Feasibility of $g_3(x,y)$.} Now we consider $g_3(x,y)=1-xy-\frac12(1+5 \cos(\theta))(x-x^2)(y-y^2)(x+y)(2-x-y)$ on $\bd[0,1]^2$, we obtain that $g_2(x,y)=1-xy$ on $\bd[0,1]^2$ as required. It remains to show that

$$\L g_3=\frac{\partial^2}{\partial x^2} g_3(x,y) + \frac{\partial^2}{\partial y^2} g_3(x,y) +2 \cos{\theta}\frac{\partial^2}{\partial x \partial y} g_3(x,y) \geq 0$$
for all $ (x,y)\in (0,1)^2$ for any $-\frac12\geq \cos(\theta) \geq -1$.

To show $\L g_3\geq 0$, we consider $\L g_3=p(x,y,\cos(\theta))$ as a polynomial in $x,y$ and $\cos(\theta)$. Replacing $z=\cos (\theta)$, our aim is to show $p(x,y,z)\geq 0$ if $0\leq x,y\leq 1$ and $-\frac12\leq z\leq -1$. Equivalently, we need to show $p(x,y,z)\geq 0$ whenever $r_1(x,y,z):=x-x^2\geq 0$, $r_2(x,y,z):=y-y^2\geq 0$ and $r_3(x,y,z):=-(z+\frac12) \geq 0 $ and $r_4(x,y,z):=(z+1)\geq 0$. This we show by obtaining polynomials $q_i(x,y,z)$ for $i=0,1,2,3,4$ such that $q_i$ is a sum of square polynomial of fixed degree and we have
 $$p(x,y,z)=q_0(x,y,z)+\sum_{i=1}^4 q_i(x,y,z)r_i(x,y,z).$$
 Observe that above polynomial inequality shows the desired inequality.  Indeed evaluate the above identity for any $0\leq x,y\leq 1$ and $-\frac12\geq z\geq -1$. Clearly, the RHS is non-negative. Each $q_i$ is non-negative since it is a SoS and each $r_i$ is non-negative by construction. We mention that we obtain these proofs via solving a semi-definite program of fixed degree (4) for each of $q_i's$. We also remark that these SoS expressions are obtained with a small error of order $\delta<10^{-5}$. This, formally, implies that the approximation factors of slightly worse than $\frac78$.
\end{proof}

\ratiomax*
\begin{proof}
We prove the three inequalities.  We also remark that the SoS expressions below are obtained with a small error of order $\delta<10^{-5}$. This, formally, implies that the approximation factors of slightly worse than $\frac78$.
\begin{enumerate}
\item If $1\geq \cos(\theta) \geq 0$, then $g_1(x,y)\geq 1\cdot SDP(x,y,\theta)$. Observe that $g_1(x,y)=SDP(x,y,\theta)$ and inequality holds.
\item If $0\geq \cos(\theta) \geq -\frac12$, then $g_2(x,y)\geq 7/8 \cdot SDP(x,y,\theta)$.
We need to show that
\begin{align*}
1-xy-2\cos(\theta)(x-x^2)(y-y^2)\geq \MtwoRatioText \cdot \left(1-xy-\cos(\theta)\sqrt{x-x^2}\sqrt{y-y^2}\right)
\end{align*}
which holds if

\begin{align*}
1-xy-16\cos(\theta)(x-x^2)(y-y^2)\geq -7\cos(\theta)\sqrt{x-x^2}\sqrt{y-y^2}
\end{align*}

Since both sides are non-negative ($1-xy\geq 0$ and $\cos(\theta)\leq 0$), it is enough to show
\begin{align*}
\left(1-xy-16\cos(\theta)(x-x^2)(y-y^2)\right)^2 -49\cos^2(\theta)(x-x^2)(y-y^2)\geq 0
\end{align*}

subject to $r_1(x,y,\cos (\theta)):=x-x^2\geq 0$, $r_2(x,y,\cos (\theta)):=y-y^2\geq 0$, $r_3(x,y,\cos (\theta)):=-\cos(\theta)\geq 0$, ,$r_4(x,y,\cos (\theta)):=xy-(1-x)(1-y)\cos^2(\theta)\geq 0$,
$r_5(x,y,\cos (\theta)):=(1-x)(1-y)-xy\cos^2(\theta)\geq 0$ where the last two constraints follow from Lemma~\ref{lem:feasibility}. Thus again, we construct sum of squares polynomials $q_{i}(x,y,\cos(\theta))$ for $0\leq i\leq 5$ such that

\begin{align*}
&\left(1-xy-16\cos(\theta)(x-x^2)(y-y^2)\right)^2 -49\cos^2(\theta)(x-x^2)(y-y^2)\\
&= q_0(x,y,\cos(\theta))+\sum_{i=1}^5 q_i(x,y,\cos(\theta)) r_i(x,y,\cos(\theta))
\end{align*}

\item If $-\frac12\geq \cos(\theta) \geq -1$, then $g_3(x,y)\geq \frac{7}{8} SDP(x,y,\theta)$. The similar argument as above allows us to obtain SoS proofs. We omit the details.
\end{enumerate}

\end{proof} 

\section{Omitted Proofs from Section~\ref{sect:mc-constraints}}

\paragraph*{Baseline Approximation}

\simpleapprox*
\begin{proof}
  If the constraints specified by $\FF$ and $\bb$ are satisfiable, then
  surely the following linear program also has a solution.
  \begin{align*}
      &\sum_{j \in F_i}{x_j} = b_i &\forall i = 1, \ldots, k\\
      &0 \le x_j \le 1 &\forall j = 1, \ldots, k
  \end{align*}
  We compute a solution $\bx \in \Re^n$ to the program, and form
  a vector $\by\in \Re^n$ by defining $y_j = (1-\eps)x_j +
  \frac{\eps}{2}$ for all $j \in [n]$. The vector $\by$ still
  satisfies the constraints approximately, i.e.~for all $i \in [k]$ we have
  \begin{equation}
    \label{eq:approx-linsys-soln}
    \left|\sum_{j \in F_i}{x_j} - b_i\right| \le \frac{\eps n}{2}.
  \end{equation}
  We now apply standard randomized rounding to $\by$: we form a set $S$
  by independently including any $j \in [n]$ in $S$ with probability
  $y_j$. By \eqref{eq:approx-linsys-soln}, and a Hoeffding and a union bound, 
  \[
  \Pr\Biggl[\exists i: \Bigl|\sum_{j \in F_i}{x_j} - b_i\Bigr| > \eps  n\Biggr]
  \le 
  2ke^{-\eps^2 n /2}.
  \]
  By the assumption we made on $n$, the right hand side is at most
  $\frac{\eps}{4}$. 

  Next we analyze the weight of the cut edges $a(\delta(S))$. Any edge
  $e = (i,j)$ has probability 
  \[
  y_i + y_j - y_i y_j \ge 2\eps(1-\eps) \ge \eps.
  \]
  to be cut. Therefore, $\E[a(\delta(S))] \ge \eps a(E)$. By Markov's
  inequality applied to $a(E) - a(\delta(S))$, 
  \[
  \Pr\Bigl[a(\delta(S)) < \frac{\eps}{2}\Bigr] 
  < \frac{1-\eps}{1 - \frac{\eps}{2}}
  \le 1 - \frac{\eps}{2}.
  \]
  Therefore, the probability that $S$ satisfies every constraint up to
  an additive error of $\eps n$, and $a(\delta(S)) \ge \frac{\eps}{2}
  a(E)$ is at least $\frac{\eps}{4}$. We get the high probability
  guarantee by repeating the entire rounding procedure a sufficient
  number of times. 
\end{proof}

\paragraph*{Approximation Ratio Analysis}

\mcall*
\begin{proof}
  Let us denote by $\theta_{ij}$ the angle between the unit vectors
  $\bw_i$ and $\bw_j$, i.e.~$\theta_{ij} =
  \arccos(\bkt{{\bw}_i}{{\bw}_j})$. Recall that, for any $i$,
  $\|\bv_i\|^2 = x_i$, and $\bv_i = x_i \bv_0 + \sqrt{x_i^2 - x_i}
  \bw_i$, where $\bv_0$ and $\bw_i$ are orthogonal to each
  other. Therefore, for any pair $\{i,j\}$, $\|\bv_i - \bv_j\|$ is
  characterized entirely by the triple $(x_i,x_j,\theta)$, and is
  equal to
  \begin{equation}
    \label{eq:mcsc-obj}
    \norm{\bv_i - \bv_j}^2 = x_i + x_j -2\cos(\theta_{ij})
    \sqrt{x_i(1-x_i)x_j(1-x_j)}.
  \end{equation}
  We will refer to triples $(x,y,\theta)$ as configurations, and will
  denote the expression on the right hand side of~\eqref{eq:mcsc-obj}
  with $x_i = x$, $x_j = y$, and $\theta_{ij} = \theta$ by
  $\SDP(x,y,\theta)$.
    
  To calculate $\Pr[(\bX_{\tau_n})_i \neq (\bX_{\tau_n})_j)] $, we use
  the techniques introduced in~\cref{sec:pde}. More concretely, let
  \[ 
  u_{\theta}(x,y) = \Pr \big[(\bX_{\tau_n})_i \neq (\bX_{\tau_n})_j) \mid
  ((\bX_{0})_i,(\bX_{0})_j)=(x,y) \big]. 
  \] 
  As shown in ~\Cref{sec:pde}, the function $u_{\theta}$ is the unique
  solution to the partial differential equations
  \begin{align}
    \frac{\partial^2 u_\theta}{\partial x^2} + \frac{\partial^2 u_\theta}{\partial y^2} + 2\cos(\theta) \frac{\partial^2 u_\theta}{\partial x \partial y} &= 0 &\qquad \forall (x,y) \in \Int[0,1]^2 \label{eq:mcsc-pde}\\
    u_\theta(x,y) &=  x+y-2xy & \quad\forall (x,y) \in \bd[0,1]^2\label{eq:mcsc-pde-bd}
  \end{align}
  The above system is a Dirichlet problem and can be solved
  numerically for any configuration $(x,y,\theta)$.

  To calculate the worst case approximation ratio of the Sticky
  Brownian Motion algorithm, it suffices to evaluate
  $\min_{x,y,\theta} \frac{u_{\theta}(x,y)}{\SDP(x,y,\theta)}$.
  However, just taking a minimum over all $(x,y,\theta) \in [0,1]^2
  \times [0,\pi]$ is too pessimistic, since there are many
  configurations $(x,y,\theta)$ which never arise as solutions to
  $\SoS_{\ell}$ for any $\ell \geq 2$.  It is therefore necessary to
  consider only configurations that may arise as solutions to some
  instance. In particular, we know that any vectors $\bv_0, \bv_i,
  \bv_j$ in the SDP solution satisfy the triangle inequalities
  \eqref{mc-mult-tri1}--\eqref{mc-mult-tri3}. Translating these
  inequalities to inequalities involving $x_i, x_j$ and $\theta_{ij}$ gives
  \begin{align*}
    \cos(\theta) &\geq \max\left( 
      -\sqrt{\frac{(1-x_i)\cdot (1-x_j)}{x_i \cdot x_j}},
      -\sqrt{\frac{(x_i \cdot x_j )}{(1-x_i)(1-x_j)}}  \right)  \\
    \cos(\theta) &\leq \min \left( 
      \sqrt{\frac{x_i \cdot (1-x_j)}{ (1-x_i) \cdot x_j}},
      \sqrt{\frac{(1-x_i) \cdot x_j )}{x_i \cdot (1-x_j)}}  \right)  
  \end{align*}
  To compute the worst case approximation ratio, we use numerical
  methods.  In particular, we solve the Dirichlet problem
  \eqref{eq:mcsc-pde}--\eqref{eq:mcsc-pde-bd} for all configurations
  $(x,y,\theta)$ satisfying the inequalities above, with a granularity
  of $0.02$ in each coordinate. This numerical computation shows that
  the ratio $\frac{u_{\theta}(x,y)}{\SDP(x,y,\theta)}$ is at least
  $0.843$ for all valid configurations.
\end{proof}

\paragraph*{Hitting Time Analysis}

\hittingtime*
\begin{proof}
  We first make some observations about Brownian motion in $\Re$. 
  Let $Z_t$ be a standard one-dimensional Brownian motion started in
  $Z_0 = z \in [0,1]$, and let $\sigma = \inf\{t: Z_t \in
  \{0,1\}\}$ be the first time $Z_t$ exits the interval $[0,1]$. By
  Theorem~2.49.~in~\cite{MP10}, $\E[\sigma] = z(1-z) \le
  \frac14$. Therefore, by Markov's inequality, $\Pr[\sigma > 1] <
  \frac14$. Now observe that, by the Markov property of Brownian
  motion, for any integer $t \ge 0$ we have
  \[
  \Pr[\sigma > t]
  = 
  \prod_{r = 0}^{t-1}   \Pr\bigl[\forall s \in [r, r+1]: 0 < Z_s < 1
  \mid 0 < Z_r < 1 \bigr].
  \]
  But, conditional on $Z_r$, the process $\{Z_s\}_{s \ge r}$ is a
  Brownian motion started at $Z_r$, and, as we observed above, each of
  the conditional probabilities on the right hand side above is bounded by
  $\frac14$. Therefore, we have $\Pr[\sigma > t] < 4^{-t}$.

  To prove the lemma, we just notice that, until the first
  time $\sigma_i$ when $(\bX_t)_i$ reaches $\{0,1\}$, it is distributed like a
  one-dimensional Brownian motion started at $x_i$. This follows
  because, at any $t < \sigma_i$,  the variance per step of
  $(\bX_t)_i$ is $(\bW_t)_{i,i}  = \bW_{i,i} = 1$. Then, by
  observations above, $\Pr[\sigma_i > t] \le 4^{-t}$. 
\end{proof}

\section{Omitted Proofs from Section~\ref{sec:general}}

\paragraph*{Hitting Times Analysis}

\begin{lemma}
\label{lem:slowdown-stop}
    The expected hitting time $\E[\tau] $ for the diffusion process defined for
    Brownian Walk Algorithm with Slowdown  when the starting point $\bX_0 \in [-1+\delta, 1-\delta]^n$
    and $\alpha$ is a constant.
\end{lemma}

While the hitting time is only defined for the points away from the boundary,
this is the region where the discrete algorithm runs. Therefore, this is sufficient 
for the analysis of our algorithm. 

\begin{proof}[Proof Sketch]
    Without loss of generality, we assume the number of dimensions is $1$. In the one
    dimensional walk, the diffusion process satisfies the stochastic
    differential equation:
    \begin{equation}
        d\bX_t  = (1-\bX_t^2)^{\alpha/2}  dB_t  \label{eqn:sde-1d}
    \end{equation}
    To show this we use Dynkin's equation to compute stopping times which we present below
    specialised to the diffusion process at~\Cref{eqn:sde-1d}. 

    \noindent {\bf Dynkin's Equation} (Theorem~7.4.1 in~\cite{Oksendal-SDE})
    Let $f \in C^2\left([-1+\delta,1-\delta]\right)$. Suppose $\mu$ is a
    finite stopping time, then \[ \E^{x}[f(\bX_{\mu})] 
	= f(x) + \E^{x} \left[\int_{0}^{\mu} \left( (1-x^2)^{\alpha} \cdot \frac{\partial^2}{\partial x^2} f(\bX_{s}) \right) ds \right] \]

    Let $f(x)$ denote the function
\[ f(x) = x^2 \cdot \pFq{2}{1}{\frac12,\alpha}{\frac32}{x^2} - \frac{1-(1-x^2)^{1-\alpha}}{2(\alpha-1)} +C_1 \cdot x + C_2 \]
    where $C_1$ and $C_2$ are chosen so that $f(1-\delta) = f(-1+\delta) = 1$.
    Observe that for a fixed $\delta >0$, $f$ is well-defined and finite 
    in the domain $[-1+\delta,1+\delta]$ and satisfies $ - K_{\delta} \leq f(x) \leq K_{\delta}$ where $K_{\delta}$.  Furthermore, $f$ satisfies the
    differential equation\footnote{We verify this using Mathematica.} 
    $(1-x^2)^{\alpha} \frac{\partial^2 f}{\partial x^2} = 1$. 

	Let $\mu_j = \min(j,\tau)$  and applying Dynkin's equation we get that
	\[ \E^{x}[f(\bX_{\mu_j})] = f(x) + \E^{x}\left[\int_{0}^{\mu_j} 1 ds \right] = f(x) + \E^{x}[\mu_j] \]
Simplifying the above we get that $2K_{\delta} \geq \E^{x}[\mu_j]$ for all $j $. 
	Since we know that $\E^{x}[\tau] = \lim_{j \to \infty} \E^{x}[ \mu_j]$ almost surely,  we can bound
$2K_{\delta}\geq \E[\tau] $

    Observe that the proof does not work when $\alpha =1$. For this case, we simply change
    \[ f(x)=C_1\cdot x + C_2 + \frac{1}{2} \left[ (1+x) \log (1+x)   + (1-x) \log (1-x) \right] \]  and the argument
    goes through verbatim.
\end{proof}

\paragraph*{Remark about~\Cref{lem:dirichlet}} \label{rem:dirichlet}

The proof is largely similar to the one described in~\Cref{lem:boundary} and~\Cref{thm:oksendal-max-principle}
with two caveats:
\begin{enumerate}
     \item~\Cref{thm:oksendal-max-principle} is stated with a fixed $\mathbf{\Sigma}$. However we can handle
    the general case, where $\mathbf{\Sigma}$ is allowed to depend on the diffusion prcoess(i.e.~$\mathbf{\Sigma}(\bX_t)$), by 
    appealing to general Theorem 9.2.14 in~\cite{Oksendal-SDE}.

    \item To apply Theorem 9.2.14 from~\cite{Oksendal-SDE}, we need the
    resulting matrix  $\mathbf{\Sigma}(\bX_t) \mathbf{\Sigma}(\bX_t)^\top$to have eigenvalues bounded away
    from zero. In our case, $\mathbf{\Sigma}(\bX_t) \mathbf{\Sigma}(\bX_t)^\top$ can have zero rows and
    columns on the boundary.  To avoid this, we simply restrict our domain to be
    the hypercube scaled by a small value $[-1+\delta,1-\delta]$.  This is
    sufficient since our discrete algorithm will only run in this region. 

\end{enumerate}

\end{document}